\documentclass[11pt]{article}
\usepackage[utf8]{inputenc}

\usepackage{tkz-graph}

\usepackage{cancel}
\usepackage[normalem]{ulem}
\usepackage[T1]{fontenc} 
\usepackage{color}
\usepackage[leqno]{amsmath}
\usepackage{amsfonts, amsthm, amssymb, commath} 
\usepackage{mathtools}
\usepackage{graphicx}
\usepackage{mathrsfs}
\usepackage{caption}
\usepackage{subcaption}
\usepackage{listings}
\usepackage[]{eufrak}
\usepackage[]{hyperref}
\usepackage[]{float}
\usepackage{tikz}
\usetikzlibrary{calc, arrows,matrix,positioning,scopes, intersections}
\usetikzlibrary{shapes.geometric}
\usetikzlibrary{decorations.markings}
\usepackage{tikz-cd}
\usetikzlibrary{cd}
\usepackage{booktabs}
\usepackage{enumerate}
\usepackage{eufrak}
\usepackage{MnSymbol}
\usepackage{tikz}
\usepackage{algorithm}
\usepackage{algorithmic}
\usetikzlibrary{matrix,decorations.pathreplacing,calc}
\pgfkeys{tikz/mymatrixenv/.style={decoration=brace,every left delimiter/.style={xshift=3pt},every right delimiter/.style={xshift=-3pt}}}
\pgfkeys{tikz/mymatrix/.style={matrix of math nodes,left delimiter={(}, right delimiter={)}, inner sep=1pt,column sep=1.5em,row sep=0.5em,nodes={inner sep=0pt}}}
\pgfkeys{tikz/mymatrixbrace/.style={decorate,thick}}

\usepackage{cuted}
\newcommand{\Tr}{\mathop{\bf Tr}}
\newcommand{\ca}[1]{\mathcal{#1}}

\newcommand{\bb}[1]{\mathbb{#1}}

\newcommand{\rank}{\text{\bf rank}}

\newcommand{\vect}{\mbox{\bf vec}}

\newcommand{\argmin}{\mathop{\rm arg\,min}}

\newcommand{\dom}{\mathop{\bf dom}} 

\makeatletter
\newcommand{\leqnomode}{\tagsleft@true}
\newcommand{\reqnomode}{\tagsleft@false}
\makeatother

\newtheorem{theorem}{Theorem}[section]
\newtheorem{corollary}{Corollary}[theorem]
\newtheorem{lemma}[theorem]{Lemma}
\newtheorem{proposition}[theorem]{Proposition}

\newtheorem{remark}[theorem]{Remark}

\tikzset{%
    add/.style args={#1 and #2}{
        to path={%
 ($(\tikztostart)!-#1!(\tikztotarget)$)--($(\tikztotarget)!-#2!(\tikztostart)$)%
  \tikztonodes},add/.default={.2 and .2}}
}  
\usepackage{cite}
\usepackage{balance}
\usepackage[margin=1.0in]{geometry}
\usepackage[toc,page]{appendix}

\begin{document}
\title{\bf LQR through the Lens of First Order Methods: \\Discrete-time Case}
\author{Jingjing Bu, Afshin Mesbahi, Maryam Fazel, and Mehran Mesbahi\thanks{The authors are with the University of Washington, Seattle; Emails:{\tt \{bu+amesbahi+mfazel+mesbahi\}@uw.edu}}}
\date{July 19, 2019}
\maketitle

\begin{abstract}
  We consider the Linear-Quadratic-Regulator (LQR) problem in terms of optimizing a real-valued matrix function over the set of feedback gains. Such a setup facilitates examining the implications of a natural initial-state independent formulation of LQR in designing first order algorithms. It is shown that this cost function is smooth and coercive, and provide an alternate means of noting its gradient dominated property. In the process, we provide a number of analytic observations on the LQR cost when directly analyzed in terms of the feedback gain. We then examine three types of well-posed flows for LQR: gradient flow, natural gradient flow and the quasi-Newton flow. The coercive property suggests that these flows admit unique solutions while gradient dominated property indicates that the corresponding Lyapunov functionals decay at an exponential rate; we also prove that these flows are exponentially stable in the sense of Lyapunov. We then discuss the forward Euler discretization of these flows, realized as gradient descent, natural gradient descent and the quasi-Newton iteration. We present stepsize criteria for gradient descent and natural gradient descent, guaranteeing that both algorithms converge linearly to the global optima. An optimal stepsize for the quasi-Newton iteration is also proposed, guaranteeing a $Q$-quadratic convergence rate--and in the meantime--recovering the Hewer algorithm.
%
  We then examine LQR state feedback synthesis with a sparsity pattern. In this case, we develop the necessary formalism and insights for projected gradient descent, allowing us to guarantee a sublinear rate of convergence to a first-order stationary point.
\end{abstract}

{\bf Keywords:} Linear quadratic regulators; first order methods; distributed control

\section{Introduction}
\label{sec:intro}
Linear-quadratic-regulator (LQR) has been one of the cornerstones of control theory since Kalman's original work in the 1960s.
LQR is formulated around an optimization problem for determining a sequence of (control) inputs to a linear system in order to minimize a given (integral) quadratic cost over an infinite horizon.\footnote{We shall not delve into the finite-horizon LQR in this paper.}
From the theoretical point of view, a fundamental property of LQR synthesis is that the resulting optimal input is in the form of a state feedback; as such, it can be represented as a constant feedback gain on the state of the system~\cite{Kalman_BSMM_1960,Anderson_book_1990}.
The state feedback gain that ``solves'' the infinite-horizon LQR problem, in turn, can be obtained by solving the algebraic Riccati equation (ARE). That is, in the traditional approach to LQR design, the state feedback gain is revealed after obtaining the ``certificate'' or ``cost-to-go'' for the underlying optimal control problem.\footnote{The analogy here would be solving the dual, followed by the recovery of the primal solution.}
Historically, a large number of works have studied the solution of ARE, including approaches based on iterative algorithms~\cite{Hewer1971TAC}, algebraic solution methods~\cite{Lancaster1995algebraic}, and semidefinite programming~\cite{Balakrishnan2003TAC}.

Although the cost function plays a fundamental role in the LQR problem, it is generally not ``recommended'' to {\em directly} compute the optimal gain (policy) using this cost function without solving the associated Riccati equation. This approach, in the meantime, is in sharp contrast to how one would typically go about minimizing a cost function over the variable of interest in introductory optimization, say, through gradient descent.\footnote{This is essentially due to the dynamic nature of the constraint set.} With recent advances in  sophisticated statistical and optimization methods, there has been a surge of interest in constructing optimal control strategies directly, viewing control synthesis through the lens of first order methods.\footnote{One might as well extrapolate that these methods provide a streamline recipe for learning optimal feedback gains in real-time.}
Adopting such a point of view has been partially inspired by the application of
learning algorithms, such as Reinforcement Learning (RL), where 
using principles of Dynamic Programming (DP), one can devise 
real-time model-free methods for both continuous-time and discrete-time LQR~\cite{Jiang2012Auto,Young2012Auto,Lee2019TAC,Bradtke1994ACC,Lewis2009CSM,lewis2012reinforcement,Chun2016IJC}.
%
Learning (over time) also has a spatial counterpart, realized in terms of  control of distributed systems. Such systems have become increasingly important in recent years; as such, it is desired to design feedback mechanisms that conform to a given sparsity pattern mirroring the underlying interaction topology amongst the various subsystems. That is, each ``node'' in the network forms its control action by employing local information collected from its neighbors; the corresponding zero pattern in the feedback gain mirrors this locality in the information exchange. 
Such design problems have gained a lot of attention in the  system and control community over the past two decades. However, there remains a host of issues in further understanding such class of problems. For example, in the case of structured synthesis, even the existence of an optimal structured LQR gain is nontrivial to assert. 
A rather brief sampling of related works on the structured synthesis problem is as follows.\footnote{With apologies for not going over a large body of work in this area.}
In~\cite{wenk1980parameter}, a combined primal-dual method with a penalty function is employed to obtain a feedback controller with the desired zero pattern. The work~\cite{jilg2013optimized} proposes a relaxed mixed-integer semidefinite-programming in which the graph topology is enforced through the integer constraints. Directly related to the present work is~\cite{maartensson2009gradient}, where the authors propose a projected gradient descent algorithm for structured synthesis. 

In this paper, inspired by the work~\cite{fazel2018global}, we examine first order methods for solving the centralized and distributed LQR problem.
In this direction, we first ``tweak'' the LQR problem formulation, motivated by the well-known fact 
that the state-feedback law is independent of the initial state of the system.
In order to eliminate the dependence on this initial state, we adopt a cost function that sums the traditional LQR cost over a set of linearly independent initial states.
This cost function can then be viewed as a well-defined matrix function over stabilizing feedback gains. We argue that this formulation (see \S\ref{sec:cost_function} for details) is necessary for the adoption of first order methods for LQR-type problems. More importantly, in this setting, we show that the cost is smooth, coercive and gradient dominated over its effective domain.\footnote{The property was first observed in~\cite{fazel2018global}; in this paper, we provide an alternate proof of this fact.} We then proceed to show that the LQR cost over the set of stabilizing state feedback gains does attain a minimum, by showing that all its sub-level sets are compact. 
Subsequently, using the topological and metrical properties of the set of
(static) stabilizing feedback gains, one can conclude that the proposed optimization formulation of LQR synthesis
does attain its global minimum. 
This cost function also gives rise to three types of well-posed flows over the set of stabilizing controllers, namely, gradient flow, natural gradient flow and the quasi-Newton flow. In this direction, we prove that the Lyapunov functionals for these flows decay at an exponential rate and the corresponding trajectories are exponentially stable in the sense of Lyapunov. 

We then proceed to discuss the forward Euler discretization of these flows, realized as gradient descent, natural gradient descent and the quasi-Newton iteration (hence, the state feedback gain can be updated iteratively). The problem of solving the LQR using direct gain (policy) update has been addressed in \cite{fazel2018global}, where it is shown that first-order gradient descent in fact converges to the optimal feedback gain.\footnote{To be more precise, the work~\cite{fazel2018global} establishes the convergence of cost function; as such, the convergence of iterates, i.e., feedback gains, is not shown explicitly. This setup was also considered in~\cite{mart2012phd}, without a convergence analysis.}
In~\cite{fazel2018global}, the gradient dominated property~\cite{polyak1963gradient}, is used to guarantee the global convergence of gradient descent and natural gradient descent.
The discretization scheme obtained in the present work is consistent with the setup adopted in~\cite{fazel2018global}, but in some ways, approaches the problem more directly and indeed, provides a practical choice of stepsize for gradient descent and improves the choice of stepsize for natural gradient descent and quasi-Newton iteraton.\footnote{Our approach primary aims to mold the LQR synthesis problem in the spirit of~\cite{polyak1963gradient}.}
We show that the stepsizes in the gradient descent natural gradient descent can be obtained via the Lyapunov equations in two consecutive updates; the coerciveness of the cost function on the other hand, ensures that the updated feedback gains remain stabilizing.
As such, both the function values and feedback gains converge linearly to the corresponding global minimum.
In view of these observations, one can then state that the proposed iterations generate a sequence of stabilizing feedback gains that converge linearly to the optimal LQR gain. 
Particularly in the case of natural gradient descent we obtain a sequence of value matrices that is monotonically decreasing on the positive semidefinite cone.\footnote{The terminology ``natural gradient descent'' (flow) is reserved for a particular choice of Riemannian metric; see \S\ref{sec:ngf} for details.} Convergence rate of the quasi-Newton iteration is also analyzed,\footnote{The quasi-Newton iteration is consistent with Hewer's algorithm~\cite{Hewer1971TAC}, essentially a Newton's iteration. However, traditionally, the emphasis has been placed on the convergence of the value matrices, rather than direct policy update. We provide a more transparent motivation as to why the proposed algorithm is a ``quasi-Newton'' iteration over direct policy space.} proving that the corresponding iterates and function values converge quadratically to the global optima.\footnote{The algorithm is referred to as ``Gauss-Newton'' in~\cite{fazel2018global} and the convergence was only shown to be linear {rather than} $Q$-quadratic.}
\par
Our work also considers the extension of the proposed synthesis framework to the problem of designing feedback gains with an arbitrary sparsity pattern. This setup is inspired by the scheme adopted in~\cite{mart2012phd}. In this direction, we propose a formalism to set up the problem where projected gradient descent has a simple realization. 
In the case of structured synthesis, the LQR cost function is
no longer ``gradient dominated'' and the choice of stepsize can not be generalized from the unstructured case. 
On the other hand, the proposed stepsize choice in~\cite{maartensson2009gradient} assumes a rather involved analytical form and convergence analysis to first-order stationary point is not straightforward.\footnote{The stepsize sequence described in~\cite{maartensson2009gradient} is asymptotically vanishing. As such, the convergence is only guaranteed if the sequence is square summable but not absolutely summable. However, these conditions were not verified in~\cite{maartensson2009gradient}. Furthermore, in~\cite{maartensson2009gradient}, it has been stated that the proposed algorithm will converge to a local minimum. This is not necessary valid as gradient descent for nonconvex objectives can in principle only converge to a first-order stationary point. One might invoke an ``escaping saddle'' type argument here but this requires more work.}
In this work, we adapt the machinery developed for the unstructured LQR for the structured synthesis: we first define the initial state independent LQR formulation and then show that the cost function can be equivalently defined as the unstructured LQR cost function restricted to the linear space defined by the information-exchange graph; as such, the cost function is smooth in the subspace topology and has a coercive property. Using this setup, we can obtain the gradient and Hessian of the cost function, leading to a natural choice of stepsize by bounding the Hessian over the initial sublevel set. We show this stepsize will guarantee a nonasymptotic sublinear convergence rate to the first-order stationary point.

The remainder of this paper is as follows. The LQR problem statement and related definitions are provided in \S\ref{sec:definition}. The averaged LQR cost over a set of linearly independent initial states is also defined in \S\ref{sec:definition}. \S \ref{sec:LQR} introduces the analytical properties of the LQR cost function.. Subsequently, gradient flow, natural (Riemannian) gradient flow and quasi-Newton flow are introduced in \S\ref{sec:gf}, \S\ref{sec:ngf} and \S\ref{sec:qnf}, respectively. Discrete realizations of these flows, namely, gradient descent, natural gradient descent and quasi-Newton iterations are addressed in \S\ref{sec:gd}, \S\ref{sec:ngd} and \S\ref{sec:qn_iteration}. \S\ref{sec:structured_LQR} introduces the formalism for setting up a first order approach for structured LQR synthesis, supplemented with the stepsize selection analysis and sublinear convergence to the first-order stationary point. \S\ref{sec:simulation} presents simulation results to illustrate the theoretical contributions of the paper; in \S\ref{sec:discussion}, we provide a few concluding remarks.
\section{Notation and Preliminaries}\label{sec:notations}
We denote by ${\bb M}_{n \times m}(\bb R)$ the set of $n \times m$ real matrices and ${\bb {GL}}_n(\bb R)$ as the set of invertible square matrices; $\bb R^n$ denotes the $n$-dimensional real Euclidean space with the $n=1$ case identified with real number. The set of non-negative numbers is denoted by $\bb R_+$ and natural numbers as
$\bb N$; $\bb S_n$ denotes the set of $n \times n$ real symmetric matrices.
Other notation includes $A^\top$, $\rho(A)$, $\rank(A)$, $\Tr (A)$, $\vect(A)$ representing the transpose, spectral radius, rank, trace, and vectorization of the matrix $A$, respectively;  $A \otimes B$ is the Kronecker product of matrices $A$ and $B$, and $\text{\bf rbd } \ca K$ designates the relative boundary of the set $\ca K$. The real inner product between a pair of vectors $x$ and $y$ is denoted by $\langle x,y \rangle$. $\|A\|_2$ denotes the spectral (operator) norm of a square matrix $A$ and $\|A\|_{F}$ denotes its Frobenius norm.\footnote{2-norm is assumed when we use $\|.\|$.}
Lastly, the notation $A \succeq B$ for two symmetric matrices refers to the positive semi-definiteness of their difference $A-B$; analogously for positive definiteness of this difference using $A \succ B$. We let $\lambda_i(A)$ denote the eigenvalues of a square matrix $A$. These eigenvalues are indexed in an increasing order with respect to their real parts, i.e.,
\begin{align*}
  {\bf Re}(\lambda_1(A)) \le \dots \le {\bf Re}(\lambda_n(A)).
  \end{align*}
  If $A$ is symmetic, the ordering becomes $\lambda_1(A) \le \dots \le \lambda_n(A)$. 
  When $A \succeq 0$, $\|A\| = \lambda_n(A)$ and we shall use these interchangeably.
We use $C^{\omega}(U)$ to denote the set of real analytic functions over an open set $U \subseteq \bb R^n$. A function $f: U \to \bb R$ is $C^{\infty}$-smooth if it is infinitely differentiable. A function $f$ is $L$-smooth when $f$ is \emph{continuously differentiable} and its gradient is $L$-Lipschitz, i.e., $\|\nabla f(x)-\nabla f(y)\| \le L \|x-y\|$.
A pair $(A, B)$ with $A \in {\bb M}_{n \times n}(\bb R)$ and $B \in {\bb M}_{n \times m}( \bb R)$ is called controllable if the Kalman rank condition~\cite{sontag2013mathematical},
\begin{align*}
  \rank([B, AB, A^2B, \dots, A^{n-1}B]) = n,
\end{align*}
is satisfied. 
Given such system matrices, $\ca S$ denotes the set of Schur stabilizing feedback gains, 
$$\ca S =\{K \in {\bb M}_{m \times n}(\bb R): \rho(A-BK) < 1\}.$$


We will frequently use several linear algebraic facts on matrix equations; some of these are collected in the following proposition.
\begin{proposition}
  \label{prop:linalg_facts}
  The following relations hold:
  \begin{enumerate}
    \item For matrices $A,B,C$ of appropriate dimensions, $\vect(ABC) = (C^{\top} \otimes A) \vect(B)$.
    \item When $X \succ 0$,
\begin{align}
  \label{eq:psd_ineq1}
 M^{\top} X N + N^{\top} X M & \succeq - ( a M^{\top} X M + \frac{1}{a}N^{\top} X N),\\
  \label{eq:psd_ineq2}
 M^{\top} X N + N^{\top} X M & \preceq a M^{\top} X M +  \frac{1}{a} N^{\top} X N,
\end{align}
where $M, N \in {\bb M}_{n \times m}(\bb R)$ with $m \le n$ and $a \in \bb R_+$.
\item Suppose that $A \in {\bb M}_{n \times n}(\bb R)$ has spectral radius bounded by $1$, i.e., $\rho(A) < 1$. Then
\begin{align*}
  A^{\top} X A + Q - X =0
\end{align*}
has a unique solution,
\begin{align*}
  X = \sum_{j=0}^{\infty} (A^{\top})^j Q A^j ,
\end{align*}
and when $Q \succ 0$ then $X \succ 0$. Moreover, if $\widetilde{X}$ satisfies
\begin{align*}
  A^{\top} \widetilde {X} A + \widetilde {Q} -\widetilde{X} = 0,
\end{align*}
with $\widetilde {Q}\preceq Q$,
then $\widetilde{X} \preceq X$.
\item If $X, Y$ are both positive definite, then
\begin{align}
 \lambda_1(Y) \Tr(X) \le \Tr(XY) \le \lambda_n(Y) \Tr(X).
\end{align}
  \end{enumerate}
\end{proposition}
The proofs of these observations can be found in~\cite{horn2012matrix}.
\section{Problem Setup and its Analytic Properties}\label{sec:definition}
In this section, we provide an overview of LQR, and in particular its modified initial state independent version, as well as a few analytic observations that are of independent interest. Although the reader might know of the extensive LQR literature, we note that some of these observations have only become necessary when the LQR optimization is viewed {\em directly} on the set of stabilizing feedback gains.
\subsection{Discrete-time LQR} 
In the standard setup of LQR, we consider a (discrete-time) linear time invariant model of the form,
\begin{align}
  x_{k+1} = Ax_k + B u_k, \label{LTI}
\end{align}
where $A \in {\bb M}_{n \times n}( \bb R)$ and $B \in  {\bb M}_{n \times m}( \bb R)$. The LQR problem is the optimization problem of devising a linear feedback gain $K \in {\bb M}_{m \times n} (\bb R)$ for which $u_k = -K x_k$, minimizing,\footnote{The condition that  $u_k$ has the form $-K x_k$ is not set a priori in the LQR formulation; this feedback form is typically shown via the adoption of a dynamic programming step.}
\begin{align*}
  J (x_0) = \sum_{k=0}^{\infty} \left[ \langle x_k, Q x_k\rangle + \langle u_k, R u_k\rangle\right],
\end{align*}
where $x_0$ is the initial condition, and the quadratic cost is parameterized by 
$0 \preceq Q \in \bb S_n,$ and $0 \prec R \in \bb S_m$.
LQR is traditionally solved via dynamic programming or calculus of variations, leading to the celebrated Algebraic Riccati Equation (ARE)~\cite{Anderson_book_1990}.\footnote{For the dynamic programming case, one starts with the finite horizon case, apply the optimality principle, and then identify a solution concept for the infinite horizon case using a limit argument; calculus of variations provide another approach for deriving necessary conditions for LQ-type problems.}\par
\subsection{Cost function for direct policy update}\label{sec:cost_function}
In order to update the feedback gain (policy) directly, it will be conceptually appealing to consider the cost as a matrix function over the set of feedback gains. With this aim in mind, we may define $J_{x_0} \colon {\bb M}_{m \times n}(\bb R) \to \bb R$ as,
\begin{align}
  \label{eq:naive_cost_function}
  K \mapsto J_{x_0}( K) = \sum_{k=0}^{\infty} & \left[ \langle (A-BK)^k x_0, Q(A-BK)^k x_0\rangle \right.\left. + \langle K(A-BK)^k  x_0, RK(A-BK)^k x_0 \rangle\right],
\end{align}
for some fixed initial condition $x_0 \in \bb R^n$.
Our first task in this direct optimization setup is to determine the domain over which the function is well-defined. 
In other words, we are interested in the effective domain $\dom(J_{x_0}) = \{K \in {\bb M}_{m \times n} (\bb R): J_{x_0}(K) < +\infty\}$. Addressing this seemingly natural analytical question turns out to be subtle. If $K$ is stabilizing, i.e., $\rho(A-BK) < 1$, then $K \in \dom(J_{x_0})$. In the meantime, for a non-stabilizing $K$, i.e., $\rho(A-BK) \ge 1$, when the system matrix $A-B K$ has both stable and unstable modes, if $x_0$ is chosen to be in the span of eigenspace corresponding to stable modes, $J_{x_0}(K) < \infty$. That is, $\{K: \rho(A-BK) < 1\}$ is a proper subset of $\dom(J_{x_0})$. Indeed, $\{K: \rho(A-BK) < 1\}$ is the interior of $\dom(J_{x_0})$. Before proving this, we show that 
the set of feedback gains for which a fixed vector $x$ is not orthogonal to any eigenvector of the closed-loop system is dense.
\begin{proposition}
  \label{prop:eigen_components}
  Suppose that $(A, B)$ is controllable and $x \in \bb R^n$ is a fixed vector. 
  Then the set,
  \begin{align*}
    \ca W = \{K \in \bb M_{m \times n}(\bb R): \mbox{all} \; \text{ eigen-pairs } (\lambda, v) \text{ of } A-BK, \; \langle v,x \rangle \not = 0\},
  \end{align*}
  is dense in $\bb M_{m \times n}(\bb R)$.
\end{proposition}
\begin{proof}
  Without loss of generality, we may assume that $x = (0, \dots, 0, 1)^{\top} \in \bb R^n$. We note that,
  \begin{align*}
    \ca W^c = \{K \in \bb M_{m \times n}(\bb R): \exists \text{ eigen-pair } (\lambda, v) \text{ of } A-BK, \; \langle v,x \rangle = 0\},
  \end{align*}
  and 
  if $K \in \ca W^c$, then the first $n-1$ columns of $\lambda I - (A-BK)$ has rank smaller than $n-1$ since the non-trivial kernel of the first $n-1$ columns of $\lambda I - (A-BK)$ appending $0$ would be an eigenvector orthogonal to $x$. But this condition is equivalent to vanishing all $(n-1) \times (n-1)$ minors of the first $n-1$ columns of $\lambda I - (A-BK)$. Each of these minors is a polynomial $p_j(t, K, \lambda)$ in the entries of $K$ and $\lambda$. Denote $E_j = \{(K, \lambda) \in \bb M_{n \times n}(\bb R) \times \bb C: p_j(t, K, \lambda )=0\}$. Each set $E_j$ is closed in Zariski topology; as such, $\bigcap E_j$ is Zariski closed in $\bb M_{m \times n}(\bb R) \times \bb P^1$, where $\bb P^1$ the projective variety. But $\ca W^c$ is precisely the projection of $\bigcap E_j$ onto $\ca \bb M_{n \times n}(\bb R)$. Since this projection  
  is a closed map (in the Zariski topology), $\ca W^c$ is Zariski closed. Consequently as an nonempty set, $\ca W$ is Zariski open and thus dense.
\end{proof}
We are now in the position to prove a result concerning the interior of $\dom(J_{x_0})$.
\begin{lemma}
  Suppose that $x_0 \in \bb R^n$ is fixed. If $J_{x_0}$ is defined by~\eqref{eq:naive_cost_function}, then the set of Schur stabilizing feedback gains $\ca S$ is the interior of $\dom(J_{x_0})$.
\end{lemma}
\begin{proof}
  Clearly $\ca S \subseteq \text{int}(\dom(J_{x_0}))$. On the other hand,
let $M \in \dom(J_{x_0})\setminus \ca S$ and every $\varepsilon > 0$, by Proposition~\ref{prop:eigen_components}, there is some $N \in \bb M_{m \times n}(\bb R)$ such that $\|M-N\|_F < \varepsilon$ and the projection of $x_0$ onto every eigenvector of $A-BN$ is nontrivial. We observe that $\|A-BM - (A-BN)\|_F \le \|B\|_F\|M-N\|_F$. Since spectral radius is continuous and $\rho(A-BM) \ge 1$, $\rho(A-BN) \ge 1$. As such, $J_{x_0}(N) = \infty$ and $N \notin \dom(J_{x_0})$. Hence, $M \notin \text{int}(\dom(J_{x_0}))$ and $\text{int}(\dom(J_{x_0})) = \ca S$.
\end{proof}
The above lemma implies that $J_{x_0}(K)$ is not differentiable everywhere on its domain. More precisely, $J_{x_0}(K)$ is differentiable on $\ca S$ but non-differentiable on $\dom(J_{x_0}) \setminus \ca S$. This complication is rather unnecessary as we are primarily interested in stabilizing controllers. This motivates us to examine initial condition independent formulation of LQR.\footnote{This is indeed necessary if we want to formulate an unconstrained optimization problem over the set of stabilizing feedback gains.}
\subsection{Initial condition independent formulation of LQR}
Ideally, the objective function $f: {\bb M}_{m \times n} ( \bb R) \to \bb R$ for our LQR calculus has an effective domain that coincides with the set of stabilizing feedback gains $\{K \in {\bb M}_{m \times n}( \bb R): \rho(A-BK) < 1\}$. This can be achieved by choosing a set of linearly independent vectors $\{x_0^1, \dots, x_0^n\} \subseteq \bb R^n$ and defining,\footnote{Of course, one may choose the standard basis $\{e_1, \dots, e_n\}$, where $e_i$ is the vector with zero entries except a ``1'' at the $i$th entry; the choice of an arbitrary basis simply retains flexibility.} 
\begin{align}
  f(K) = \sum_{j=1}^n J_{x_0^j}(K). \label{f(k)}
\end{align}
As such, the function $f$ would be infinite if $K$ is not stabilizing (see Lemma \ref{lemma:coercive} for details).
\begin{remark}
  The initial independent formulation is rather natural for general optimal control problems. In such problems, it is often desired to constrain the control synthesis to stabilizing feedback gains. For a learning algorithm that is built around a descent direction, such a formulation allows for an automatic enforcement of this stabilizing feature.
\end{remark}
 We shall now see that $f$ (\ref{f(k)}) enjoys several favorable properties, {e.g.}, $f$ is differentiable over its effective domain and $f$ diverges to infinity when $K$ tends to the boundary of this domain, i.e., $f$ is coercive. More importantly, for every $K \in \dom(f)$, the function $f(K)$ can be written as,
\begin{align*}
  f(K) = \sum_{j=1}^n \Tr(X {\bf \Sigma}^j),
\end{align*}
where ${\bf \Sigma}^j = x_0^j (x_0^j)^\top$ and $X$ satisfies the Lyapunov equation,
$$(A-BK)^\top X (A-BK) + K^{\top} R K + Q = X.$$ 
Note that $J_{x_0^j}(K)$ does not necessarily admit the compact form $J_{x_0^j}(K)= \Tr(X {\bf \Sigma}^j)$ for every $K \in \dom(J_{x_0^j}(K))$. This is due to the fact that matrix $X$ only makes (mathematical) sense if $K$ is stabilizing, but $\dom(J_{x_0^j}(K))$ contains non-stabilizing feedback gains; see \ref{sec:cost_function}.
\par
\begin{remark}
  Alternatively, we could let $x_0 \sim \ca D$, where $\ca D$ denotes some probability distribution, and let
  \begin{align*}
    f(K) = \bb E_{x_0 \sim \ca D}(J_{x_0}).
  \end{align*}
  As long as the samples span the whole space with probability $1$, the function enjoys same properties as we have defined above. This is indeed the formulation adopted in~\cite{mart2012phd, fazel2018global}, without discussing its implications on differentiablility and coerciveness of $(\ref{f(k)})$.
\end{remark}
\subsection{Analytical Properties of the LQR cost function} \label{sec:LQR}
In this section, we investigate the properties of the LQR cost (\ref{f(k)}). We will observe that,
\begin{itemize}
  \item $f$ is a real analytic function over its domain.
  \item $f$ is coercive and has compact sublevel sets.
    \item $f$ is gradient dominated.
      \item The Hessian $\nabla^2 f$ is characterized.
\end{itemize}
To simplify the notation, in the rest of this paper, we shall denote\footnote{On some occations, we use subscript $M_K$ to emphasize the dependence on feedback gain $K$.}
\begin{align*}
     {\bf \Sigma} = \sum_{j=1}^n {\bf \Sigma}^{j}, \qquad A_K \coloneqq A-BK,  \quad \mbox{and} \quad M := RK-B^{\top} X(A-BK).
\end{align*}
%
Let us recall some of the topological properties of the set of Schur stabilizing feedback gains $\ca S$; the proofs can be found in~\cite{bu2019topological_mimo}.

%
\begin{lemma}
  The set $\ca S$ is regular open, contractible, and unbounded when $m \ge 2$ and the boundary $\partial \ca S$ is precisely the set $\ca B=\{K \in {\bb M}_{m \times n}(\bb R): \rho(A-BK) = 1\}$. 
\end{lemma}
We now observe that $f(K)$ is real analytic over $\ca S$.
\begin{lemma}
  \label{lemma:well-defined}
 For the LQR cost $(\ref{f(k)})$, we have  $f \in C^{\omega}(\ca S)$.
\end{lemma}
\begin{proof}
  For every $K \in \ca S$, let $X$ be the solution to the Lyapunov equation,
  \begin{align}
    \label{eq:lyapunov_eq}
    (A-BK)^{\top} X (A-BK) + K^{\top} R K + Q = X.
  \end{align}
   Then
  \begin{align}
    \label{eq:lyapunov_matrix}
    \vect(X) =  \left( A_K^{\top} \otimes A_K^{\top} \right) \vect(X) + \vect( K^{\top} R K  + Q).
  \end{align}
  Since the eigenvalues of $ I \otimes I - A_K \otimes A_K$ are $\{1 - \lambda_i(A_K) \lambda_j(A_K): i,j = 1, \dots, n\}$, $I \otimes I - A_K^{\top} \otimes A_K^{\top}$ is invertible. Hence, 
  \begin{align*}
    \vect(X) = \left( I \otimes I - A_K \otimes A_K \right)^{-1} \vect\left( K^{\top} R K + Q\right).
  \end{align*}
  By Cramer's rule, $X(K)$ is a rational function of polynomials in the entries of $K$ and thus
  the map $K \mapsto X(K)$ is $C^{\omega}$.\footnote{For a given $K$, $X(K)$ will be referred to as the the ``cost matrix'' as it characterizes the infinite horizon closed loop cost from the current state when this cost is finite. We shall use $X(K)$ and $X$ interchangebly.}
  Hence, $f$ can be viewed in terms of the composition,\footnote{Mind that this perspective is only valid in $\ca S$.}
  \begin{align*}
    K \mapsto X(K) \mapsto \Tr(X {\bf \Sigma}).
  \end{align*}
  As a composition of $C^{\omega}$ maps, $f$ is thus real analytic. \par
\end{proof}
With the initial condition independent formulation, the function $f$ (\ref{f(k)}) diverges to infinity smoothly as $K$ approaches the boundary $\partial \ca S$ or when $K$ diverges to infinity.
\begin{lemma}
  \label{lemma:coercive}
  The LQR cost (\ref{f(k)}) is coercive in the sense that,
  \begin{align*}
    &\lim_{K_j \to K \in \partial S} f(K_j) = + \infty, 
    \end{align*}
    or 
   \begin{align*}
    &f(K) \to \infty \text{ if } K \in \ca S \text{ and } \|K\| \to \infty.
  \end{align*}
\end{lemma}
\begin{proof} Suppose that the sequence $\{K_j\} \subseteq \ca S$ and $K_j \to K \in \partial \ca S$. 
By continuity of the spectral radius, we have $\rho(A-BK_j) \to \rho(A-BK)$. This means that for every $\varepsilon > 0$, there exists some $N= N(\varepsilon) \in \bb N$ for which $|\rho(A-BK_j) - \rho(A-BK)| < \varepsilon$ for every $j \ge N$. That is $1 > \rho(A-BK_j) > 1-\varepsilon$ for all $j \ge N$. Let $X$ be the cost matrix associated with $K_j$.
We observe that,
\begin{align*}
  f(K_j) & = \Tr(X{\bf \Sigma}) \ge \lambda_1({\bf \Sigma}) \Tr(X) = \lambda_1({\bf \Sigma}) \, \Tr (\sum_{i=0}^{\infty} (A_{K_j}^{\top})^i \, (Q+K_j^{\top} R K_j) \, (A_{K_j})^i ) \\
&\ge \lambda_1({\bf \Sigma}) \, \lambda_1(Q) \, \sum_{i=0}^{\infty} \Tr \left( (A_{K_j}^{\top})^i (A_{K_j})^i\right).
\end{align*}
Note that $\Tr ((A_{K_j}^{\top})^i (A_{K_j})^i) \ge \rho(A_{K_j})^{2i}$ since,
\begin{align*}
  \lambda_n \, ((A_{K_j}^{\top})^i (A_{K_j})^i ) &= \sup_{\|v\|_2 = 1}v^{\top} (A_{K_j}^{\top})^i (A_{K_j})^i v = \sup_{\|v\|_2=1} \|(A_{K_j})^i v\|_2^2 \geq \rho(A-BK_j)^{2i}.
\end{align*}
It thus follows that,
\begin{align*}
  f(K_j) &\ge \lambda_1({\bf \Sigma}) \, \lambda_1(Q) \, \sum_{i=0}^{\infty}\rho(A-BK_j)^{2i} \ge\lambda_1({\bf \Sigma}) \, \lambda_1(Q) \, \frac{1}{1-(1-\varepsilon)^2}.
\end{align*}
For any $M > 0$, picking a sufficiently small $\varepsilon$ would lead to $f(K_j) \ge M$ for all $j \ge N(\varepsilon)$.

  On the other hand, 
  \begin{align*}
    f(K) &\ge  \lambda_1 \, (\sum_{i=0}^\infty (A_K)^i \, {\bf \Sigma} \, (A_K^{\top})^i) \, \Tr(Q+K^\top R K) \ge \lambda_1({\bf \Sigma}) \, \lambda_1(R)\|K\|_F^2.
  \end{align*}
Thereby, for any $M > 0$, $f(K) \ge M$ for $\|K\|$ sufficiently large.
\end{proof}
With the coercive property in place, i.e., growth to infinity smoothly, we can continuously extend the function to $\bb M_{n \times n}(\bb R)$ as an extended real-valued function which allows $+\infty$ as a function value. This in turn will imply that all sublevel sets of $f(K)$ are compact.\footnote{This can also be proved directly. The condition $f(K) \to +\infty$ as $\|K \| \to \infty$ implies that the sublevel sets if $f$ are bounded; condition $f(K_j) \to +\infty$ as $K_j \to K \in \partial \ca S$ implies that every sublevel set is bounded away from the boundary and hence closed in the Euclidean topology. Note the continuity of $f$ only guarantees that the sublevel set is closed in $\ca S$.}
\begin{corollary} \label{cor:compact}
  The sublevel set $\ca S_{\alpha} = \{ K \in \ca S: f(K) \le \alpha\}$ is compact for every $\alpha > 0$.
\end{corollary}
\begin{proof}
  By Lemma \ref{lemma:coercive}, we can continuously extend $f$ to $\tilde{f}: \bb M_{m \times n}(\bb R) \to \bb R \cup \{+\infty\}$, where,
\begin{align*}
  \tilde{f} = \begin{cases}
    f(K), & \text{ if } K \in \ca S,\\
    +\infty, & \text{ if } K \in \ca S^c.
  \end{cases}
\end{align*}
The sublevel sets $\tilde{\ca S}_{\alpha} = \{K \in \bb M_{n \times n}(\bb R): \tilde{f}(K) \le \alpha\}$ of $\tilde{f}$, in the meantime, are compact by Proposition $11.12$ in \cite{bauschke2017convex}. The proof is completed by observing that
  $\ca S_{\alpha} = \tilde{\ca S}_{\alpha}$
when $\alpha$ is finite.
\end{proof}
As $f \in C^{\omega}(\ca S)$, the gradient of $f$ can be characterized explicitly.
\begin{proposition}{(Proposition $1$ in \cite{mart2012phd})}
  \label{prop:gradient}
  For $K \in \ca S$, $\nabla f(K) = 2 \left(RK-B^{\top}X (A-BK)\right)Y$, where $Y$ solves the Lyapunov matrix equation,
\begin{align}
  \label{eq:lyapunov_Y}
  A_{K}Y A_K^{\top} - Y + {\bf \Sigma} = 0.
\end{align}
\end{proposition}
We emphasize that Proposition~\ref{prop:gradient} only makes sense when $K \in \ca S$.\footnote{This point has not been discussed in~\cite{mart2012phd}. Indeed, it does not even make sense to have $X$ and $Y$ if $K$ is not stabilizing.} One is then tempted to set $\nabla f(K) = 0$ to obtain a stationary point. However, since $X$ is a function of $K$, whether or not $\nabla f(K) = 0$ is solvable in $\ca S$ needs clarification.
\begin{lemma}
The matrix $K_* = (B^\top X_* B + R)^{-1} B^{\top} X_* A$ is the unique global minimizer of $f(K)$\footnote{As we are establishing the global minimizer is unique, throughout the paper we shall use $K_*$ to denote the global minimizer.}, where $X_*$ is the corresponding solution of the Lyapunov equation \eqref{eq:lyapunov_matrix}.\footnote{Here we are not assuming prior knowledge of control theory. Of course, control experts and students alike may readily identify that the solution is indeed the optimal LQR gain via the ARE.}
\end{lemma}
\begin{proof}
  Since $(A, B)$ is controllable, $\ca S$ is nonempty. As such, for some finite $c > 0$, the set $\ca S_c = \{K \in \ca S: f(K) \le c\}$ is a nonempty compact set. Therefore, $f(K)$ achieves its minimum on $\ca S_c$. Note that as $f(K)$ is not constant, this minimum must be in the interior of $\ca S_c$ and as such, $\nabla f(K_*) = 0$. Thereby, $K_* = (B^\top X_* B + R)^{-1} B^{\top} X_* A$ must be in $\ca S$ (this expression is now more precise!). Since $f$ has only stationary point, $K_*$ must be the global minimum. 
\end{proof}
Next we derive a formula for the Hessian of $f(K)$. The upper bound on the norm of this Hessian will then suggest a viable choice of stepsize for (projected) gradient descent.
\begin{proposition}
  For $K \in \ca S$, the (self-adjoint) Hessian of the LQR cost $f$ (\ref{f(k)}) is characterized by,
\begin{equation}
\begin{split}
  \label{eq:hessian_discrete}
    \nabla^2 f(K)[E, E] &= 2\langle (RE+B^{\top} X BE)Y, E \rangle - 4\langle (B^{\top} (X'(K)[E]) A_K)Y, E \rangle,
\end{split}
\end{equation}
  where $E \in {\bb M}_{m \times n}(\bb R)$ and $X'(K)[E]$ denotes the action of differential of the map $K \mapsto X(K)$\footnote{To be more precise, the differential of $X: M_{m \times n}(\bb R) \to M_{n}(\bb R)$ is a map $X': M_{m \times n}(\bb R) \to \ca L(M_{m \times n}, \ca L(M_{m \times n}(\bb R), M_n(\bb R)))$, where $\ca L$ denotes the set of bounded linear maps. As such, $X'(K) \in \ca L(M_{m \times n}(\bb R), M_n(\bb R))$ and $X'(K)[E] \in M_n(\bb R)$.}\footnote{Recall that $X$ solves the Lyapunov equation~\eqref{eq:lyapunov_matrix} and $Y$ solves Lyapunov matrix equation~\eqref{eq:lyapunov_Y}.}.
  \end{proposition}
\begin{proof}
  We note that $\nabla f(K): {\bb M}_{m \times n}(\bb R) \to {\bb M}_{m \times n}(\bb R)$.
  Let $g(K) = \nabla f(K) = P(K)Y(K)$, where $P(K) \coloneqq 2M= 2(RK-B^{\top} X {A_K})$. By Lebnitz' rule,
  \begin{align*}
    \nabla g(K) = P'(K) Y(K) + P(K) Y'(K).
  \end{align*}
  Note that for $E \in M_{m \times n}(\bb R)$, the action of $\nabla g(K)[E]$ is given by,\footnote{Throughout this paper, we use the notation $\nabla F(K)[E]$ to denote the action on $E$ of the differential $\nabla F$ evaluated at $K$, i.e., $\nabla F(K)$.}
  \begin{align*}
    \nabla g(K)[E] = P'(K)[E] Y(K) + P(K) Y'(K)[E],
  \end{align*}
  where $P'(K)[E]$ and $Y'(K)[E]$ denote the usual matrix multiplication.
 Hence, 
  \begin{align*}
    \nabla^2 f(K)[E, E ]=  2\langle (RE + B^\top X B E - B^{\top} (X'(K)[E]) A_K)Y, E \rangle + 2\langle (RK-B^{\top} X A_K)(Y'(K)[E]), E \rangle,
  \end{align*}
  where $X'(K)[E]$ satisfies,
  \begin{align*}
    A_K^{\top}X (-B E) + (-BE)^{\top} X A_K + A_K^{\top} (X'(K)[E]) A_K + E^{\top} R K + K^{\top} R E = X'(K)[E],
  \end{align*}
  and 
  \begin{align*}
    Y'(K)[E] = (-BE) Y A_K^{\top} + A_K Y(-BE)^{\top} + A_K (Y'(K)[E]) A_K^{\top}. 
  \end{align*}
 Note that $X'(K)[E]$ and $Y'(K)[E]$ are uniquely defined if $K \in \ca S$ and can be written as,
  \begin{align}
    \label{eq:X_derivative}
    X'(K)[E] &= \sum_{j=0}^{\infty} (A_K^\top)^j \left( (K^{\top}R-A_K^{\top}XB)E \right. \left.+ E^{\top}(RK-B^{\top}XA_K)\right) (A_K)^j, \\ 
    Y'(K)[E] &= \sum_{j=0}^{\infty} (A_K)^j \left( -BEYA_K^{\top}  -A_KYE^{\top} B^{\top}\right)(A_K^\top)^j. \nonumber
  \end{align}
Using the cyclic property of the matrix trace, we observe that,
  \begin{align*}
    \langle B^{\top} (X'(K)[E]) A_K)  Y , E \rangle =  \langle (B^{\top} X A_K-RK)(Y'(K)[E]), E \rangle.
  \end{align*}
  The action of the Hessian can hence be simplified as,
  \begin{align*}
    \nabla^2 f(K)[E, E] = 2\langle (RE + B^\top X B E )Y, E \rangle  - 4\langle ( B^{\top} (X'(K)[E]) A_K)Y, E \rangle.
  \end{align*}
\end{proof}
We note that as $\nabla^2 f(K)$ is self-adjoint, its operator norm can be characterized as,
\begin{align}
    \|\nabla^2 f(K)\|^2 &= \sup_{\|E\|_F = 1} \|\nabla^2f(K)[E]\|_F^2 = \sup_{\|E\|_F = 1} \left(\nabla^2f(K)[E, E]\right)^2.
                      \label{eq:hessian-operator-norm}
  \end{align}
\begin{remark}
  \label{remark:hessian_stationary}
  We note that at $K_*$, for every $E \in \bb M_{m \times n}(\bb R)$, the action of Hessian is positive:
  \begin{align*}
    \nabla^2 f(K_*)[E, E] = 2\langle (RE+B^{\top} X_* BE)Y_*, E \rangle > 0,
  \end{align*}
  namely, $\nabla^2 f(K)$ is positive definite. This validates that $K_*$ is a local minimizer--and thus--the global minimizer, as $K_*$ is the unique stationary point.
\end{remark}
{
  We now observe that the LQR cost (\ref{f(k)}) is a \emph{gradient dominated} function~\cite{polyak1963gradient}.\footnote{This property is also referred as Polyak-{\L}ojasiewicz condition, as a special case of what had been proposed in~\cite{lojasiewicz1963propriete}.} The proof of this property in~\cite{fazel2018global} (Corollary 5) is based on a careful comparison of the cost difference in each time step between the optimal policy and a specified policy. Here, we provide an alternate proof of this important property. This alternate approach is more control-theoretic in the sense that it is mainly concerned with the properties of the Lyapunov equation. Moreover, this approach allows determining an upper bound on the gradient dominance coefficient--that in turn--facilitates estimating the iteration compexity of the gradient descent algorithm to reach an $\varepsilon$-precision solution for LQR.
  }
    \begin{lemma}
      \label{lemma:gradient_dominant}
      Let $K_*$ be the optimal feedback gain. For $K \in \ca S$,
\begin{align*}
  f(K) - f(K_*) \le \frac{\lambda_n(Y_*)}{4\lambda_1(R+B^{\top}X_* B) \lambda_1^2({\bf \Sigma})} \langle \nabla f(K), \nabla f(K)\rangle,
\end{align*}
where $Y_* = \sum_{j=0}^{\infty} (A_{K_*})^{j} \, {\bf \Sigma} \, (A_{K_*}^{\top})^j$ and $X_*$ solves the Lyapunov matrix equation
      \begin{align}
       \label{eq:optimal_value_matrix} {A}_{K_*}^{\top} X_* A_{K_*} - X_* + K_*^{\top} R K_* + Q = 0.
      \end{align}
    \end{lemma}
    \begin{proof}
      Recall that $M \coloneqq RK-BX{A_K}$ and $X$ is the solution of~\eqref{eq:lyapunov_eq}.
      %
      Taking the difference of equations~\eqref{eq:lyapunov_eq} and~\eqref{eq:optimal_value_matrix}, i.e., $\eqref{eq:lyapunov_eq} - \eqref{eq:optimal_value_matrix}$, we obtain,
      \begin{align*}
        A_{K}^{\top} X A_K - X + K^{\top} R K - A_{K_*}^{\top} X_* A_{K_*} + X_* - K_*^{\top} R K_* = 0.
      \end{align*}
      A few algebraic manipulations now yield,
      \begin{equation} \label{eq:difference}
      \begin{split}
        &{A}_{K_*}^{\top} (X-X_*) {A}_{K_*} - (X-X_*) + (K-K_*)^{\top}(RK-B^{\top} X A_K) + (K^{\top}R - A_K^{\top} X B)(K-K_*) \\
& - (K-K_*)^{\top} R (K-K_*)  - (A_K-A_{K_*})^{\top} X (A_K-{A}_{K_*}) = 0. 
      \end{split}
   \end{equation}
      By Proposition~\ref{prop:linalg_facts} (part \eqref{eq:psd_ineq2}), we note that for every $\alpha > 0$,
      \begin{align*}
        (K-K_*)^{\top}(RK-B^{\top} X {A_K})  + (K^{\top}R - A_K^{\top} X B)(K-K_*) \preceq \frac{1}{\alpha} (K-K_*)^{\top} (K-K_*) + \alpha M^{\top}M.
      \end{align*}
Picking $\alpha = 1/(\lambda_1(R+ B^{\top} X_* B)$, we then have,
\begin{align*}
 & (K-K_*)^{\top}M + M^{\top}(K-K_*)  - (K-K_*)^{\top} R (K-K_*) - ({A}-{A}_{K_*})^{\top} X (A_K-{A}_{K_*}) \\
  & \preceq \frac{1}{\lambda_1(R + B^{\top} X_* B)} M^{\top} M + (K-K_*)^{\top}(\lambda_1(R + B^{\top} X_* B) I - R - B^{\top} X B)(K-K_*)  \\
  &\preceq \frac{1}{\lambda_1(R + B^{\top} X_* B)} M^{\top} M.
\end{align*}
Let $Z$ be the solution of the Lyapunov equation,
\begin{align*}
  A_{K_*}^{\top} Z A_{K_*} - Z + \frac{1}{\lambda_1(R + B^{\top} X_* B)} M^{\top} M = 0;
\end{align*}
by  Proposition~\ref{prop:linalg_facts} (part (c)) , $X-X_* \preceq Z$ and
\begin{align*}
  Z = \sum_{j=0}^{\infty} ({A}_{K_*}^{\top})^{j} \, (\frac{1}{\lambda_1(R + B^{\top} X_* B)} M^{\top} M) \, (A_{K_*})^j.
\end{align*}
 It thus follows that,
\begin{align*}
  \Tr  \left( (X  -X_*)  {\bf \Sigma} \right) &\le \Tr(Z {\bf \Sigma})  = \Tr ( (\sum_{j=0}^{\infty} ({A}_{K_*}^{\top})^{j} \, (\frac{1}{\lambda_1(R + B^{\top} X_* B)} M^{\top} M) \,A_{K_*}^j) {\bf \Sigma})\\
&\qquad = \frac{1}{\lambda_1(R + B^{\top} X_* B)} \, \Tr (M^{\top} M (\sum_{j=0}^{\infty} {A}_*^{j} {\bf \Sigma} (A_*^{\top})^j)),
\end{align*}
where in the last equality we have used the cyclic property of the matrix trace.
We note that $Y_* = \sum_{j=0}^{\infty} ({A}_{K_*})^{j} {\bf \Sigma} (A_{K_*}^{\top})^j$ is uniquely determined by the system parameters $A, B, Q, R, {\bf \Sigma}$. Now
\begin{align*}
  \Tr((X-X_*) {\bf \Sigma}) &\le \frac{\lambda_n(Y_*)}{\lambda_1(R + B^{\top} X_* B)} \Tr(M^{\top} M) \\
&\le \frac{\lambda_n(Y_*)}{\lambda_1(R + B^{\top} X_* B) \lambda_1^2(Y)}\Tr(YM^{\top} M Y) \\
&= \frac{\lambda_n(Y_*)}{4\lambda_1(R + B^{\top} X_* B) \lambda_1^2(Y)} \langle \nabla f(K), \nabla f(K)\rangle,
\end{align*}
where the last inequality follows from Proposition~\ref{prop:linalg_facts} (part (d)).
It remains to lower bound $\lambda_1(Y)$ but this is straightforward since,
\begin{align*}
  Y = \sum_{j=0}^{\infty} ({A_K}^{\top})^j \, {\bf \Sigma} \, (A_K)^j \succeq {{\bf \Sigma}}.
\end{align*}
Hence,
\begin{align*}
  f(K)-f(K_*) \le \frac{\lambda_n(Y_*)}{4\lambda_1(R + B^{\top} X_* B) \lambda_1^2({\bf \Sigma})} \langle \nabla f(K), \nabla f(K)\rangle.
\end{align*}
    \end{proof}
    {
      We now provide an estimate of the gradient dominance coefficient,
      \begin{align}
      \tau \coloneqq \frac{\lambda_n(Y_*)}{4\lambda_1(R + B^{\top} X_* B) \lambda_1^2({\bf \Sigma})}. \label{tau}
      \end{align}
      This coefficient determines the linear convergence rate of gradient descent for LQR.
      \begin{proposition}
        \label{prop:gradient_dominant_bound}
        Over the sublevel set $S_{f(K_0)}$,
        \begin{align*}
          \tau \le \frac{f(K_0)}{4 \lambda_1(Q) \, \lambda_1(R) \, \lambda_1^2({\bf \Sigma})}.
          \end{align*}
        \end{proposition}
        \begin{proof}
          Essentially, we only need to estimate $\|Y_*\|$; we shall estimate $\Tr(Y_*)$ instead. We first observe that,
          \begin{align*}
            \Tr(Y_*) = \Tr( \sum_{j=0}^{\infty} (A_{K_*})^j \, {\bf \Sigma} \, (A_{K_*}^{\top})^j) = \Tr(\sum_{j=0}^{\infty} {\bf \Sigma} (A_{K_*}^{\top})^j A_{K_*}^j  ).
          \end{align*}
          Putting $Z = \sum_{j=0}^{\infty} (A_{K_*}^{\top})^j (A_{K_*})^j$, we note that $Z$ solves the Lyapunov equation,
          \begin{align*}
            A_{K_*}^{\top} Z A_{K_*} + I - Z = 0.
            \end{align*}
            Since $I \preceq Q/\lambda_1(Q)$, it follows by Proposition~\ref{prop:linalg_facts} that $Z \preceq X_*$, where $X_*$ solves the Lyapunov equation~\eqref{eq:optimal_value_matrix}. Hence,
            \begin{align*}
              \Tr(Y_*) \le \frac{1}{\lambda_1(Q)}\Tr({\bf \Sigma} X_*) = \frac{1}{\lambda_1(Q)} f(K_*) \le \frac{1}{\lambda_1(Q)} f(K_0).
              \end{align*}
         \end{proof}

    }
    \section{Gradient Flow on $\ca S$}
    \label{sec:gf}
    In this section, we show that the LQR cost function~(\ref{f(k)}) gives rise to a well-posed gradient flow,
    \begin{align}
      \label{eq:gradient_flow}
      \dot{K}_t = -\nabla f(K_t).
    \end{align}Let us first observe that~\eqref{eq:gradient_flow} admits a unique solution for all time $t$.
\begin{lemma}
  For every $K_0 \in \ca S$ and $t_0 \in \bb R$, there
  exists a unique solution $K_t \in C^{\infty}(\bb R, \ca S)$ for the initial value problem,
\begin{align}
  \label{eq:ivp}
  \begin{cases}
  \dot{K}(t) = -\nabla f(K), \\
K(t_0)= K_0.
  \end{cases}
\end{align}
\end{lemma}
\begin{proof}
  Note that $K \mapsto 2 \left(RK-B^{\top}X(A-BK)\right)Y$ is $C^{\infty}$ smooth.
  The statement now follows from Corollary~\ref{cor:compact} and Proposition $3.7$ in~\cite{helmke2012optimization}.
\end{proof}
We next show that the unique trajectory of~\eqref{eq:gradient_flow} is in fact exponentially stable; without loss of generality, we assume that $t_0= 0$.
\begin{theorem}
  \label{thrm:exponential_trajectory}
For $K_0 \in \ca S$, denote by $K_t$ as the solution of~\eqref{eq:ivp}. Then the trajectory $K_t$ is globally exponentially stable in the sense of Lyapunov, i.e., 
\begin{align*}
  \|K_t - K_*\|_F^2 \le c e^{-\alpha t} \|K_0 - K_*\|_F^2,
\end{align*}
where $\alpha,c \in \bb R_+$ are constants determined by the LQR parameters $A, B, Q, R$ and initial condition $K_0$.
\end{theorem}
To prove this result, we first observe that the Lyapunov functional $V(K_t) \coloneqq f(K_t) - f(K_*)$ converges exponentially to the origin.
\begin{lemma}
  \label{lemma:convergence_energy}
  For $K_0 \in \ca S$, denote by $K_t$ as the solution of~\eqref{eq:ivp}. Then
\begin{align*}
  f(K_t) - f(K_*) \le e^{-\alpha t} \left(f(K_0)-f(K_*)\right),
\end{align*}
where $\alpha \in \bb R_+$ is constant determined by system parameters $A,B,Q,R$ and $K_0$.
\end{lemma}
\begin{proof}
  Putting
\begin{align*}
  \kappa = \frac{\lambda_n(Y_*)}{4 \lambda_1(R+B^{\top}X_*B)\lambda_1^2(Y)},
\end{align*}
recall we have proved in Lemma~\ref{lemma:gradient_dominant},
\begin{align*}
  f(K)-f(K_*) \le \kappa \langle \nabla f(K), \nabla f(K)\rangle.
\end{align*}
It then suffices to observe
\begin{align*}
  \dot{V}(K_t) = -\langle \nabla f(K_t), \nabla f(K_t)\rangle \le - \alpha V(K),
\end{align*}
with $\alpha = 1/\kappa$.
Hence,
 $f(K_t) - f(K_*) \le e^{-\alpha t} (f(K_0)-f(K_*)).$
\end{proof}
We are now ready to prove Theorem~\ref{thrm:exponential_trajectory}.
\begin{proof}
  Observe that the Lyapunov functional is smooth, positive definite and radially unbounded;\footnote{In control literature, this is sometimes referred to as weakly coercive; nevertheless, as shown here, this is equivalent to being coercive.} thus $K_t$ is globally asymptotic stable, i.e.,
  \begin{align*}
    \lim_{t \to \infty} K_t = K_*.
  \end{align*}
  Note that,
  \begin{align}
   \label{eq:trajectory_eq1}
   \langle K_t - K_*, \dot{K}_t\rangle &= \langle K_t-K_*, -\nabla f(K_t)\rangle = -\Tr \left((K_t-K_*)^{\top} \nabla f(K_t)\right)  \\
&= -\frac{1}{2} \Tr\left( (K_t-K_*)^{\top}\nabla f(K_t) \right.                                           +\left. (\nabla f(K_t))^{\top}(K_t-K_*)\right).\nonumber
 \end{align}
  Furthermore, 
  \begin{align*}
    {\bf U} &\coloneqq (K_t-K_*)^{\top} \nabla f(K_t) +\nabla f(K_t)^{\top}(K_t-K_*) \\
    &\preceq \beta \, \nabla f(K_t)^{\top} \nabla f(K_t) + \frac{1}{\beta} (K_t-K_*)^{\top} (K_t-K_*)
  \end{align*}
  for every $\beta > 0$ and thus,
  \begin{align*}
    \Tr({\bf U}) & \le  \beta \Tr \left([\nabla f(K_t)]^{\top} \nabla f(K_t)\right)  + \frac{1}{\beta}\Tr ( (K_t-K_*)^{\top}(K_t-K_*) ).
  \end{align*}
By~\eqref{eq:difference}, we have,
 \begin{align*}
   {\bf V}  & \coloneqq  (K_t-K_*)^{\top} (RK-B^{\top}XA_K) + (RK_t-B^{\top}XA_K)^{\top} (K_t-K_*) \\
 &  \succeq (A_K-A_{K_*})^{\top}X(A_K-A_{K_*}) + (K-K_*)^{\top} R(K-K_*) \succeq 0,
 \end{align*}
 since $X_t-X_* \succeq 0$. Moreover since $\lambda_1(Y) \ge 1$,
 \begin{align*}
   \Tr({\bf U}) \ge 2\lambda_1(Y) \Tr({\bf V}) \ge 2\lambda_1(R) \Tr \left((K_t-K_*)^{\top}(K_t-K_*)\right). 
 \end{align*}
 Thereby,
 \begin{align*}
   (2\lambda_1(R) -  \frac{1}{\beta}) \Tr ((K_t-K_*)^{\top}(K_t-K_*))    \le \beta \Tr ([\nabla f(K_t)]^{\top}\nabla f(K_t)).
 \end{align*}
 Picking $2\lambda_1(R) > {1}/{\beta}$, we have
 \begin{align*}
   \Tr({\bf U}) & \le (\beta + \frac{1}{2\lambda_1(R)-\frac{1}{\beta}}) \Tr (\nabla f(K_t)^{\top} \nabla f(K_t) ) \\
&   \eqqcolon c' \Tr ([\nabla f(K_t)]^{\top} \nabla f(K_t) ).
 \end{align*}
 It now follows that,
 \begin{align}
   \label{eq:trajectory_eq1}
   \langle K_t-K_*, \dot{K}_t\rangle &= -\langle K_t-K_*, \nabla f(K_t)\rangle \\
                                     &\ge - \frac{c'}{2} \langle \nabla f(K_t), \nabla f(K_t)\rangle \nonumber \\
                                     &= -\frac{c'}{2} \langle \nabla f(K_t), \dot{K}_t\rangle. \nonumber
 \end{align}
Integrating both sides of~\eqref{eq:trajectory_eq1},
 \begin{align*}
   & \int_{t}^\infty \langle K_t - K_*, \dot{K}_t\rangle \ge -\frac{c'}{2}\int_t^{\infty} \langle \dot{K}_t, \nabla f(K_t)\rangle,
    \end{align*}
   implies that
     \begin{align*}
-\frac{1}{2} \|K_t-K_*\|^2 \ge - \frac{c'}{2}  V(K_t),
 \end{align*}
 and consequently, 
  \begin{align*}
\|K_t - K_*\|_F^2 \le c' V(K_t) \le c' e^{-\alpha t} (f(K_0)-f(K_*)).
 \end{align*}
 Putting $c = c' (f(K_0)-f(K_*))/(\|K_0-K_*\|_F^2)$ now completes the proof.
\end{proof}

{
  \subsection{Discretization of Gradient Flow}
  \label{sec:gd}
  In this section, we examine the discretization of the gradient flow~\eqref{eq:gradient_flow}. As we have observed in Lemma~\ref{lemma:convergence_energy} and Theorem~\ref{thrm:exponential_trajectory}, both the energy functional and the trajectory of this flow converge exponentially to their respective global minimum. Ideally, a gradient descent algorithm converges linearly for the function values as well as the iterates. In this direction, the forward Euler discretization of the {gradient flow} yields,
\begin{align}
  \label{eq:gradient_descent}
  K_{j+1} = K_j - \eta_j \nabla f(K_j),
\end{align}
where $\eta_j$ is a nonnegative stepsize to be determined. The stepsize (or learning rate) should reflect two principles during the iterative process: (1) stay stabilizing and (2) sufficiently decrease the function value. In following, we shall see that the gradient dominated property leads to a stepsize that results in a sufficient decrease in the function values while the coerciveness guarantees that the acquired feedback gain is stabilizing. To begin, we observe that if $K_{j+1} = K_j - \eta_j \nabla f(K_j)$, provided that $K_j$ and $K_{j+1}$ are both stabilizing, the difference of the value matrix $X_{j+1} - X_j$ can be characterized as follows.\footnote{This relationship is used in~\cite{fazel2018global}.}
\begin{lemma}
  \label{lemma:value_difference}
  If $K_{j+1} = K_j - \eta_j \nabla f(K_j)$ and $K_j, K_{j+1}$ are both stabilizing, then $Z \coloneqq X_{j+1}-X_j$ solves the Lyapunov matrix equation,
  \begin{align*}
  A_{K_{j+1}} Z A_{K_{j+1}}^{\top} - Z - 2\eta_j Y_j^{\top} M_{j}^{\top} M_j - 2\eta_j M_j^{\top} M_j Y_j  + Y_j^{\top} M_j^{\top}(4\eta_j^2 R + 4 \eta_j^2 B^{\top} X_j B) M_j Y_j = 0.
    \end{align*}
  \end{lemma}
  \begin{proof}
  Following the same strategy used in the proof of  Lemma~\ref{lemma:gradient_dominant}, namely, taking the difference of the corresponding Lyapunov matrix equations, we observe that,
      \begin{equation} \label{eq:value_difference}
      \begin{split}
        &{A}_{K_{j+1}}^{\top} (X_{j+1}-X_j) {A}_{K_{j+1}} - (X_{j+1}-X_{j}) + (K_{j+1}-K_{j})^{\top}(RK_j-B^{\top} X_j A_{K_j}) \\
&+ (K_j^{\top}R - A_{K_j}^{\top} X_j B)(K_{j+1}-K_{j}) + (K_{j+1}-K_{j})^{\top} R (K_{j+1}-K_j)\\
&  + (A_{K_{j+1}}-A_{K_j})^{\top} X (A_{K_{j+1}}-{A}_{K_{j}}) = 0. 
      \end{split}
   \end{equation}
   Substituting $K_{j+1} - K_j = - 2\eta_j (RK_j - B^{\top} X_j B)Y_j$, we then have,
\begin{align*}
  A_{K_{j+1}} Z A_{K_{j+1}}^{\top} - Z - 2\eta_j Y_j^{\top} M_{j}^{\top} M_j - 2\eta_j M_j^{\top} M_j Y_j  + Y_j^{\top} M_j^{\top}(4\eta_j^2 R + 4 \eta_j^2 B^{\top} X_j B) M_j Y_j = 0.
\end{align*}
    \end{proof}
    We now observe that with appropriately chosen $\eta_j$, we can guarantee a sufficient decrease in the function value while ensuring stabilization (for the analogous result in~\cite{fazel2018global}, see the third part of Theorem 7 and Lemma 24).
\begin{lemma}
  \label{lemma:gd_function_decrease}
Consider the sequence $\{K_j\}$ generated by~\eqref{eq:gradient_descent} with stepsize $\eta_j$. 
  Denote by $\{X_j\}$ the corresponding Lyapunov matrix solutions with respect to $\{K_j\}$. When 
  \begin{align}
    \label{eq:stepsize_gradient}
\eta_j < \sqrt{\frac{1}{c_j} + \frac{b_j^2}{4c_j^2}} - \frac{b_j}{2c_j},
\end{align}
where
\begin{align*}
  b_j = \lambda_n(R+B^{\top}X_j B) \frac{f(K_j)}{\lambda_1(Q)} + \frac{4f(K_j)\|BM_KY_j\|_2\lambda_n(Y)}{\lambda_1(Q)\lambda_1({\bf \Sigma})}, \qquad c_j = \lambda_n(R + B^{\top}X_j B) \frac{4\|BM_K Y_j\|_2 \lambda_n(Y)f(K_j)}{\lambda_1(Q)} ,
\end{align*}
 then $\{K_j\}$ is stabilizing for every $j \ge 0$. In particular,
\begin{align*}
  f(K_{j+1}) -f(K_j) &\le 4\Tr(Y_j M_j^{\top} M_j Y_j)(\eta_j - b_j \eta_j^2 - c_j \eta_j^3).
\end{align*}
\end{lemma}
{
Before presenting the proof of this result, we shall first outline its basic idea. The crucial property we shall leverage is the compactness of the sublevel sets, analogous to devising the stepsize. If we start at a stabilizing control gain $K$ where the gradient does not vanish and consider the ray of $\{K - \eta \nabla f(K): \eta \ge 0\}$, by compactness of the sublevel set, there is some $\zeta$ for which $f(K') = f(K)$, where $K' \coloneqq K-\zeta \nabla f(K)$ (See Figure~\ref{fig:level_curve}). What we shall demonstrate is that with the stepsize $\eta_j$ given in the Lemma, if $K_{j+1}$ stays in the {compact} sublevel set, then $K_{j+1}$ must stay in the interior of the sublevel set, namely, $f(K_{j+1}) < f(K_j)$. We then proceed to examine two alternatives: (1) $K_{j+1}$ is not stabilizing, or (2) $K_{j+1}$ is stabilizing but $f(K_{j+1}) > f(K_j)$; either alternative would lead to a contradiction.
\begin{figure}[h!]
  \centering
\begin{tikzpicture}[scale=0.6]
    \path[font={\tiny}]
        (0 , 0)   coordinate (A1) 
        (1  , 1)   coordinate (A2)
        (3   , 1)   coordinate (A3)
        (4   , 3)   coordinate (A4)
        (5, 3) coordinate (A5)
        (4, -1) coordinate (A6)
        (2, -1.5) coordinate (A7)
        (-1, 0) coordinate (B1) 
        (2, 2) coordinate (B2)
        (3, 2) coordinate (B3)
        (4, 4) coordinate (B4)
        (6, 4) coordinate (B5)
        (5, -2) coordinate (B6)
        (3, -2.5) coordinate (B7)
        (2, 0) coordinate (C1)
        (2.5, 0.5) coordinate (C2)
        (3, 0.8) coordinate (C3)
        (4, 1.5) coordinate (C4)
        (3.5, 0) coordinate (C5)
        (3, -1) coordinate (C6)
        (2, -1) coordinate (C7)
        ($(B6)!(B7)!(B1)!0.5!(B7)$) coordinate (AA)
        ($(B7)!16.5!(AA)$) coordinate (BB)
    ;
\draw[black, name path=curve 1] plot [smooth cycle] coordinates {(B1) (B2) (B3) (B4) (B5) (B6) (B7)};
\draw[black] plot [smooth cycle] coordinates {(A1) (A2) (A3) (A4) (A5) (A6) (A7)};
\draw[black] plot [smooth cycle] coordinates {(C1) (C2) (C3) (C4) (C5) (C6) (C7)};
\draw[name path=normal line, red, ->, add=0 and 15] (B7) to (AA);
\fill[red,name intersections={of=curve 1 and normal line,total=\t}]
    \foreach \s in {1,...,\t}{(intersection-\s) circle (2pt)};
\draw node[anchor=south east] at (intersection-2) {$K$};
\draw node[anchor=south east] at (intersection-1) {$K-\zeta \nabla f(K)$};
\draw node at (BB) {$K-\eta \nabla f(K)$};
\end{tikzpicture}
\caption{Gradient descent interacting with the level curves of $f$ (\ref{f(k)}).}  \label{fig:level_curve}
  \end{figure}
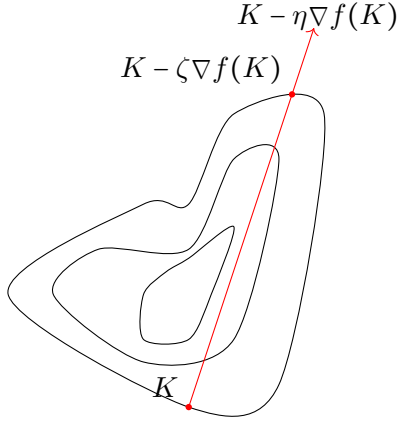
}
\begin{proof}
 Suppose that the sequence generated by the choice of $\eta_j$ is in fact stabilizing (to be proved subsequently!). This is crucial in our analysis as we use the Lyapunov matrix equation for the closed loop system, admitting a solution when $K_j$ is stabilizing; without this assumption, the matrix $X_j$ is not well-defined. By Lemma~\ref{lemma:value_difference}, we have,
  \begin{align*}
    f(K_{j+1}) - f(K_j) &= \Tr((X_{j+1}-X_j) {\bf \Sigma}) \\
    &= \Tr\left( Y_{j+1} \left(  - 2\eta_j Y_j^{\top} M_{j}^{\top} M_j - 2\eta_j M_j^{\top} M_j Y_j + Y_j^{\top} M_j^{\top}(4\eta_j^2 R + 4 \eta_j^2 B^{\top} X_j B) M_j Y_j  \right) \right) \\
    &\le 4 \eta_j \Tr \left( M_j^{\top}M_j (  -Y_j Y_{j+1} + \eta_j \lambda_n(R+B^{\top }X_j B) Y_j Y_{j+1} Y_j)\right).
    \end{align*}
In order to determine a stepsize $\eta_j$ such that $f(K_{j+1}) < f(K_j)$, we consider a univariate function,\footnote{Note that since  the products $Y_j Y_{j+1}$ and $M_j^{\top} M_j Y_j$ are not generally symmetric, the inequalities in Proposition~\ref{prop:linalg_facts} are not  necessarily applicable.}
\begin{align*}
  g(\eta) = \Tr \left( M^{\top} M ( Y Y(\eta) - \eta a Y Y(\eta) Y)\right),
\end{align*} 
where $M = RK-B^{\top} X A_K$, $a \coloneqq \lambda_n(R+B^{\top}XB)$, and $Y(\eta)$ is the solution of the matrix equation,\footnote{The function is not defined for every $\eta > 0$ but only for an interval for which $K-\eta 2MY$ is stabilizing.}
\begin{align*}
  Y(\eta) = (A-B(K-\eta 2MY )) Y(\eta) (A-B(K-\eta 2 MY)) + {\bf \Sigma}.
\end{align*}
Note that in defining the function $g$ we have dropped the indices as this function is used to determine stepsize for {\em every} iteration. Assuming that the choice of $\eta$ ensures staying in the sublevel set of $f(K)$, i.e., $f(K-\eta 2 MY) \le f(K)$,
we now examine whether $g(\eta) > 0$. By the Mean Value Theorem, we have
\begin{align*}
  g(\eta) = g(0) + \eta g'(\theta),
\end{align*}
for some $\theta \in [0, \eta]$; first note that,
\begin{align*}
  g'(\theta) = \Tr(M^{\top} M ( Y Y'(\theta) - a Y Y(\theta) Y - \theta a Y Y'(\theta) Y)),
\end{align*}
and hence,
\begin{align*}
  g(\eta) &= \Tr( M^{\top} M (Y^2 + \eta Y Y'(\theta) - \eta a Y Y(\theta) Y - \eta \theta a Y Y'(\theta ) Y)) \\
  &= \Tr(Y M^{\top} M Y( I - \eta a Y(\theta) - \eta^2 a Y'(\theta))) + \eta \Tr(M^{\top}MY Y'(\theta) Y^{-1}Y) \\
  &\ge \Tr( Y M^{\top} M Y ) (1 - \eta a \lambda_n(Y(\theta)) - \eta^2 a \|Y'(\theta)\|2 - \eta  \|Y'(\theta)\|_2 \|Y^{-1}\|_2),
\end{align*}
where the last inequality follows from Von Neumann's trace inequality~\cite{horn2012matrix}.\footnote{An explicit form of the inequality we use here can be found in~\cite{mori1988comments}.} 
Noting that $\|Y^{-1}\|_2 = 1/\lambda_1(Y)$, ensuring that $g(\eta) > 0$ reduces to characterizing $\eta$ for which,
\begin{align*}
  1 - \eta a \lambda_n(Y(\theta)) - \eta^2 a \|Y'(\theta)\|_2 - \eta \frac{\|Y'(\theta)\|_2}{\lambda_1(Y)} > 0.
\end{align*}
The largest eigenvalue of $Y(\theta)$ and largest singular value of $Y'(\theta)$ over the sublevel set $\{K': f(K') \le f(K)\}$ can be bounded as,
\begin{align*}
  \lambda_n(Y(\theta)) \le \frac{f(K)}{\lambda_1(Q)}, \quad \|Y'(\theta)\|_2 \le  \frac{4\|BMY\|_2 \lambda_n(Y) f(K)}{\lambda_1(Q)};
\end{align*}
the proof of the latter inequality is deferred to Appendix~\ref{appendix:bound_Y_prime}.
Note
\begin{align*}
  b &\coloneqq a \frac{f(K)}{\lambda_1(Q)} + \frac{4f(K)\|BM_KY\|_2 \lambda_n(Y)}{\lambda_1(Q)\lambda_1({\bf \Sigma})} \ge a \lambda_n(Y(\theta)) + \frac{\|Y'(\theta)\|}{\lambda_1(Y)}, \\
  c &\coloneqq a \frac{4\|BM_K Y\|_2\lambda_n(Y) f(K)}{\lambda_1^2(Q)} \ge a\|Y'(\theta)\|_2;
\end{align*}
it now suffices to determine $\eta$ such that $1 - b\eta - c \eta^2 > 0 $. As such, we require that,
\begin{align*}
  \eta < \sqrt{\frac{1}{c} + \frac{b^2}{4c^2}} - \frac{b}{2c}.
\end{align*}
It remains to show that if $\eta_j$ is chosen as above, our two opening assumptions are valid: (1) the sequence $\{K_j\}$ is stabilizing, and (2) $K_{j+1}$ remains in the sublevel set of $f(K_j)$. We prove these by contradiction. First, note that we can not have $K_{j+1}$ be stabilizing while $K_{j+1}\notin S_{f(K_j)}$. Suppose that this is the case. The sublevel set $S_{K_j} \coloneqq \{K: f(K) \le f(K_j)\}$ is compact and the ray $\{K_j - \zeta \nabla f(K_j) : \zeta \ge 0\}$ intersects the boundary of $S_{K_j}$ for some $\zeta > 0$; suppose that $K' = K_j - \zeta' \nabla f(K_j) \in \partial S_{K_j}$, where $\zeta'$ is the smallest positive real number for which this intersection occurs, i.e., the first time the ray intersects the boundary. 
It is clear $\zeta'$ must be greater than $\eta_j$ as otherwise we would have $\zeta' < \eta_j$ and $f(K_j - \zeta' \nabla f(K_j)) < f(K_j)$, a contradiction\footnote{Note what we proved above is: if a stepsize is strictly smaller than $\eta_j$, the function value is strictly decreasing if the gradient is not vanishing.}.  Now we prove that $K_j$ is stabilizing. If not, we must have $$[0, \eta_j) \subseteq [0, \zeta'],$$ since otherwise, there exists $s' < \eta_j$ such that $s' = \zeta'$ and $f(K') = f(K_j)$, which would also contradict the inequality $f(K') < f(K_j)$.
\end{proof}
%
\begin{theorem}
  \label{thrm:gd_linear}
  Putting $d_j = \max(b_j, c_j)$ where $b_j, c_j$ are given in Lemma~\ref{lemma:gd_function_decrease},
if $\eta_j = \sqrt{\frac{1}{3d_j} + \frac{1}{9}} - \frac{1}{3}$, we have,
\begin{align*}
  f(K_{j}) - f(K_*) \le q^j (f(K_0)-f(K_*)), \text{ and } \|K_j-K_*\|_F \le c_1 q^{j/2},
\end{align*}
where $q \in (0, 1)$ and $c_1 > 0$
are constants.
\end{theorem}
\begin{remark}
$\eta_j$ is acquired by noting that according to Lemma~\ref{lemma:gd_function_decrease},
\begin{align*}
  f(K_{j})-f(K_{j+1}) \ge 4\Tr(Y_j M_j^{\top} M_j Y_j)(\eta_j - d_j \eta_j^2 - d_j \eta_j^3).
  \end{align*}
  Maximizing $\eta_j-d_j \eta_j^2 - d_j \eta_j^3$ while esnuring $1-d_j \eta_j - d_j \eta_j^2 > 0$ yields the desired quantity. 
  \end{remark}
\begin{proof}
   Note the proposed stepsize rule satisfies $1-2 d_j \eta_j - 3 d_j \eta_j^2 =  0$. Putting $r_j = f(K_j)-f(K_*)$,
  we observe that with the chosen stepsize $\eta_j$,
\begin{align*}
  r_{j} - r_{j+1}  
                  &\ge 4\Tr(Y_j M_j^{\top}M_j Y_j) (d_j \eta_j^2 + 2d_j^2 \eta_j^3 ) \eqqcolon \nu_j r_j.
\end{align*}
It follows that,
\begin{align*}
  r_{j+1} \le (1-\nu_j) r_j \eqqcolon q_j r_j.
\end{align*}By Proposition~\ref{prop:stepsize_bound}, the proposed stepsize is bounded away from $0$, i.e., $\eta_j \ge \epsilon$ for some constant $\epsilon > 0$.
Hence, the sequence $\{q_j\}$ is upper bounded away from $1$\footnote{It is rather clear $d_j$ is lower bounded away from $0$. So $d_j \eta_j^2 + 2 d_j^2 \eta_j^3 > 0$.}, namely, for every $j$
\begin{align*}
  q_j \le q < 1.
  \end{align*}
Thereby, 
\begin{align*}
  f(K_{j}) - f(K_*) \le q^j \left(f(K_0)-f(K_*)\right).
\end{align*}
To show the convergence of the iterates, we first observe that,
\begin{align*}
  \|K_{j+1} - K_j\|_F^2 &= \eta_j^2 \|\nabla f(K_j)\|_F^2 \le \frac{\eta_j^2}{\tau} r_j \\
                      &\le \frac{\eta_j}{\tau} q^j r_0,
\end{align*}
with $\tau$ is as in (\ref{tau}).
It is clear the sequence $\{\eta_j\} \subseteq \bb R_+$ is upper bounded, denoting as $\mu$, namely $\mu \ge \eta_j$ for every $j$.
The sequence of iterates $\{K_j\}$ is thus Cauchy and converges to some stationary point; however, there is only one stationary point $K_*$. This implies that $\lim_{j \to \infty} K_j = K_*$ and hence,
\begin{align*}
  \|K_j - K_*\|_F &= \lim_{n \to \infty} \|K_j - K_n\|_F \le \sum_{j=n}^{\infty}\|K_{j+1}-K_{j}\|_F \\
  &\le  \sqrt{\frac{\mu}{\tau}} r_0 \sum_{j=n}^{\infty} q^{j/2} = \frac{\sqrt{\frac{\mu}{\tau}} r_0 }{1-\sqrt{q}} \, q^{j/2}.
\end{align*}
\end{proof}
    \begin{remark}
In our simulations, the linear rate is much better than what is estimated by the above result. \par
         \end{remark}
}
    {
      It is now straightforward to bound the number of iterations needed to reach $\varepsilon$-precision in terms of problem data.
      \begin{corollary}
        Suppose that $K_0 \in \ca S$ and the sequence of stabilizing gains $\{K_j\}$ with stepsize $\eta_j$ given in Theorem~\ref{thrm:gd_linear} has been generated. Then, for
        \begin{align*}
          N \ge \frac{1}{\log(q)} {\log(\frac{\varepsilon}{f(K_0) - f(K_*)})},
        \end{align*}
        we have,
\begin{align*}
  f(K_N) - f(K_*) \le \varepsilon.
\end{align*}
    \end{corollary}
\begin{remark}
  To obtain the iteration complexity solely in terms of problem data $(A, B, Q, R, {\bf \Sigma}, K_0)$, it suffices to note that $f(K_0) - f(K_*) \le f(K_0)$ and we may replace $f(K_0) - f(K_*)$ by $f(K_0)$ in the above estimates. \par
  We shall point out this complexity bound is very conservative as in determining stepsize, several crude bounds were used. Empirically, we observe that the actual convergence rate is faster than the one given here.
\end{remark}
    }

{
\section{Natural Gradient Flow on $\ca S$}
\label{sec:ngf}
If we inspect the proof of gradient dominated property (Lemma~\ref{lemma:gradient_dominant}) and the Lyapunov stability of the gradient system (Theorem~\ref{thrm:exponential_trajectory}), the positive definite matrix $Y$ does not affect the qualitative nature of these properties. Nevertheless, the matrix $Y$ introduces a constant factor in the corresponding upper bounds. In this section, we consider a family of gradient systems of the form,
\begin{align}
  \label{eq:ngf} \dot{K}_t &= -\nabla f(K) Y^{-\gamma} = -2(RK-B^{\top}XA_K)Y^{1-\gamma}, 
\end{align}
where $\gamma > 0$ is (real) scalar.\footnote{When $\gamma = 1$, this flow can be viewed as the continuous limit of the natural gradient descent as discussed in~\cite{fazel2018global}.}
As discussed subsequently,  such parameterized gradient system can achieve better convergence rate for different values of $\gamma$.
Viewing such a gradient flow in the context of a flow on a Riemannian manifold is particularly pertinent.\footnote{We will see that in our case, it is better to choose $\gamma$ other than $\gamma=1$.}
In fact, as $\ca S$ is open, it is a \emph{submanifold} in $\bb M_{m \times n}(\bb R)$. We first observe that the inner product induced by $Y^{\gamma}$, i.e., $\langle M, N \rangle_{Y^{\gamma}} = \Tr(M^{\top} NY^{\gamma})$ is a well-defined Riemannian metric over $\ca S$. 
\begin{proposition}
  Over $\ca S$, the inner product $\langle \cdot, \cdot\rangle_{Y(K)^\gamma}$ induces a Riemannian metric.
\end{proposition}
\begin{proof}
  Note that $Y(K)$ is positive definite for every $K \in \ca S$.
  It suffices to show that $Y(K)$ varies smoothly with $K$. But this follows from,
  \begin{align*}
    \vect(Y) = (I \otimes I - A_K \otimes A_K)^{-1} \vect({\bf \Sigma}).
  \end{align*}
\end{proof}
We can thus view $\ca S$ as a Riemannian manifold with metric induced by $\langle \cdot, \cdot \rangle_{Y^{\gamma}}$; the function $f: \ca S \to \bb R$ is then a scalar-valued function defined on this manifold. Let us now consider the gradient of $f$, denoted by $\text{grad} f$, with respect to the Riemannian metric induced by $\langle \cdot, \cdot\rangle_{Y^{\gamma}}$ on $\ca S$\footnote{We will use standard notitions in Riemmaninan manifold theory~\cite{tu2017differential}. For example, $df$ will denote $1$-form and $\text{grad} f$ will denote the gradient with respect to a Riemannian metric. As we are working in Euclidean space, we implicitly identiy all tangent vectors by stardard isomorphism, i.e., $T_K {\ca S} \approx \bb M_{m \times n}(\bb R)$.}.
\begin{proposition}
  Over the Riemannian manifold $\left(\ca S, \langle \cdot, \cdot\rangle_{Y^{\gamma}}\right)$, $\text{grad} f = 2(RK-B^{\top} X A_K)Y^{1-\gamma}$.
\end{proposition}
\begin{proof}
  It suffices to note that,
  \begin{align*}
    df(K)[E] = 2\Tr(E^T (RK-2B^{\top} X A_K ) Y) = \langle E, 2(RK-2B^{\top}XA_K)Y^{1-\gamma}\rangle_{Y^{\gamma}}.
  \end{align*}
\end{proof}
Now the gradient flow of interest on this manifold is,
\begin{align*}
  \dot{K}_t = - \text{grad} f(K_t).
\end{align*}
First recall two inequalities that we encountered previously.
\begin{proposition}
  \label{prop:inequality_ngf}
  For $K \in \ca S$,
       \begin{align*}
  \text{a)} &  \;  f(K)-f(K_*) \le \frac{\lambda_n(Y_*)}{\lambda_1(R+B^{\top}X_*B)} \Tr(M^{\top} M),
    \\
      \text{b)} &  \; 0 \preceq (K-K_*)^{\top} M + M^{\top}(K-K_*) \preceq \frac{1}{\lambda_1(R+B^{\top}X_*B)}M^{\top} M,
      \end{align*}
      where $M = RK-B^{\top}XA_K$.
\end{proposition}
We observe that with respect to the Riemannian metric, the potential function decays at an exponential rate (compare the difference with the gradient flow in Lemma~\ref{lemma:convergence_energy}).
\begin{lemma}
  \label{lemma:convergence_energy_ngf}
  For $K_0 \in \ca S$, denote $K(t)$ as the solution of~\eqref{eq:ngf}. Then
\begin{align*}
  f(K_t) - f(K_*) \le e^{- r t} (f(K_0)-f(K_*)),
\end{align*}
where $r$ is a constant determined by the system parameters $A,B,Q,R$ and $K_0$.
\end{lemma}
\begin{proof}
  The proof proceeds similar to Lemma~\ref{lemma:convergence_energy}. We only need to note that with respect to the Riemannian metric,
  \begin{align*}
    \dot{V}(K_t) &= df(K_t) \dot{K}(t) = \langle \text{grad} f(K_t), \dot{K_t}\rangle_{Y^\gamma} = -4 \Tr(M^{\top} MY^{2-\gamma}).
  \end{align*}
 According to Proposition~\ref{prop:inequality_ngf}, we now have,
  \begin{align*}
    \dot{V}(K_t) \le -\frac{4 \lambda_1(R+B^{\top}X_*B)\lambda_1(Y^{2-\gamma})}{\lambda_n(Y_*) } V(K_t).
  \end{align*}
  \end{proof}
  \begin{remark}
  Lemma~\ref{lemma:convergence_energy} shows that the gradient descent~\eqref{eq:ivp} converges to the equilibrium point at an exponential rate ${4 \lambda_1(R+B^{\top}X_*B) \lambda_1^2(Y)}/{\lambda_n(Y_*)}$.
 Hence, the natural gradient flow~\eqref{eq:ngf} modifies the exponential convergence rate of the gradient descent algorithm to ${4 \lambda_1(R+B^{\top}X_*B)\lambda_1(Y^{2-\gamma})}/{\lambda_n(Y_*) }$, by a constant factor of ${\lambda_1(Y^{2-\gamma})}/{\lambda_1^2(Y)}$. This factor depends on the largest and smallest eigenvalues of the matrix $Y$. For example, if $\gamma \ge 3$, then $\lambda_1(Y^{2-\gamma}) = 1/\lambda_n(Y^{\gamma - 2})$.
  \end{remark}
Over the Riemannian manifold, the Lyapunov functional converges exponentially to the origin via the natural gradient flow, which leads to an exponentially stable trajectory.
\begin{theorem}
  Over $\left(\ca S, \langle \cdot, \cdot\rangle_{Y^{\gamma}}\right)$, for the natural gradient flow~\eqref{eq:ngf}, the energy functional $f(K_t)-f(K_*)$ converges exponentially to the origin. Moreover, the trajectory $K_t$ is exponentially stable in the sense of Lyapunov.
\end{theorem}
\begin{proof}
Over the Riemannian manifold $\left(\ca S, \langle \cdot, \cdot\rangle_{Y^{\gamma}}\right)$, we have,
\begin{align*}
    f(K_{t_1}) - f(K_{t_2}) &= -\int_{t_1}^{t_2} df(K_t) = -\int_{t_1}^{t_2} \langle \text{grad} f(K_t), \dot{K}(t) \rangle_{Y^{\gamma}} dt = \int_{t_1}^{t_2} \|\text{grad} f(K_t)\|_{Y^{\gamma}}^2 dt.
\end{align*}
So
\begin{align*}
  f(K_t) - f(K_*) &= \int_{t}^{\infty} \langle \text{grad} f, \text{grad} f \rangle_{Y^{\gamma}} dt \\
                  &=\int_t^{\infty} 4\Tr(M_t^{\top} M_t Y_t^{2-\gamma})dt \\
                  &\ge \lambda_1(R+B^{\top}X_* B)\int_t^{\infty} 8\Tr((K_t-K_*)^{\top} M_t Y_t^{2-\gamma})\\
  &=8d \int_t^{\infty} \Tr((K_t-K_*)^{\top} M_t Y_t) dt\\
                  &= -2d\|K_t-K_*\|^2 \big\vert_t^{\infty} \\
                  &= 2d \|K_t-K_*\|^2,
\end{align*}
where $d = \lambda_1(R+B^{\top}X_*B) \lambda_1(Y^{1-\gamma})$.
Hence,
  $\|K_t -K_*\|^2 \le \frac{1}{2d}V(K_t) \le \frac{1}{2d} e^{-r t} V(0).$
\end{proof}
\begin{remark} \label{remark:gf_vs_ngf}
  We note that the convergence rate of trajectory $K_t$ is dependent on $\lambda_1(Y)$ and $\lambda_n(Y)$. For example, when $\gamma=1$ and ${\bf \Sigma} = 2I$, then the natural gradient flow converges faster than the gradient flow since $\lambda_1(Y) > 1$. On the other hand, if $\gamma = 1$ and $\lambda_1(Y) < 1$, then gradient flow converges faster than natural gradient flow.\footnote{This can be done by an ${\bf \Sigma}$ that has a spectrum bounded by $1$.} Simulation results in~\S\ref{sec:simulation} show that this parameterized gradient flow offers a significant computational advantage for LQR.
\end{remark}
{We remark that in the particular case of $\gamma = 1$, the natural gradient flow has a favorable property with respect to the induced flow on the value matrix $X_t$. Consider again the flow, 
  \begin{align}
    \label{eq:ngf_1}
    \dot{K}_t = - 2(RK_t - B^{\top} XA_{K_t}),
    \end{align}
inducing the flow over the ``value'' matrix $X_t \coloneqq X(K_t)$ given by,
    \begin{align}
      \label{eq:ngf_cone}
      \dot{X}_t = \frac{d X_t}{d K} \dot{K}_t.
      \end{align}
      \begin{lemma}
        \label{lemma:mono_X_t}
        For $K_0 \in \ca S$, the gradient flow~\eqref{eq:ngf_1} induces a well-posed flow over the positive semidefinite cone $X_t$~\eqref{eq:ngf_cone}. Moreover, the trajectory $\{X_t\}$ is monotonically decreasing in Loewner ordering.
        \end{lemma}
        \begin{proof}
          The well-posedness follows from the well-posedness of $\{K_t\}$. To show that the trajectory is monotonically decreasing, it suffices to observe,
          \begin{align*}
            \dot{X}_t &= \frac{d X_t}{d K} \dot{K}_t 
                      = \sum_{j=0}^{\infty} \left(A_{K_t}^{\top} \right)^j \left( \dot{K}_t^{\top} M_t + M_t \dot{K}_t \right) \left(A_{K_t}^j \right)\\
                      &=-\sum_{j=0}^{\infty} \left(A_{K_t}^{\top} \right)^j \left( 2M_t^{\top} M_t + 2M_tM_t^{\top} \right) A_{K_t}^j 
 \preceq 0,
            \end{align*}
            where the second inequality follows from~\eqref{eq:X_derivative}.
          \end{proof}
      Note that this monotonicity does not hold in general for gradient flow: in this case the flow is dictated by ${\bf \Sigma}$ and along the trajectory, one can only guarantee that the function value $\Tr(X_t {\bf \Sigma})$ decreases.
}
{
    \subsection{Discretization of Natural Gradient Flow}
    \label{sec:ngd}
    In this section, we delve into the discretization of natural gradient flow; we shall only consider the case when $\gamma = 1$.\footnote{Other choices can be analyzed in a similar manner.} Specifically, we consider the gradient flow,
\begin{align*}
  \dot{K}_t = -2(RK_t - B^{\top} X A_K).
\end{align*}
The forward Euler discretization yields,
\begin{align}
  \label{eq:ngd}
  K_{j+1} = K_j - 2 \eta_j(RK_j - B^{\top} X_j A_{K_j}),
\end{align}
where $\eta_j$ is the stepsize to be determined.
In discretizing gradient flow, our guideline is to choose a stepsize such that the function value is sufficiently decreased while keeping iterates stabilizing. However, in natural gradient flow with $\gamma = 1$, we observe that by Lemma~\ref{lemma:mono_X_t}: if we follow the natural gradient flow, the value matrix is monotonic with respect to the semidefinite cone. This essentially means that taking a sufficiently small stepsize in the direction of the natural gradient would guarantee a decrease in the value of the Lyapunov matrix solution $X_{t + \delta} \preceq X_{\delta}$. The reader is also referred to~\cite{fazel2018global} (Lemma 15) where a similar stepsize for the natural gradient update has been derived).
\begin{lemma}
  \label{lemma:ngd_value_matrix_difference}
Consider the sequence $\{K_j\}$ generated by~\eqref{eq:ngd}. Denote by $\{X_j\}$ the corresponding Lyapunov matrix solution with respect to $K_j$.
  If $\eta_j \le {1}/({\lambda_n(R) + B^{\top} X_0 B})$, then $K_j$ is stabilizing for every $j \ge 0$ and $X_{j+1} \preceq X_j$. In particular,
  $Z \coloneqq X_{j+1}-X_j$ solves the Lyapunov matrix equation,
\begin{align*}
  A_{K_{j+1}} Z A_{K_{j+1}}^{\top} - Z +  M_{j}^{\top}(-4\eta_j I + 4\eta_j^2 (R +  B^{\top} X_j B) ) M_{{j}}  = 0,
\end{align*}
where $M_j = RK_j-B^{\top} X_j A_{K_j}$.
\end{lemma}
\begin{proof}
  The proof proceeds similar to Lemma~\ref{lemma:gd_function_decrease}.
  First, we suppose that the sequence generated by the choice of $\eta_j$ is in fact stabilizing (to be proved subsequently). By Lemma~\ref{lemma:value_difference},
      \begin{equation} \label{eq:value_difference}
      \begin{split}
        &{A}_{K_{j+1}}^{\top} (X_{j+1}-X_j) {A}_{K_{j+1}} - (X_{j+1}-X_{j}) + (K_{j+1}-K_{j})^{\top}(RK_j-B^{\top} X_j A_{K_j}) \\
&+ (K_j^{\top}R - A_{K_j}^{\top} X_j B)(K_{j+1}-K_{j}) + (K_{j+1}-K_{j})^{\top} R (K_{j+1}-K_j)\\
&  + (A_{K_{j+1}}-A_{K_j})^{\top} X (A_{K_{j+1}}-{A}_{K_{j}}) = 0. 
      \end{split}
   \end{equation}
   If $K_{j+1} - K_j = - 2\eta_j (RK_j - B^{\top} X_j B)$, then
\begin{align*}
        &{A}_{K_{j+1}}^{\top} (X_{j+1}-X_j) {A}_{K_{j+1}} - (X_{j+1}-X_{j}) + M^{\top}(-4\eta_j I + 4\eta_j^2 R + 4 \eta_j^2 B^{\top} X_j B) M = 0. 
\end{align*}
Hence, if $-4\eta_j I + 4\eta_j^2 R + 4 \eta_j^2 B^{\top} X_j B \preceq 0$, then $X_{j+1} \preceq X_j$. This can be guaranteed by choosing,
\begin{align*}
  \eta_j \le \frac{1}{\lambda_n(R + B^{\top} X_j B)}.
\end{align*}
It now remains to show that if $\eta_j$ is chosen as above, the sequence will be stabilizing. Suppose that $K_j$ is stabilizing. Note that the sublevel set $\ca S_{K_j} \coloneqq \{K: f(K) \le f(K_j)\}$ is compact and the ray $K_j - \zeta M_j$ intersects the boundary of $\ca S_{K_j}$ for some $\zeta=\zeta' >0$; suppose that   $K' = K_j - \zeta' M_j \in \partial \ca S_{K_j}$. But this implies that $$[0, \frac{1}{\lambda_n(R+B^{\top}X_j B)}] \subseteq [0, \zeta'],$$ since otherwise, there would exist $s' \le {1}/{\lambda_n(B^{\top} X_j B + R)}$ such that $s' = \zeta'$ and $f(K_j - s' M_j) = f(K_j)$, contradicting $f(K_j - s' M) < f(K_j)$.
\end{proof}
The problem of determining the optimal stepsize can be done by minimizing the expression,
\begin{align*}
  -4\eta_j I + 4\eta_j^2 (R + B^{\top} X_j B) \preceq 0,
\end{align*}
over the positive semidefinite cone. This is equivalent to minimizing,
\begin{align*}
  -4\eta_j + 4\eta_j^2 \left(\lambda_n(R+B^{\top} X_j B)\right),
\end{align*}
at $\eta_j \in [0, 1/\lambda_n(R+B^{\top}X_j B)]$. Obviously, the optimal stepsize should be $\eta_j ={1}/({2\lambda_n(R+B^{\top} X_j B)})$. With this choice of stepsize, the function value converges linearly to the optimal value function.
\begin{theorem}
  \label{thrm:ngd_convergence}
  If $\eta_j = {1}/({2\lambda_n(R+B^{\top} X_j B)})$, we have,
\begin{align*}
  f(K_{j}) - f(K_*) \le q_0^j (f(K_0)-f(K_*)), \text{ and } \|K_{j} - K_*\|_F \le c_2 q_0^{j/2}.
\end{align*}
where
$q_0 = (1- {4 \lambda_1(R)})/({\lambda_n(Y_*) \lambda_n(R+B^{\top} X_0 B)})$ and $c_2$ is some positive constant.
\end{theorem}
\begin{proof}
  Putting $r_j = f(K_j)-f(K_*)$, we observe that with the chosen $\eta_j$,
\begin{align}
  \label{eq:lower_bound_M_K}
  \begin{split}
  r_{j} - r_{j+1} = \Tr((X_{j} - X_{j+1}) {\bf \Sigma}) &\ge \Tr(\frac{1}{\lambda_n(R+B^{\top} X_j B)} M_j^{\top} M_j Y_{j+1}) \\
&\ge \frac{\|Y_{j+1}\|}{\lambda_n(R+B^{\top}X_jB)} \Tr(M_j^{\top} M_j)\\
 &\ge \frac{4 \lambda_1(R)}{\lambda_n(Y_*) \lambda_n(R+B^{\top} X_j B)} r_j.
 \end{split}
\end{align}
It thus follows that,
\begin{align*}
  r_{j+1} \le (1-\frac{4 \lambda_1(R)}{\lambda_n(Y_*) \lambda_n(R+B^{\top} X_j B)}) r_j \eqqcolon q_j r_j.
\end{align*}
Note that by the choice of stepsize, $\{X_j\}$ monotonically decreases over the positive semidefinite cone and thus $q_j \le q_0$  for $j \ge 1$, where,
\begin{align*}
  q_0 &= 1- \frac{4 \lambda_1(R)}{\lambda_n(Y_*) \lambda_n(R+B^{\top} X_0 B)} \\
&\le 1 - \frac{4 \lambda_1(Q) \lambda_1(R)}{f(K_0) \lambda_n(R+B^{\top}X_0B)};
\end{align*}
in the last inequality we have used the estimate $\|Y_*\| \le {f(K_0)}/{\lambda_1(Q)}$ in Proposition~\ref{prop:gradient_dominant_bound}.
Thereby, 
\begin{align*}
  f(K_{j}) - f(K_*) \le q_0^j \left(f(K_0)-f(K_*)\right).
\end{align*}
The proof to the convergence of the iterates is almost identical to the one in Theorem~\ref{thrm:gd_linear}
\end{proof}
\begin{remark}
 We note that the discretization of natural gradient flow can perform better than gradient descent. One can monitor the one step progression $r_j - r_{j+1}$ to confirm such a behavior. This is different from the continuous flows as if $\lambda_1(Y) > 1$, then gradient flow performs better than natural gradient flow.
\end{remark}
}
{
  \section{Quasi-Newton Flow on $\ca S$}
  \label{sec:qnf}
  In this section, we motivate a quasi-Newton flow over the set of stabilizing feedback gains (policy) $\ca S$.\footnote{The justification for calling this evolution a quasi-Newton flow becomes apparent subseqeuntly.} As observed previously, the Hessian of the LQR cost $f(K)$ is not positive definite everywhere. As such, there is no well-defined notion of (global) Newton iteration over policy space. However, examining Lemmas~\ref{lemma:value_difference} and \ref{lemma:ngd_value_matrix_difference} allows us to derive a local second-order approximation of the LQR cost under the Riemannian metric $Y$. With is metric, recall that the gradient of $f$ is,
  \begin{align*}
    \text{grad} f(K) = 2(RK-B^{\top} X A_K) \eqqcolon 2M_K.
    \end{align*}
    We now provide the second-order approximation of the cost function.\footnote{Lemma~\ref{lemma:second_order} can be considered as a slight extension of Lemma $6$ in~\cite{fazel2018global}. However, the emphasis in~\cite{fazel2018global} was on the asymptotic behavior of the first-order approximation; this setup was subsequently utilized for a different purpose in~\cite{fazel2018global}. For our purpose, it is important to prove that for the second-order approximation, the remainder of the approximation is $O(\|\Delta K\|^2)$.}
  \begin{lemma}
    \label{lemma:second_order}
  When $K$ and $K+\Delta K$ are both stabilizing for sufficiently small $\Delta K$,\footnote{By openness of $\ca S$, if $\Delta K$ is sufficiently small, $K+\Delta K$ is stabilizing provided that $K$ is.} then,
    \begin{align*}
      f(K+\Delta K) = f(K) + \langle \text{grad} f(K), \Delta K \rangle_Y + \langle \Delta K, (R+B^{\top} X B) (\Delta K)\rangle_Y + \ca R(\Delta K),
    \end{align*}
    where $\|\ca R(\Delta K)\|$, the remainder of the approximation, is $O(\|\Delta K\|^2)$.
  \end{lemma}
  \begin{proof}
    Suppose that $X_{K+\Delta K}$ and $X_K$ are the corresponding value matrices for $K+\Delta K$ and $K$, respectively.
    By Lemma~\ref{lemma:value_difference}, we have,
    \begin{align*}
      X_{K+\Delta K} - X_K = \sum_{j=0}^{\infty} (A_{K+\Delta K}^{\top})^j ( (\Delta K)^T M_K + M_K^{\top} (\Delta K) + (\Delta K)^{\top}(R+B^{\top} X B)(\Delta K)) (A_{K+\Delta K})^j.
      \end{align*}
      It then follows that,
      \begin{align}
        \label{eq:second_order} 
        \begin{split}
        f(K+\Delta K) - f(K) &= \Tr( (X_{K+\Delta K} - X_K) {\bf \Sigma}) \\
                             &= \Tr(Y_{K+\Delta K} (\Delta K)^{\top} M_K) + \Tr(Y_{K+\Delta K} (\Delta K)^T (R+B^{\top} X B) (\Delta K)),
                             \end{split}
        \end{align}
        where $Y_{K+\Delta K}$ solves the Lyapunov equation,
        $$A_{K+\Delta K} Y_{K+\Delta K} A_{K + \Delta K}^{\top} + {\bf \Sigma} - Y_{K + \Delta K} = 0;$$
        $Y_{K+\Delta K}$ can be written as $Y_{K+\Delta K} = \sum_{j=0}^{\infty} (A_{K+\Delta K})^j \, {\bf \Sigma } \, (A_{K+\Delta K}^{\top})^j$. Note that if we expand the right-hand side of this last expression, we may alternatively write,
        \begin{align*}
          Y_{K+\Delta K} = Y + \ca R_Y(\Delta K),
        \end{align*}
        where $\ca R_Y(\Delta K)$ is the remainder term and consists of polynomials in $\Delta K$ with smallest degree $1$.
            Substituting the above equation in~\eqref{eq:second_order}, we have
            \begin{align*}
              f(K+\Delta K) &= f(K) + 2\Tr(Y_K (\Delta K)^{\top} M_K) + \Tr(Y_K(\Delta K)^{\top} (R+B^{\top} X B) (\Delta K) )  + \ca R(\Delta K)\\
                            &= f(K) + \langle \Delta K, \text{grad} f(K)\rangle_Y + \langle \Delta K, (R+B^{\top} X B) (\Delta K)\rangle_Y + \ca R(\Delta K);
            \end{align*}
           it is clear that $\ca R(\Delta K)$ consists of polynomials in $\Delta K$ with smallest degree $(\Delta K)^2$. 
    \end{proof}
    Lemma~\ref{lemma:second_order} essentially states that we have a somewhat ``good'' local second-order approximation of $f(K)$ with respect to the Riemannian metric $Y$. We may now devise a flow to minimize $f(K)$ by minimizing this second-order approximation, namely,
  \begin{align*}
    \dot{K}_t = (R+B^{\top} X_t B)^{-1} \text{grad} f(K_t) = (R+B^{\top} X_t B)^{-1}(RK_t-B^{\top} X_tA_K) = K_t - (R+B^{\top} X_t B)^{-1} B^{\top}X_t A.
  \end{align*}
 The analysis presented in \S\ref{sec:gf} and \S\ref{sec:ngf} allow us to obtain a streamlined proof of the convergence of this flow; as such, we omit the proof.
\subsection{Discretization of Quasi-Newton Flow}
\label{sec:qn_iteration}
The quasi-Newton flow over $\ca S$ has interesting consequences in terms of its discretization: the forward Euler leads to the iterative procedure 
\begin{align}
  \label{eq:qnd}
K_{j+1} = K_j - \eta_j(R+B^{\top}X_j B)^{-1} \text{grad} f(K_j)
\end{align}
 with stepsize $\eta_j$ to be determined; we shall show that with constant stepsize $\eta = \frac{1}{2}$, both the function value and the iterates will converge quadratically to the optima. 
\begin{remark}
  The update is consistent with the Gauss-Newton updates proposed in~\cite{fazel2018global}. We have chosen to refer to this update as quasi-Newton in this paper as it is obtained by minimizing a local second-order approximation of the LQR cost at each iteration.
\end{remark}
We first observe that if $\eta \le 1$, the corresponding sequence of value matrices $\{X_j\}$ is monotonically decreasing over the positive semidefinite cone.
        \begin{lemma}
  \label{lemma:newton_value_matrix_difference}
Consider the sequence $\{K_j\}$ generated by~\eqref{eq:qnd}. Denote by $\{X_j\}$ the corresponding Lyapunov matrix solution with respect to $K_j$.
  If $\eta_j < 1$, then $K_j$ is stabilizing for every $j \ge 0$ and $X_{j+1} \preceq X_j$. In particular
  $Z \coloneqq X_{j+1}-X_j \preceq 0$ solves the Lyapunov matrix equation,
\begin{align*}
  A_{K_{j+1}} Z A_{K_{j+1}}^{\top} - Z +  (-4 \eta_j + 4 \eta_j^2)M_{j}^{\top} (R +  B^{\top} X_j B)^{-1} M_{{j}}  = 0,
\end{align*}
where $M_j = RK_j-B^{\top} X_j A_{K_j}$.
\end{lemma}
\begin{proof}
  Suppose that with $\eta_j < 1$, the sequence generated by~\eqref{eq:qnd} are all stabilizing.\footnote{Similar to the proof to Lemma~\ref{lemma:ngd_value_matrix_difference}, we need this assumption to make sense of defining the corresponding value matrix sequence $\{X_j\}$.}
  Substituting the update rule~\eqref{eq:qnd} in~\eqref{eq:value_difference} yields,
  \begin{align*}
    A_{K+{j+1}}^{\top} (X_{j+1} - X_j) A_{K_{j+1}} - (X_{j+1} - X_j) + (-4 \eta_j + 4\eta_j^2 )M^{\top} (R+B^{\top} X_j B)^{-1} M = 0.
  \end{align*}
  It is now clear if $\eta_j < 1$, then $X_{j+1}-X_j \preceq 0$. To show the choice of $\eta_j$ guaranteeing the stability of $A-BK_j$, we may follow almost the same argument as in the proofs of Lemmas~\ref{lemma:value_difference} and~\ref{lemma:ngd_value_matrix_difference}. 
    \end{proof}
The optimal stepsize for the quasi-Newton iteration is obtained by minimizing the quantity $-4 \eta + 4 \eta^2$. As such, the optimal stepsize is $\eta_j = 1/2$ for every $j$. The corresponding update is then equivalent to,
  \begin{align}
    \label{eq:quasi-newton}
    K_{j+1} &= K_j - \frac{1}{2} (R+B^{\top}X_jB)^{-1} 2(RK_j - B^{\top} X_j (A-BK_j)) \\
            &= K_j - K_j + (R+B^{\top}X_j B)^{-1} B^{\top} X_j A \nonumber \\
            &= (R+B^{\top} X_j B)^{-1} B^{\top} X_j A. \nonumber
    \end{align}
    \begin{remark}
      With the optimal choice of stepsize as $\eta = 1/2$, the quasi-Newton over $K$ coincides with the Hewer' algorithm~\cite{Hewer1971TAC}, obtained by considering the Newton iteration over the ARE. We have thus provided an alternative point view of this algorithm: the algorithm can be obtained directly over the policy space even without the ARE.
      \end{remark}
    \begin{theorem}
      \label{thrm:qn_convergence}
      With stepsize $\eta = 1/2$, the update~\eqref{eq:quasi-newton} converges to the global minimum at a Q-quadratic rate. Namely, there exists constants $c > 0, c_3 > 0$, such that,
      \begin{align*}
        f(K_j) - f(K_*) \le c (f(K_{j-1}) - f(K_*))^2  \quad \text{ and } \quad \|K_j - K_*\|_F \le c_3 \|K_{j-1} - K_* \|_F^2.
      \end{align*}
    \end{theorem}
    \begin{proof}
      By Lemma~\ref{lemma:newton_value_matrix_difference} and noting $RK_* - B^{\top} X_* A_{K_*} = 0$, we have
      \begin{align}
        \label{eq:qn_myeq1}
        X_{j+1} - X_* = \sum_{\nu = 0}^{\infty} (A_*^{\top})^\nu (K_{j+1}-K_*)^{\top}(R+B^{\top} X_*B) (K_{j+1} - K_*) (A_*)^{\nu}.
      \end{align}
       It then follows that,
      \begin{align*}
        f(K_{j+1}) - f(K_*) &= \Tr( (X_{j+1} - X_*) {\bf \Sigma}) \\
                            &\le \|Y_*\|_2 \|R+B^{\top}X_* B\|_2 \Tr((K_{j+1} - K_*)^{\top} (K_{j+1} - K_*)).
      \end{align*}
     However, 
      \begin{align*}
        K_{j+1}-K_* &= (R+B^{\top} X_j B)^{-1} B^{\top} X_j A - (R+B^{\top} X_* B)^{-1} B X_* A \\
                    &= (R+B^{\top} X_* B)^{-1}(X_j - X_*) A +[(R+B^{\top} X_j B)^{-1} - (R+B^{\top} X_* B)^{-1} ]B^{\top} X_j A.
       \end{align*}
       Furthermore,
              \begin{align*}
         (R+B^{\top} X_j B)^{-1} - (R+B^{\top} X_* B)^{-1} &= (R+B^{\top}X_* B + B^{\top}(X_j - X_*) B)^{-1} - (R+B^{\top} X_* B)^{-1}\\
                                                           &=(I + B^{\top} (X_j - X_*) B)^{-1}(R+B^{\top} X_*B)^{-1} - (R+B^{\top} X_* B)^{-1} \\
                                                           &= (I - (I+B^{\top} (X_j - X_*) B)^{-1} B^{\top} (X_j - X_*) B )(R+B^{\top} X_* B)^{-1} \\
         &\quad -(R+B^{\top} X_* B)^{-1} \\
                                                           &=-(I + B^{\top} (X_j - X_*) B)^{-1} B^{\top}(X_j - X_*) B (R+B^{\top} X_* B)^{-1},
       \end{align*}
       where the third equality follows from $(I+N)^{-1} = I - (I+N)^{-1}N$.
      Hence,
       \begin{align*}
         \|K_{j+1} - K_*\|_F \le c' \|X_{j}-X_*\|_F,
        \end{align*}
        where $c'$ is given by
        \begin{align*}
          \|(R+B^{\top} X_* B)^{-1}\|_F \|A\|_F + \|B\|_F^2 \|(R+B^{\top}X_* B)^{-1}\|_F \|B\|_F \|X_0\|_F \|A\|_F.
          \end{align*}
          Consequently,
          \begin{align*}
            f(K_{j+1}) - f(K_*) &\le c'\|Y_*\|_2 \|R+B^{\top}X_* B\|_2 \|X_j - X_*\|_F^2 \\
                                &\le c'\|Y_*\|_2 \|R+B^{\top}X_* B\|_2 \frac{1}{\lambda_1^2 ({\bf \Sigma})} \Tr \left(\left((X_j - X_*){\bf \Sigma}\right)^2 \right)\\
                                &\le c'\|Y_*\|_2 \|R+B^{\top}X_* B\|_2 \frac{1}{\lambda_1^2 ({\bf \Sigma})} \left( \Tr\left((X_j - X_*) {\bf \Sigma}\right)\right)^2 \\
            &\eqqcolon c \left(f(K_j)-f(K_*)\right)^2.
            \end{align*}
            To establish the quadratic convergence of iterates, putting $S_* = \sum_{\nu=0}^{\infty} (A_*^{\top})^{\nu} (A_*)^{\nu}$, we observe by Proposition~\ref{prop:linalg_facts} and equation~\eqref{eq:qn_myeq1}
\begin{align*}
  \Tr(X_{j+1} - X_*) &\ge \lambda_1(R+B^{\top}X_* B) \lambda_1(S_*) \|K_{j+1} - K_*\|_F^2, \\
  \Tr(X_{j+1} - X_*) &\le \|R+B^{\top} X_* B\| \|S_*\| \|K_{j+1} - K_*\|_F^2.
\end{align*}
On the other hand, 
\begin{align*}
  \Tr((X_{j+1} - X_*) {\bf \Sigma}) &\le c(\Tr((X_j-X_*){\bf \Sigma}))^2 \\
&\le c (\|{\bf \Sigma}\|\Tr(X_j-X_*) )^2 \\
&\le c(\|{\bf \Sigma}\|  \|R+B^{\top}X_* B\| \|S_*\| \|K_{j} - K_*\|_F^2)^2.
\end{align*}
It follows,
\begin{align*}
  \|K_{j+1}-K_*\|_F^2 \le \frac{c \|{\bf \Sigma}\|^2 \|R+B^{\top} X_* B\|^2 \|S_*\|^2}{\lambda_1({\bf \Sigma}) \lambda_1(R+B^{\top} X_* B) \lambda_1(S_*)} \|K_j-K_*\|_F^4.
\end{align*}
    \end{proof}
}
    \section{Structured LQR Synthesis}
    \label{sec:structured_LQR}    In this section, we consider the problem of designing the feedback gain $K$ over a subspace. In particular, we are primary interested in feedback gains with a desired sparsity pattern. This is a natural formulation of distributed networked systems on an information-exchange graph $\ca G=(V, E)$. In such a setting, structured feedback gains reflecting the underlying interaction network are of particular interest. If the state of only a subset of agents is accessible for control implementation, the feedback gain must have a zero pattern that is compatible with this accessibility requirement, i.e., $K_{ij} = 0$ if $(i,j) \notin E(\ca G)$. 

In this section, we are interested in optimizing the LQR cost~(\ref{f(k)}) over the set,
\begin{align*}
  \ca K= \{K \in \ca U: A-BK \in \ca S\},
\end{align*}
where $\ca U$ is a linear subspace defined by the graph structure, i.e.,
\begin{align*}
  \ca U =\{M \in \bb M_{n \times m}(\bb R): M_{i,j} = 0 \text{ if and only if }(i,j) \not\in E(\ca G)\}.
\end{align*}
In light of the central theme of this work, projected gradient descent (PGD) is a natural choice for determining 
the feedback gain in the set $\ca K$, optimizing $f$ over $\ca U$. Such an approach
leads to the iteration of the form,
\begin{align}
  \label{eq:pgd}
  K_{j+1} = P_{\ca K}(K_j - \eta \nabla f(K_j)),
\end{align}
where $\eta$ is the stepsize; the choice of this stepsize will be discussed in~\S\ref{sec:structured_stepsize}. One may note that the geometry of $\ca K$ can be rather involved. Indeed, this set could have exponentially many path connected components (see~\cite{bu2019topological_mimo,feng-lavei:2019}). In the meantime, a favorable structure for $A$ and the graph $\ca G$ would guarantee that $\ca K$ has only one connected component~\cite{bu2019topological_mimo,feng-lavei:2019}. This point will not be further discussed in this paper. 
Herein, we further examine how to update the feedback gain in the path connected component of $\ca K$, once the algorithm has
been initialized in this component.\par
Even this more modest objective however faces some issues as $\ca K$ has an intricate geometry and one has
to address how to efficiently project onto it. In the sequel, we shall show that the seemingly relaxed update rule,
\begin{align*}
  K_{j+1} = P_{\ca U}(K_j - \eta \nabla f(K_j))
\end{align*}
is equivalent to \eqref{eq:pgd}, where $P_{\ca U}$ denotes the orthogonal projection onto $\ca U$. 
\begin{theorem}
  \label{thm:pgd}
  The updating rule \eqref{eq:pgd} is equivalent to
\begin{align*}
  K_{j+1} = K_j - \eta P_{\ca U}(\nabla f(K_j))
\end{align*}
provided the initial condition $K_0 \in \ca K$.
\end{theorem}
In proving this theorem, we will demonstrate that the relaxed updating rule is equivalent to the gradient descent update over a $C^{\infty}$ function $g: \ca K \to \bb R$ which is the restriction of $f$, i.e., $g = f|_{\ca K}$. We first establish several favorable properties of $g$.
\begin{lemma}
  The set $\ca K$ is open in $\ca U$ and the relative boundary of $\ca K$  is a subset of the boundary $\ca S$, i.e., $\mbox{\bf rbd } \ca K \subset \partial \ca S$. 
\end{lemma}
\begin{proof}
  Since $\ca K = \ca U \cap \ca S$ and $\ca S$ is open in $\bb M_{n \times n}(\bb R)$, the set $\ca K$ is open in the subspace topology. If $x \in \text{rbd }\ca K$, then for any $\varepsilon > 0$, $B_{\varepsilon}(x) \cap \ca U$ contains points both in $\ca U \cap \ca S$ and $(\ca U \cap \ca S)^c$. It follows then that $B_{\varepsilon}(x)$ contains points both in $\ca S$ and $\ca S^c$ and hence $x \in \ca \partial S$.
\end{proof}
As a consequence of the characterization of the relative boundary, the restriction $g$ is also coercive;
first, recall the definition of $\tilde{f}$ in Corollary~\ref{cor:compact}.
\begin{lemma}
  Let $\tilde{g} \coloneqq \tilde{f}(P_{\ca U} x): \bb M_{m \times n}(\bb R) \to \bb R \cup \{+\infty\}$. Then $\tilde{g}$ is continuous and infinitely differentiable on $\ca K$. For $K \in \ca S$ and $E \in \ca U$, we have
\begin{align*}
  \nabla \tilde{g}(K) &= P_{\ca U}(\nabla \tilde{f}(K)),
  \end{align*}
  and
  \begin{align*}
  \nabla^2 \tilde{g}(K)[E, E] &= \langle E, \nabla^2 \tilde{f}(K) E \rangle.
\end{align*}
\end{lemma}
\begin{proof}
  The function $\tilde{g}$ is continuous as it is a composition of $\tilde{f}$ and $P_{\ca U}$. 
As such $g \in C^{\infty}$ on $\ca K$ as $f$ is $C^{\infty}$ on $\ca S$. Furthermore, we note that,
\begin{align*}
  D{g}(K)[E] = D\tilde{f}(K)[P_{\ca U}(K)] P_{\ca U}(E).
\end{align*}
In terms of matrix representation in the standard basis and $K \in \ca K$,
\begin{align*}
  \langle \nabla {g}(K), E\rangle = \langle P_{\ca U}(\nabla \tilde{f}(K)), E\rangle.
\end{align*}
Moreover, 
\begin{align*}
  D^2 \tilde{g}(K)[E, F] &= D(D\tilde{f}(P_{\ca U} K) P_{\ca U} E)[F] \\
  &=D^2 {f}(P_{\ca U} K)[P_{\ca U}E, P_{\ca U}F].
\end{align*}
Hence, if $K \in \ca K$ and $E, F \in \ca U$,
\begin{align*}
  \nabla^2 \tilde{g}(K) [E, F] = \nabla^2 \tilde{f}(K)[E, F].
\end{align*}
\end{proof}
We are now ready to provide the proof for Theorem \ref{thm:pgd}.
\begin{proof}
  We note that $g = \tilde{g}|_{\ca K}$. Based on the initial state independent formulation, $g(K) < \infty$ implies that $K \in \ca K$. Thus, if $K_0 \in \ca K$, then the update rule is exactly,
\begin{align*}
  K_{1} = K_0 - \eta P_{\ca U}(\nabla f(K_0)) = K_0 - \eta \nabla g(K_0).
\end{align*}
Therefore, if $\eta$ is chosen sufficiently small for which $g(K_1) \le g(K_0)$, then $K_1 \in \ca K$. Thereby, the update rule is equivalent to,
\begin{align*}
  K_1 = P_{\ca U}(K_0 - \eta \nabla f(K_0)) = P_{\ca K}(K_0 - \eta \nabla f(K_0)).
\end{align*}
The statement of the theorem now follows by induction.
\end{proof}
\subsection{Convergence of Projected Gradient Descent}
\label{sec:structured_stepsize}
As we have argued, the projected gradient descent scheme is equivalent to gradient descent on $g$.\footnote{One should note that the analysis in \S~\ref{sec:gd} can not be adopted for the projected case. In the analysis of one step progression of gradient descent, the crucial fact is that the difference between $f(K_j)$ and $f(K_{j+1})$ is bounded in terms of product of positive semidefinite matrices. However, in projected case, if we follow the same line of reasoning, we arrive at the term $\ca V \coloneqq \Tr(M_K^{\top} [P_{\ca U}(M_K Y_K)] + [P_{\ca U}(M_K Y_K)]^{\top} M_K)$ and there is no clear lower bound for $\ca V$ in terms of $\Tr([P_{\ca U}(M_K Y_K)]^{\top}[P_{\ca U}(M_K Y_K)])$. In fact, $\ca V$ could be negative definite or indefinite in general.} Conceptually, the stepsize can be determined as follows: if $K_0 \in \ca K$, then the sublevel set $S_{g(K_0)} = \{ K \in \ca S: g(K) \le g(K_0)\}$ is compact. As $\| \nabla^2g(K)\|$ is continuous, there is a scalar $L > 0$ such that $\max_{K \in S_{g(K_0)}}\|\nabla^2g(K)\| = L$, i.e., the gradient mapping $\nabla g(K)$ is Lipschitz continuous with rank $L$ on $S_{g(K_0)}$. We may have chosen a constant stepsize $1/L$ if $g$ was a convex function. However, nonconvexity of $g$ and $\ca K$ introduce additional complications for determining the constant stepsize. In fact, we need to first address whether the sequence $\{K_j\}_{j=0}^{\infty}$ generated by \eqref{eq:gradient_descent} is guaranteed to stay in $\ca K$.

  As in the convergence analysis of $L$-smooth convex functions, at iterate $K_j$, a quadratic function majorizing $g(K)$ is formulated; in this case, minimizing the quadratic majorizing function will lead to the global minimum. In our case, the quadratic majorant,
\begin{align*}
  m(K; K_j) = g(K_j)+ \langle \nabla g(K_j), K-K_j\rangle + \frac{L}{2} \|K-K_j\|^2,
\end{align*}
only majorizes $g(K)$ over the sublevel set $S_{g(K_0)}$. Since $\ca K$ is not convex, it is not straightforward that,
\begin{align*}
  K_{j+1}= K_j - \frac{1}{L}\nabla g(K_j)= \argmin_{K} m(K;K_j)
\end{align*}
  is still stabilizing. But the coerciveness of $g$ remedies this complication.
\begin{lemma}
  \label{lemma:sequence_bounded}
Let $L = \sup_{K \in S_{g(K_0)}} \|\nabla^2 g(K)\|$ and consider the sequence $\{K_j\}_{j=0}^{\infty}$ generated by gradient descent \eqref{eq:pgd} with constant stepsize $\eta = 1/L$. If $K_0 \in \ca K$ then the sequence stays in $\ca K$.
\end{lemma}
  \begin{proof}
    Let $m(K; K_0) \coloneqq g(K_0) + \langle \nabla g(K_0), K-K_0\rangle + \frac{L}{2} \|K-K_0\|_F^2$. Since $g_{f(K_0)}$ is compact, the ray $K_0 - t\nabla g(K_0)$ will intersect $S_{f(K_0)}$ at a gain other than $K_0$; denote it by, $$K' \coloneqq  K_0 - \xi \cdot \nabla g(K_0).$$ As the line segment
    $$[K_0, K'] \coloneqq \{K: K = K_0 + t (K' - K_0), t \in [0, \xi]\}, $$ is contained in $\ca S_{g(K_0)}$, $m(K)$ majorizes $g(K)$ over the line segement $[K_0, K']$. We now define a univariate function $\phi(t) = Q(K_0 - t \nabla f(K_0))$ and note that $\phi(t)$ majorizes $g$ over $[0, \xi]$. We have $\phi(0) = g(K_0)$ and $\phi(\xi) \ge g(K')=g(K_0)$ since $K', K_0 \in \partial S_{g(K_0)} $. By Rolle's Theorem, there exists a stationary point $t \in (0, \xi)$ with $\phi'(t) = 0$. Hence, $t = K_0 - (1/L)\nabla g(K_0)$ is contained in the sublevel set $S_{g(K_0)}$. The proof is now completed by induction. 
\end{proof}
As we have established the equivalence between the projected gradient descent for $f$ and $g$, and $g$ is smooth and coercive in the subspace $\ca U$, we immediately establish the sublinear convergence to a first-order stationary point. Note that the operator norm of the Hessian $\nabla^2 g(K)$ is given by,
\begin{align*}
  \|\nabla^2 g(K)\| = \sup_{\|E\|_F=1, E \in \ca U} \langle \nabla^2 g(K)[E], E\rangle.
\end{align*}
\begin{lemma}
  \label{lemma:sublinear_pgd}
  Suppose that $K_0 \in \ca K$ and recall that the sublevel set is given by 
  $$S_{g(K_0)}=\{K \in \ca K: g(K) \le g(K_0)\}.$$
Let $L = \sup_{K \in S_{g(K_0)}} \|\nabla^2 g(K)\|$; if the stepsize $\eta$ in~\eqref{eq:pgd} is set as $t = 1/L$, then the sequence $\{K_j\}_{j=0}^{\infty}$ generated by the projected gradient descent \eqref{eq:pgd} convergences to a first-order stationary point at a sublinear rate, i.e.,
  \begin{align*}
    \|\nabla g(K_j) \|^2 \to 0,
  \end{align*}
  at a rate of $O(1/k)$.
\end{lemma}
\begin{proof}
  This is straightforward by Lemma~\ref{lemma:sequence_bounded} and \S$1.2.3$ in~\cite{nesterov2013introductory}.
\end{proof}
     \subsection{Choosing the stepsize for projected gradient descent}
     As we have pointed out, choosing an appropriate stepsize is equivalent to estimating the operator norm of the Hessian $\nabla^2g(K)$ over the sublevel set $S_{g(K_0)}$. 
\begin{proposition}
  On the sublevel set $S_{g(K_0)}$, we have
  \begin{align*}
    \sup_{K \in S_{f(K_0)}}\|\nabla^2 g(K)\| \le \sup_{K \in S_{g(K_0)}}\| \nabla^2 f(K)\|.
  \end{align*}
\end{proposition}
\begin{proof}
  We only need to observe that for each $K \in \ca K$,
  \begin{align*}
    \|\nabla^2 g(K)\| &= \sup_{\|E\|_F=1, E \in \ca U} \langle \nabla^2 g(K)[E], E\rangle \\
                      &\le \sup_{\|E\|_F=1} \langle \nabla^2 g(K)[E], E\rangle \\
                      &= \sup_{\|E\|_F=1} \langle \nabla^2 f(K)[E], E\rangle.
  \end{align*}
\end{proof}
We next provide an estimate of $\|\nabla^2 f(K)\|$ in terms of the system matrices $A, B$, cost function coefficients $Q, R$, and the initial condition $K_0$. Let $\alpha = f(K_0)$ and $S_{\alpha}= \{K \in \ca S: f(K) \le \alpha\}$. For $K \in S_\alpha$, $\Tr(X(K){\bf \Sigma}) \le \Tr(X_0{\bf \Sigma})$, where $X_0$ is the solution to the Lyapunov equation $A_{K_0}^{\top} X_0 A_{K_0} + Q + K_0^{\top} RK_0 = X_0$.\footnote{Note that on the sublevel set $S_{\alpha}$, it \emph{does not hold} that $X \preceq X_0$}
We denote the bound on the operator norm of the Hessian $D^2 f(K)$ on the sublevel set $S_{f(K_0)}$ by $L$; namely,
\begin{align*}
  L = \max_{K \in S_{f(K_0)}} \|\nabla^2 f(K)\| = \max_{K \in S_{f(K_0)}} \sup_{\|E\|_F=1} |\nabla^2f(K)[E, E]|.
\end{align*}
In order to estimate $L$, we first observe that by trianglular inequality and Proposition~\ref{prop:linalg_facts},
\begin{align*}
  \sup_{\|E\|_F = 1}  |\nabla^2 f(K)  [E, E]| &\le 2 \sup_{\|E\|_F = 1}\Tr((E^{\top}RE + E^{\top}B^{\top} X B E)  Y)  + 4\sup_{\|E\|_F=1}|\Tr(E^{\top}B^{\top}X'{A_K}Y)| \\
                                      & \le 2 \lambda_{n}(Y) \lambda_n(R+B^{\top} X B) \Tr(E^{\top} E) + 4 \sup_{\|E\|_F=1}\|E^{\top}B^{\top} X'{A_K Y^{1/2}}\|_2 \Tr(Y^{1/2})\\
                                        &\le 2 \lambda_n(Y) \lambda_n(R+B^{\top} X B)+ 4 \sup_{\|E\|_F=1}\|E^{\top}B^{\top} X'{A_KY^{1/2}}\|_2 \Tr(Y^{1/2}),
  \end{align*}
  where the second inequality follows from Theorem $2$ in~\cite{mori1988comments}.

  In what follows, we estimate each term in~\eqref{eq:hessian_discrete} on $S_{f(K_0)}$. This will be achieved by a series of propositions. 
    We first estimate a bound for $\lambda_n(Y), \Tr(Y^{1/2}), \|A_KY^{1/2}\|_2, \lambda_n(R+B^{\top}XB)$ on $S_{f(K_0)}$. Recall that $Y$ is the solution of ${A_K}Y A_K^{\top} + {\bf \Sigma} = Y$, i.e., $Y= \sum_{j=0}^{\infty}({A_K})^j \, {\bf \Sigma} \, (A_K^{\top})^j$.
\begin{proposition}  \label{prop:up_trac}
When $K \in S_{f(K_0)}$,
\begin{align*}
  & \lambda_n(Y) \le \frac{f(K_0)}{\lambda_1(Q)}, \qquad \Tr(Y^{1/2}) \le \sqrt{\frac{2f(K_0)}{\lambda_1(Q)}},  \\
  &\|A_KY^{1/2}\|_2 \le \sqrt{\frac{f(K_0)}{\lambda_1(Q)}}, \qquad \lambda_n(R+B^{\top}XB) \le \lambda_n(R) + \frac{\|B\|^2 f(K_0)}{\lambda_1({\bf \Sigma})}.
\end{align*}
\end{proposition}
\begin{proof}
Recall we have already upper bounded $\Tr(Y)$ in Proposition~\ref{prop:gradient_dominant_bound}: $\Tr(Y) \le \frac{f(K_0)}{\lambda_1(Q)}$\footnote{Indeed, we bound $\Tr(Y_*)$ in Proposition~\ref{prop:gradient_dominant_bound}. But the proof works verbatim for any $Y$ over the sublevel set.}. It follows
  \begin{align*}
    \lambda_n(Y) &\le \Tr(Y) \le \frac{f(K_0)}{\lambda_1(Q)}, \\
    \Tr(Y^{1/2}) &= \sum_{j=1}^n \lambda_j^{1/2}(Y) \le \sqrt{2 \sum_{j=1}^n \lambda_j(Y)} = \sqrt{2\Tr(Y)} \le \sqrt{\frac{2f(K_0)}{\lambda_1(Q)}}, \\
    \|A_K Y^{1/2}\|_2^2 &= \lambda_n(Y^{1/2} A_K^{\top} A_K Y^{1/2}) \le \Tr(Y^{1/2} A_K^{\top}A_K Y^{1/2}) = \Tr(A_K Y A_K^{\top}) \le \Tr(Y) \le \frac{f(K_0)}{\lambda_1(Q)}.
    \end{align*}
    For $R+B^{\top}XB$, we have
    \begin{align*}
      \lambda_n(R+B^{\top}XB) \le \lambda_n(R) + \lambda_n(B^{\top}XB) \le \lambda_n(R) + \|B\|_2^2 \Tr(X) \le \lambda_n(R) + \frac{\|B\|_2^2 f(K_0)}{\lambda_1({\bf \Sigma})}.
      \end{align*}
\end{proof}
Next, we provide an upper bound for the spectral norm of $X'(K)[E]$ on $S_{f(K_0)}$.
\begin{proposition}
      \label{prop:bound_x_prime}
     When $K \in S_{f(K_0)}$ and $\|E\|_F = 1$,
      \begin{align*}
        \|X'(K)[E]\|_2 \le \frac{1}{\lambda_1(Q)}\left[\frac{1+\|B\|^2}{\lambda_1({\bf \Sigma})} f(K_0) + \lambda_n(R) - 1 \right].
      \end{align*}
    \end{proposition}
    \begin{proof}
      To simplify the notation, let
      \begin{align*}
        \xi = \frac{1}{\lambda_1(Q)}\left[\frac{1+\|B\|^2}{\lambda_1({\bf \Sigma})} f(K_0) + \lambda_n(R) - 1 \right].
        \end{align*}
      We note that by Proposition~\ref{prop:linalg_facts}, for every $\zeta > 0$ we have:
      \begin{align*}
            X'(E) - A_K^{\top} X'(E) A_K  & =-(BE)^\top X{A_K} +A_K^{\top} X (-BE) + E^{\top} R K + K^{\top} R E  \\
&\preceq A_K^{\top} X {A_K} +  (BE)^\top X BE + K^{\top} R K + E^{\top} R E \\
           &= X- Q + (BE)^\top X BE + E^{\top} R E \eqqcolon {\bf V}.
%
      \end{align*}
    But
      \begin{align*}
        \Tr(X + (BE)^{\top} X (BE)) &= \Tr(X + X BB^{\top}) \le (1+\lambda_n(BB^{\top})) \Tr(X) \\
                                    & \le \frac{1+\|B\|}{\lambda_1({\bf \Sigma})} \Tr(X {\bf \Sigma}) \le \frac{1+\|B\|^2}{\lambda_1({\bf \Sigma})} f(K_0).
        \end{align*}
        It then follows that,
        \begin{align*}
          {\bf V} &\preceq \left[\frac{1+\|B\|^2}{\lambda_1({\bf \Sigma})} f(K_0) + \lambda_n(R) \right] I - Q \\
&\preceq \frac{1}{\lambda_1(Q)}\left[\frac{1+\|B\|^2}{\lambda_1({\bf \Sigma})} f(K_0) + \lambda_n(R) - 1 \right] Q.
          \end{align*}
      Thereby $$X'(K)[E] \preceq \xi X.$$
      Conversely, following the reverse path of the above inequalities and using Proposition~\ref{prop:linalg_facts}, it can be shown that $X'(K)[E] \succeq - \xi I$. As such, $\|X'(K)[E]\|_2 \le \xi X.$
    \end{proof}
    Combining all the bounds, we have: 
       \begin{lemma} \label{lemma:L_Lipschitz}
      On the sublevel sets $S_{f(K_0)}$, the gradient $\nabla f(K)$ is $L$-Lipschitz continuous.
    \end{lemma}
    \begin{proof} 
It suffices to observe that on $S_{f(K_0)}$,
      \begin{align*}
         \| \nabla^2f(K) \| &\le 2 (\lambda_n(R) + \|B\|^2 \|X\|)\lambda_n(Y) + 4 \|B\| \|X'\| \|A_KY^{1/2}\|_2 \Tr(Y^{1/2}) \\
                            &\le \left( 2 \lambda_n(R) + \frac{2\|B\|^2 f(K_0)}{\lambda_1({\bf \Sigma})} + 4 \sqrt{2}\xi \|B\| \frac{f(K_0)}{\lambda_1({\bf \Sigma})}\right) \frac{f(K_0)}{\lambda_1(Q)},
      \end{align*}
      where $\xi$ is the constant defined in Proposition~\ref{prop:bound_x_prime} and note $\xi$ is only determined by problem data $(A, B, Q, R, {\bf \Sigma}, K_0)$.    \end{proof}
    Lemma~\ref{lemma:L_Lipschitz} provides a Lipschitz constant in terms of the LQR parameters $A, B, Q, R$ and initial condition $K_0$; hence, a stepsize for gradient descent. 
    \begin{remark}
      We shall point out that as the projected gradient descent algorithm proceeds the function values $g(K)$ decrease. Hence, we can re-estimate the bounds in the above propositions at each iteration. For example, at iteration $K_j$, the Lipschtiz constant $L_{K_j}$ of $\nabla g(K)$ over the sublevel set $S_{g(K_j)}$ can be estimated and we may as well use a stepsize $1/L_{K_j}$ by Lemma~\ref{lemma:sequence_bounded}. The benefit is that this stepsize is certainly larger than $1/L_{f(K_0)}$. In this case, we shall have an increasing sequence of stepsizes $\{1/L_j\}_{j=0}^{\infty}$ that is bounded from above.
      
The stepsize rule devised here certainly works for unstructured case (i.e., gradient descent). However, this stepsize is typically smaller than the one we work out in Lemma~\ref{lemma:gd_function_decrease}. The reason is that here all the terms must be bounded over the whole sublevel set while in Lemma~\ref{lemma:gd_function_decrease} we carefully compare one step progression of gradient descent.
    \end{remark}
\section{Simulation Results} \label{sec:simulation}
\label{sec:numerical}
In this section, we provide a representative set of examples to demonstrate the results reported in this paper. \par
We first demonstrate the exponential stability of the proposed continuous flows. The system is of form~(\ref{LTI}) with parameters $(A, B)$, $A \in \bb R^{100 \times{ 100}}$ and $B = I$, guaranteeing the controllability of the system. The entries of $A$ are sampled from a standard normal distribution $\ca N(0, 1)$. We also scale $A$ when necessary to make it stable such that the initial feedback gain can be set as $K_0 = 0$. The cost matrices $Q, R$ are taken to be identity with appropriate dimensions. For the natural gradient, we simulate the flow with two different Riemannian metrics, one induced by $Y$ and the other by $Y^2$. Figure~\ref{fig:flow_fig1} demonstrates the exponential stability of the corresponding trajectories and Figure~\ref{fig:flow_fig2} depicts the exponential stability of the Lyapunov functionals for all flows when ${\bf \Sigma} = 0.5 I$. The results are consistent with the observations discussed in \S\ref{sec:gf}, \S\ref{sec:ngf}, and \S\ref{sec:qnf}. In particular, since $\lambda_1(Y) < 1$, the natural gradient flow converges faster than the gradient flow, and amongst the natural gradient flows, the one with the metric induced by $Y^2$ outperforms the one with metric $Y$. Figures~\ref{fig:flow_fig3} and~\ref{fig:flow_fig4} show the convergence results with the same LQR parameters $(A, B, Q, R)$, but the initial state matrix has chosen to be ${\bf \Sigma} = 2 I$. These two figures underscore the observations in Remark~\ref{remark:gf_vs_ngf}: gradient flow outperforms natural gradient flows when $\lambda_1(Y) > 1$.
\begin{figure}
    \centering
    \begin{minipage}{0.45\textwidth}
        \centering
     \includegraphics[width=0.95\textwidth]{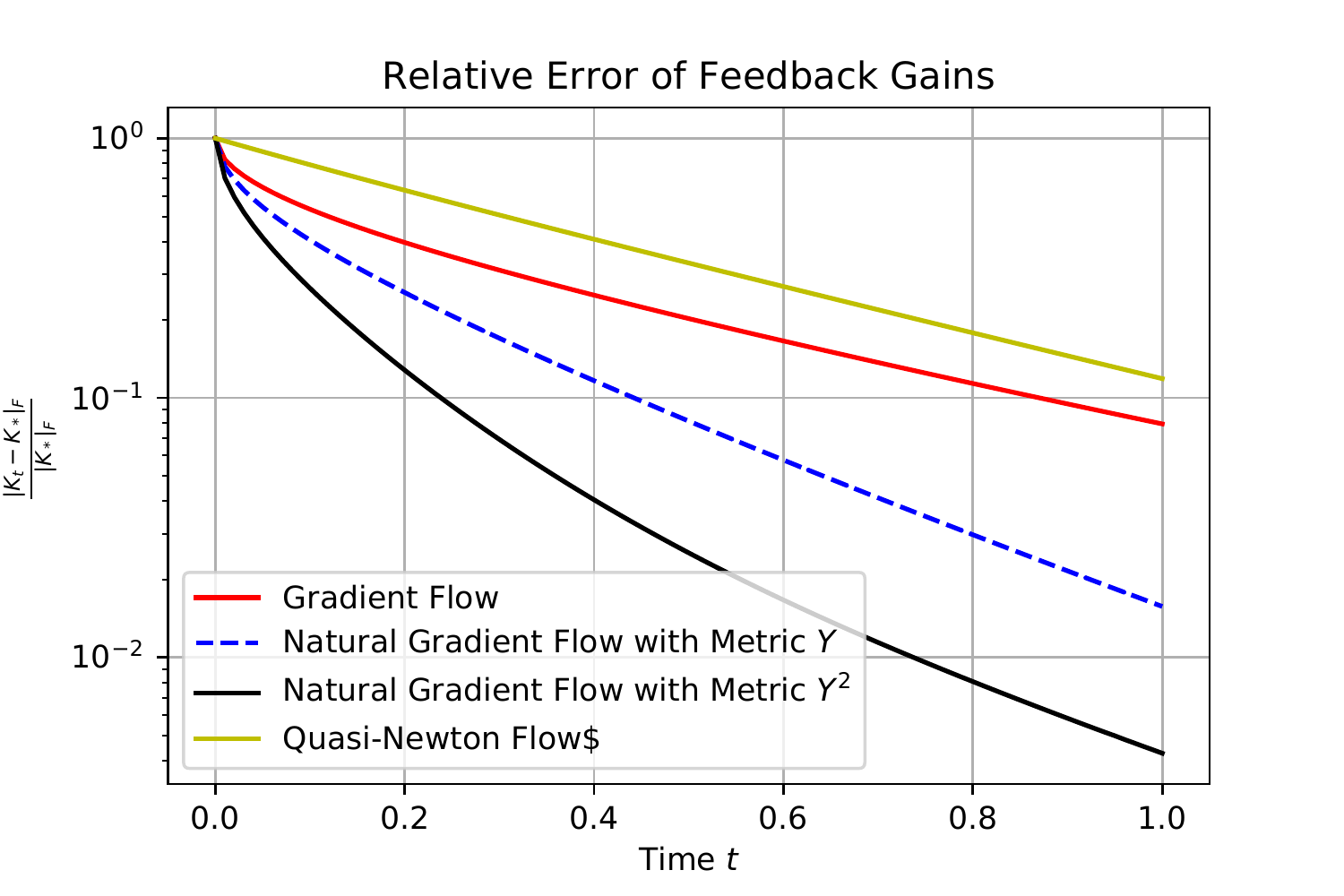}
      \caption{Exponential stability of trajectory $K_t$ with the initial state matrix ${\bf \Sigma} =0.5 I$.}
     \label{fig:flow_fig1}
    \end{minipage}\hfill
    \begin{minipage}{0.45\textwidth}
        \centering
        \includegraphics[width=0.95\textwidth]{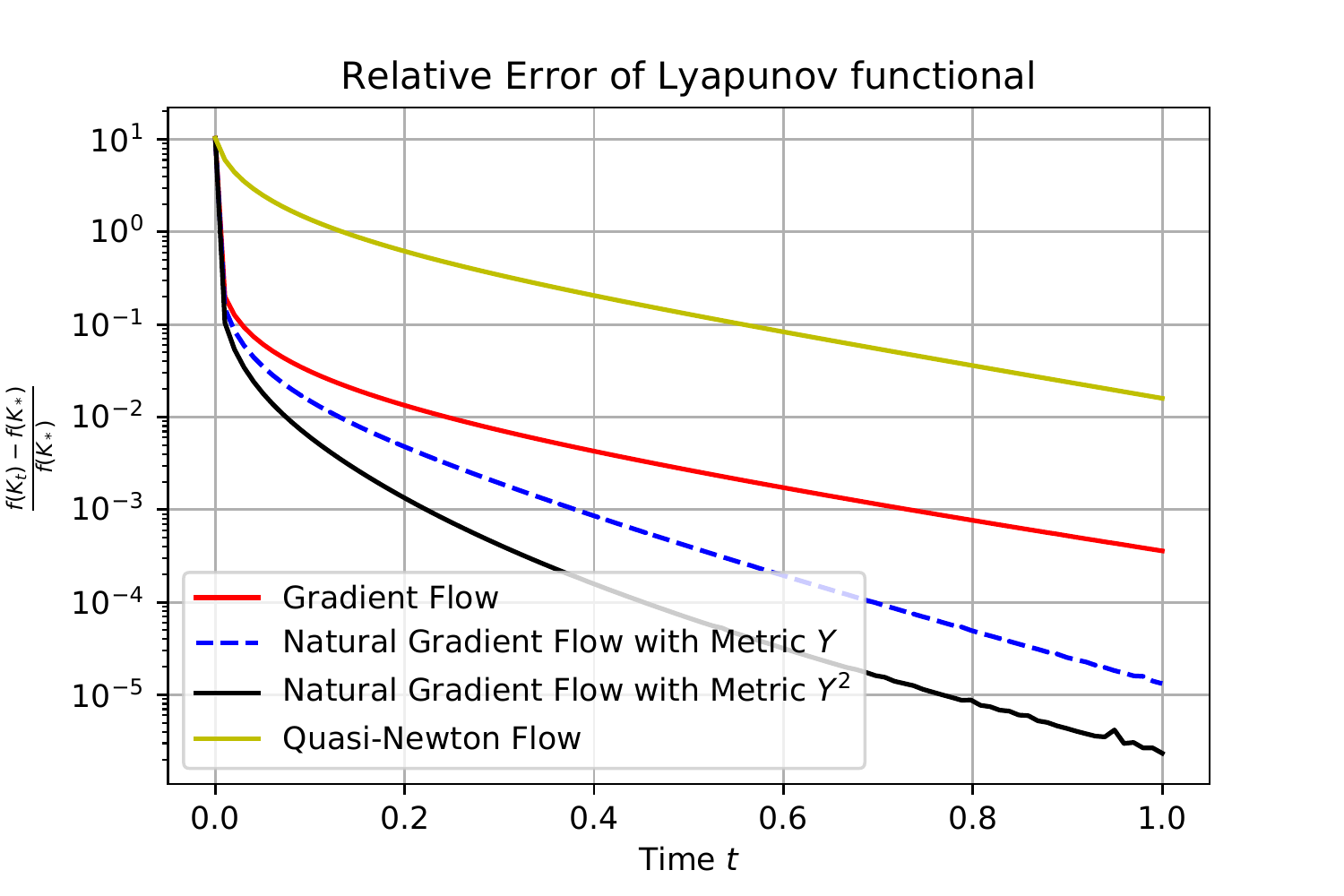}        
        \caption{Exponential decay of the Lyapunov functional with the initial state matrix ${\bf \Sigma} = 0.5 I$.}
         \label{fig:flow_fig2}
    \end{minipage}
\end{figure}
\begin{figure}[ht]
    \centering
    \begin{minipage}{0.45\textwidth}
        \centering
     \includegraphics[width=0.95\textwidth]{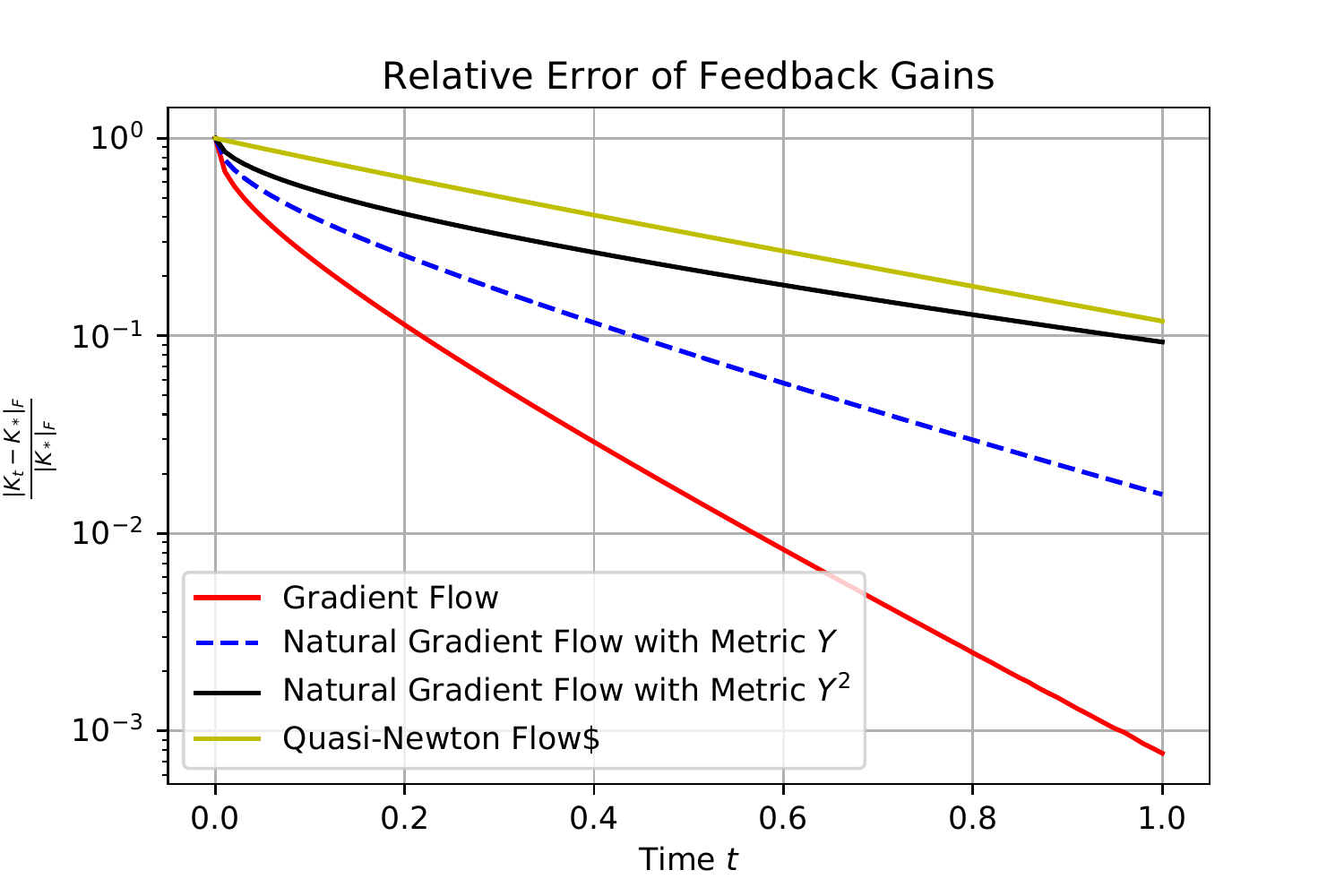}
                \caption{Exponential stability of trajectory $K_t$ with the initial state matrix ${\bf \Sigma} = 2I$.}
     \label{fig:flow_fig3}
    \end{minipage}\hfill
    \begin{minipage}{0.45\textwidth}
        \centering
        \includegraphics[width=0.95\textwidth]{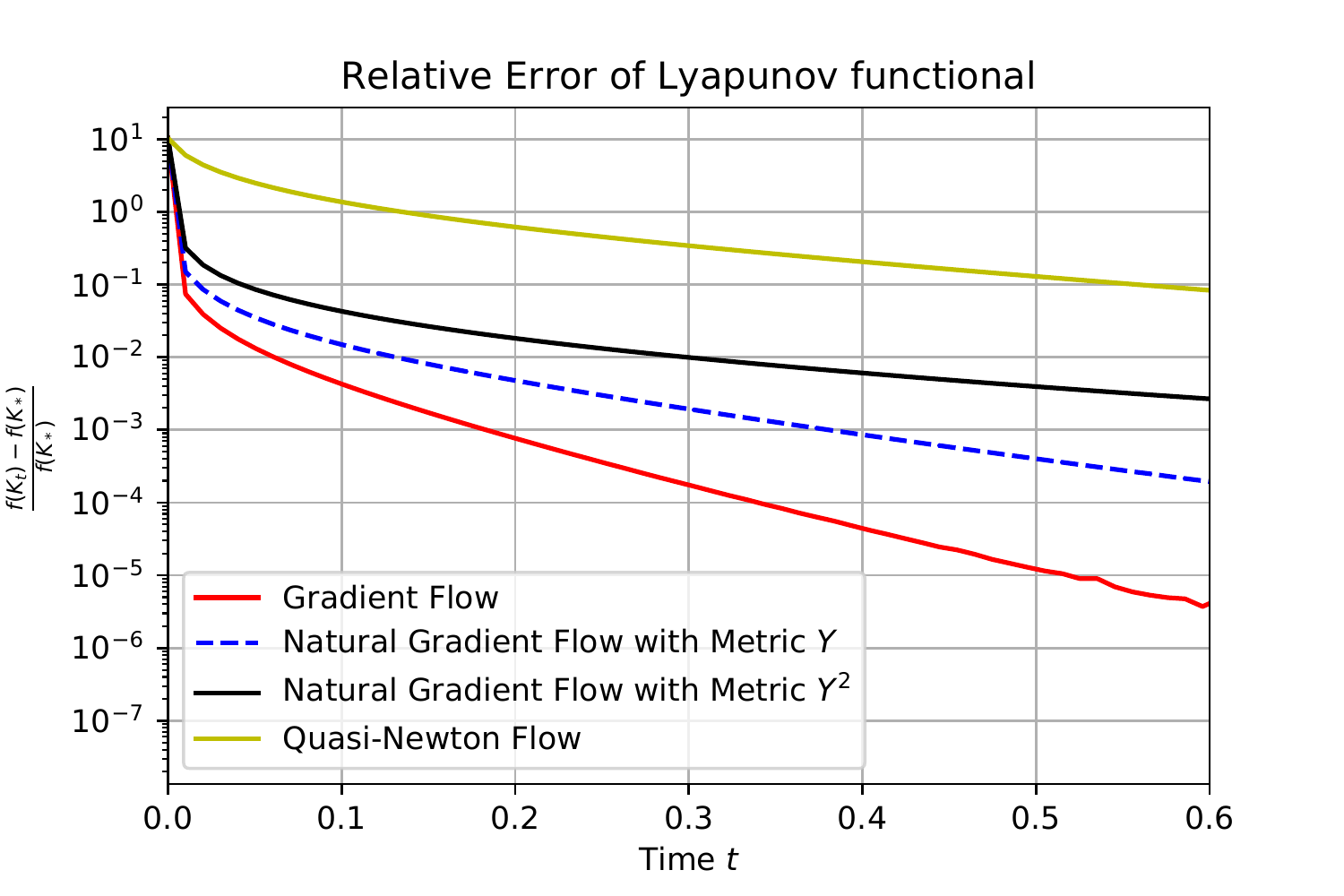}        
        \caption{Exponential decay of the Lyapunov functional with the initial state matrix ${\bf \Sigma} = 2I$.}
        \label{fig:flow_fig4}
    \end{minipage}
\end{figure}

Next we examine the discrete realizations of these flows, namely, gradient descent, natural gradient descent and the quasi-Newton iteration (with the same setup for system parameters). With the adaptive stepsize proposed in Theorem~\ref{thrm:gd_linear}, Figure~\ref{fig:fig1} demonstrates that the sequence of feedback gains generated by gradient descent is stabilizing and converges to the global optimal feedback gain. Moreover, Figure~\ref{fig:fig2} shows that the cost function $f(K)$ converges to $f(K_*)$ at a linear rate.
\begin{figure}[ht]
    \centering
    \begin{minipage}{0.45\textwidth}
        \centering
     \includegraphics[width=0.95\textwidth]{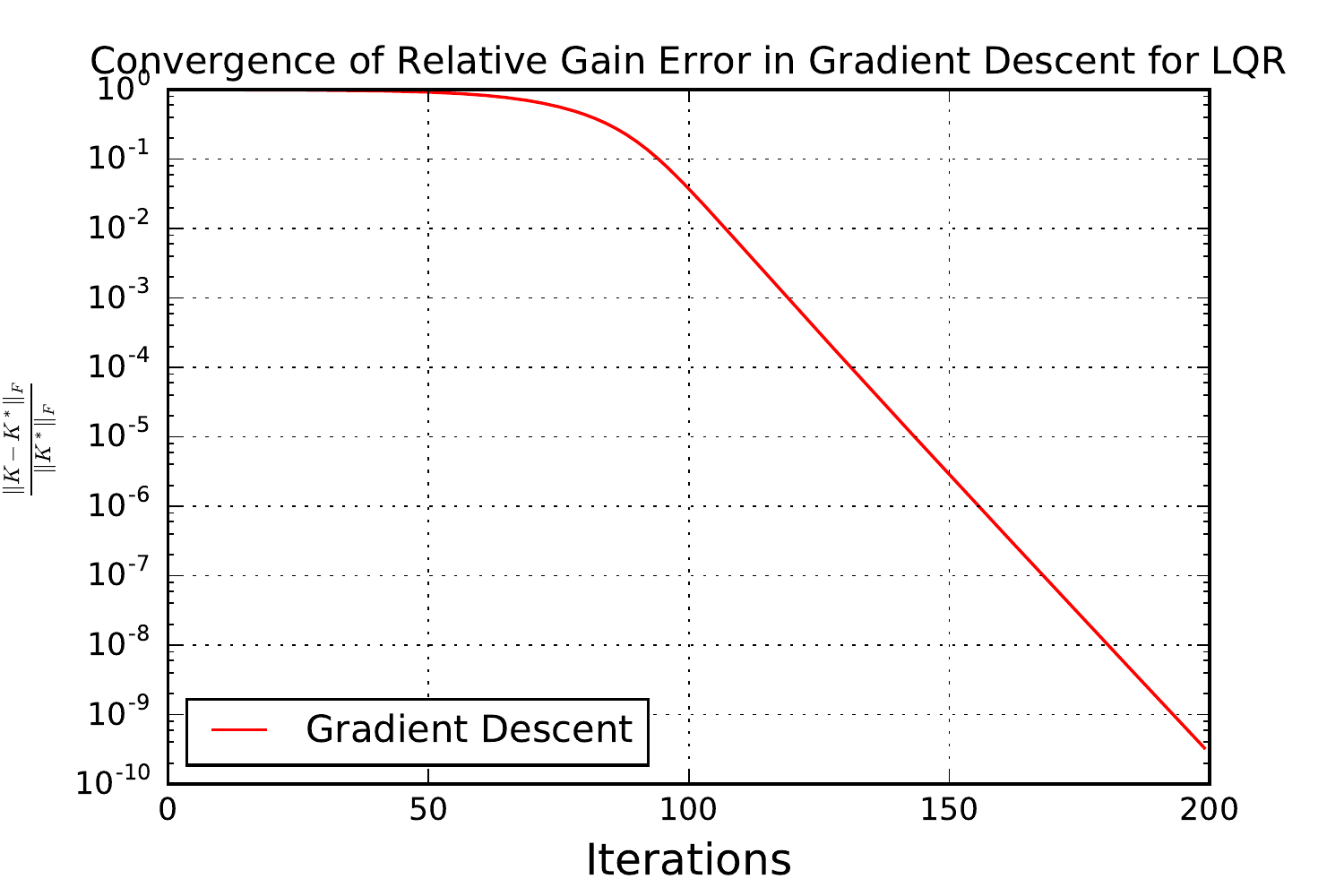}
      \caption{Convergence of the relative error for the feedback gain under gradient descent with adaptive stepsize given by~\ref{lemma:gd_function_decrease}}\label{fig:fig1}
    \end{minipage}\hfill
    \begin{minipage}{0.45\textwidth}
        \centering
        \includegraphics[width=0.95\textwidth]{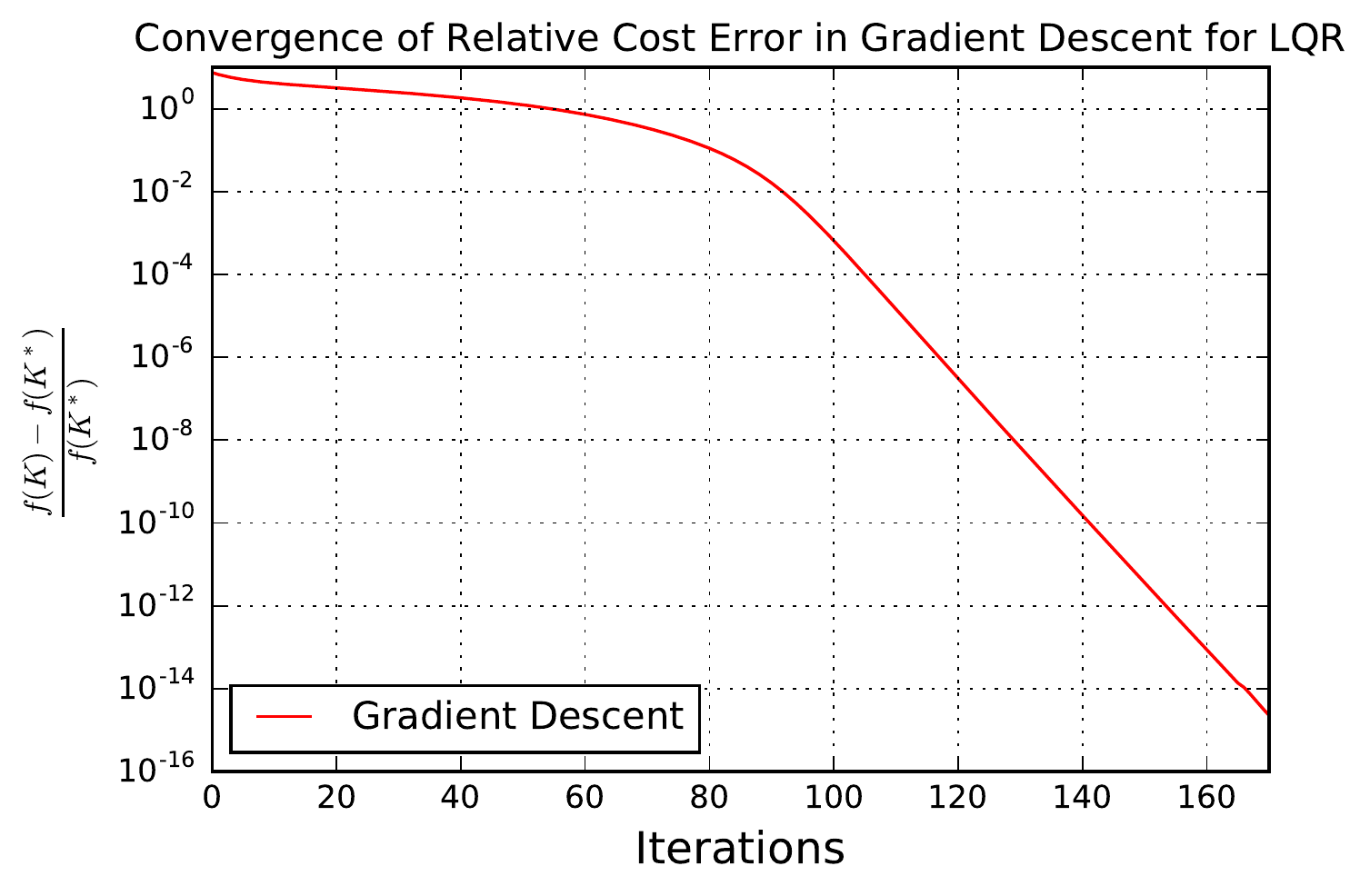}        
        \caption{Convergence of the relative error for the LQR cost under gradient descent}\label{fig:fig2}
    \end{minipage}
\end{figure}
In the meantime, Figures~\ref{fig:fig1_ngd} and~\ref{fig:fig2_ngd} demonstrate the linear convergence of the natural gradient descent algorithm. The stepsize is chosen adptively according to Theorem~\ref{thrm:ngd_convergence}; we note the faster convergence of natural gradient descent compared with gradient descent.
\begin{figure}[ht]
    \centering
    \begin{minipage}{0.45\textwidth}
        \centering
     \includegraphics[width=0.95\textwidth]{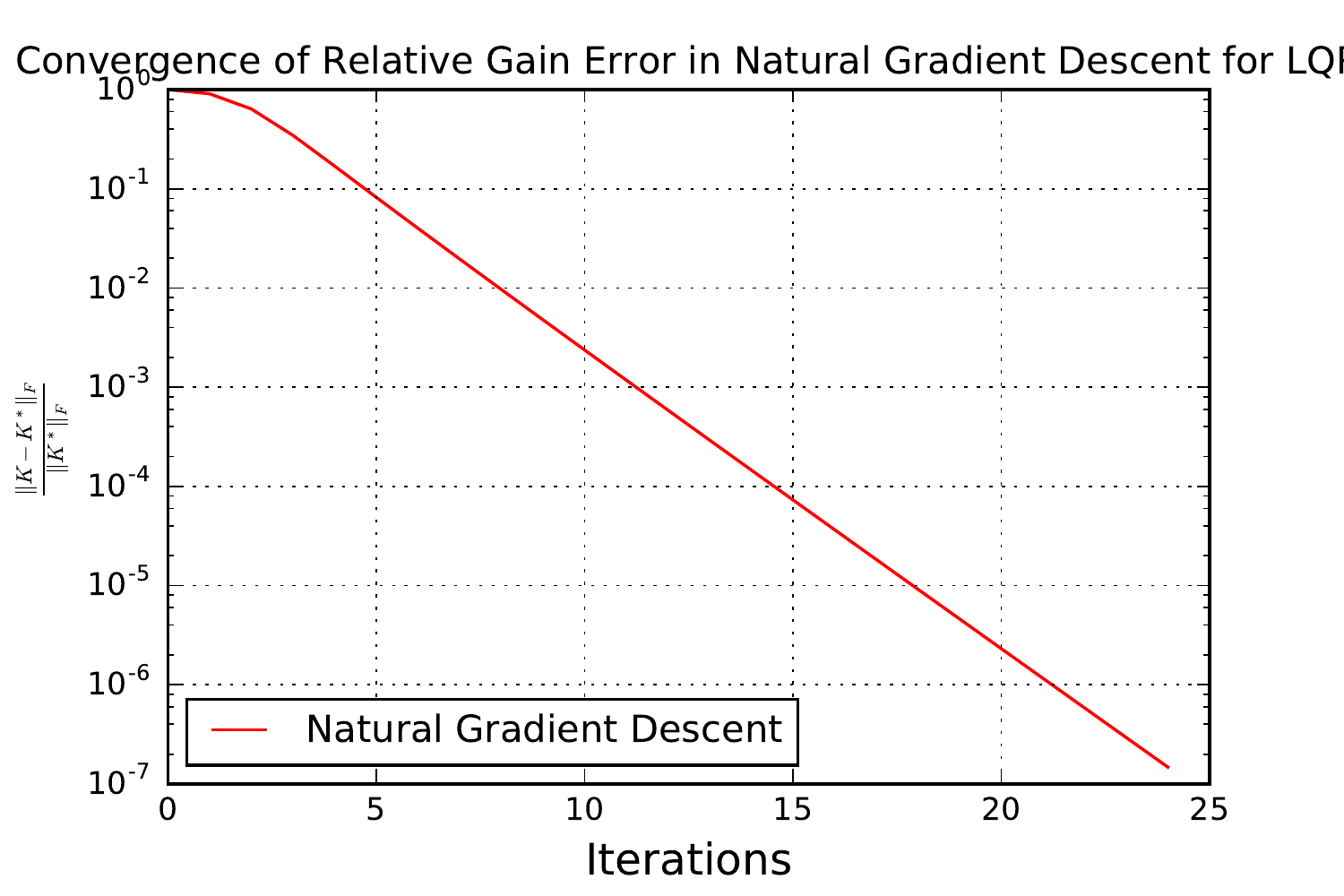}
      \caption{Convergence of the relative error for the feedback gain under natural gradient descent with adaptive stepsize given by~\ref{lemma:gd_function_decrease}}\label{fig:fig1_ngd}
    \end{minipage}\hfill
    \begin{minipage}{0.45\textwidth}
        \centering
        \includegraphics[width=0.95\textwidth]{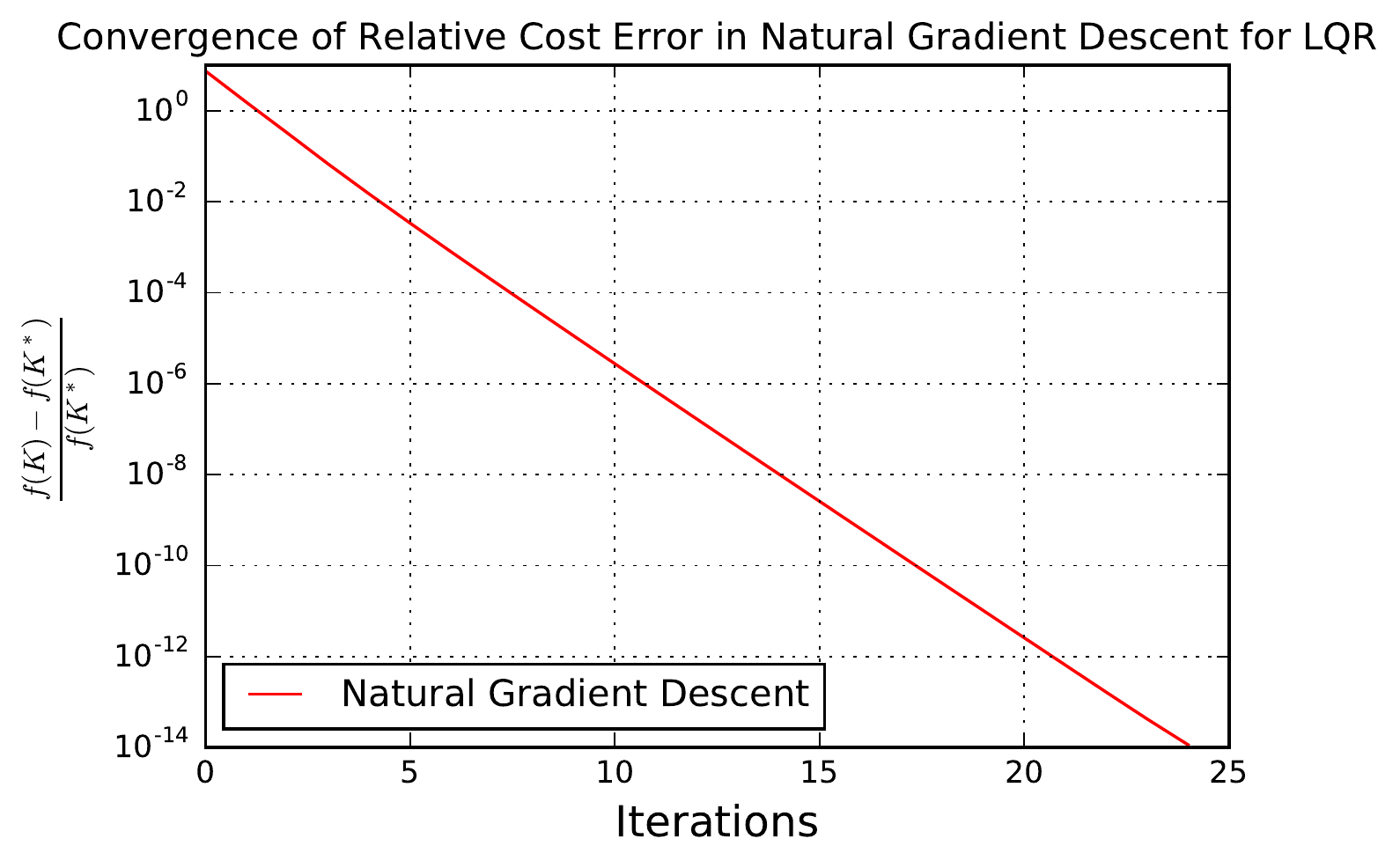}        
        \caption{Convergence of the relative error for the LQR cost under natural gradient descent}\label{fig:fig2_ngd}
    \end{minipage}
\end{figure}
Figures~\ref{fig:fig1_qn} and~\ref{fig:fig2_qn} demonstrate the quadratic convergence for the quasi-Newton iteration. The stepsize is chosen to be ${1}/{2}$; in this case, we recover the Hewer's algorithm, enjoying the fastest convergence rate.
\begin{figure}[ht]
    \centering
    \begin{minipage}{0.45\textwidth}
        \centering
     \includegraphics[width=0.95\textwidth]{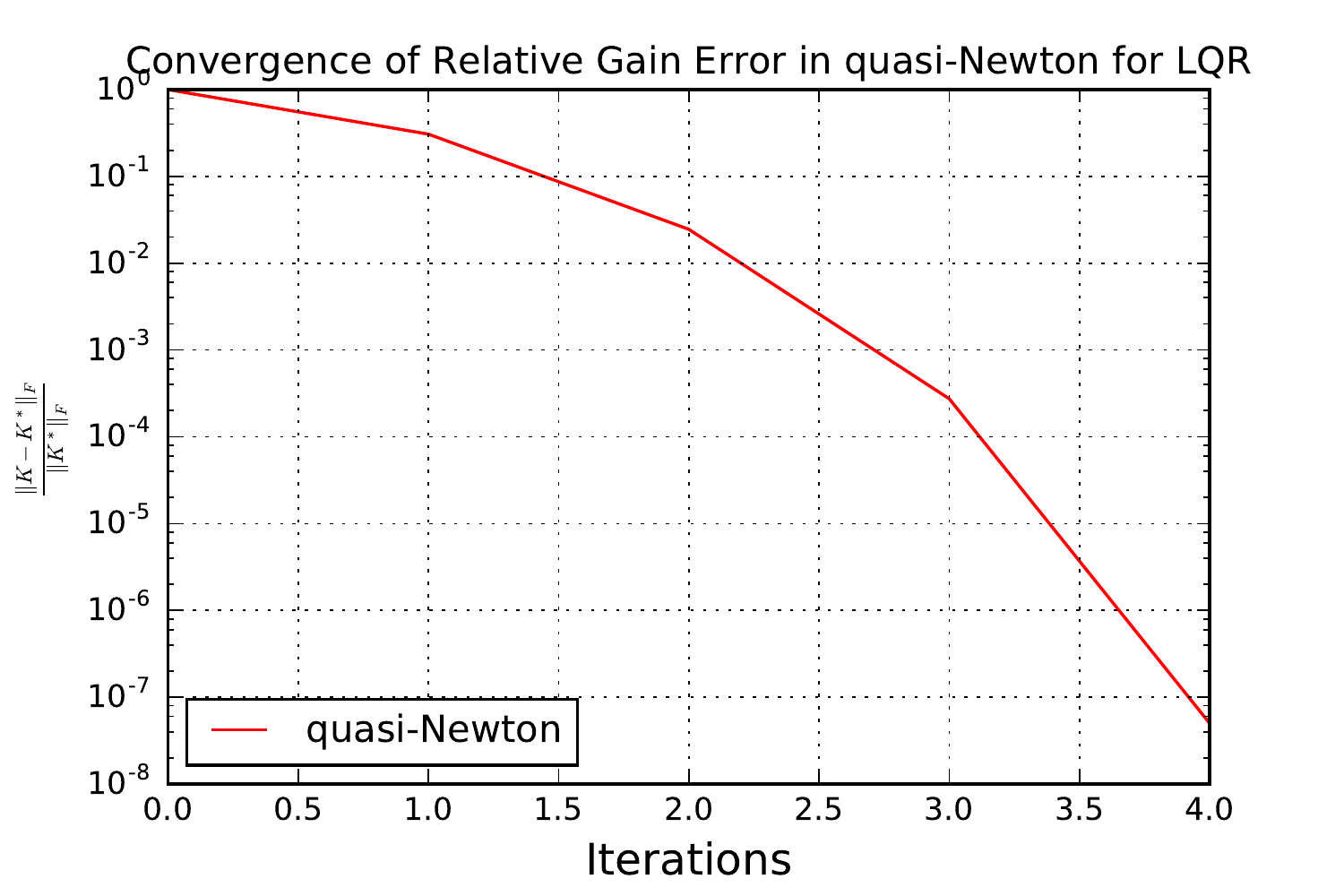}
      \caption{Convergence of the relative error for the feedback gain under quasi-Newton with constant stepsize $\frac{1}{2}$ } \label{fig:fig1_qn}
    \end{minipage}\hfill
    \begin{minipage}{0.45\textwidth}
        \centering
        \includegraphics[width=0.95\textwidth]{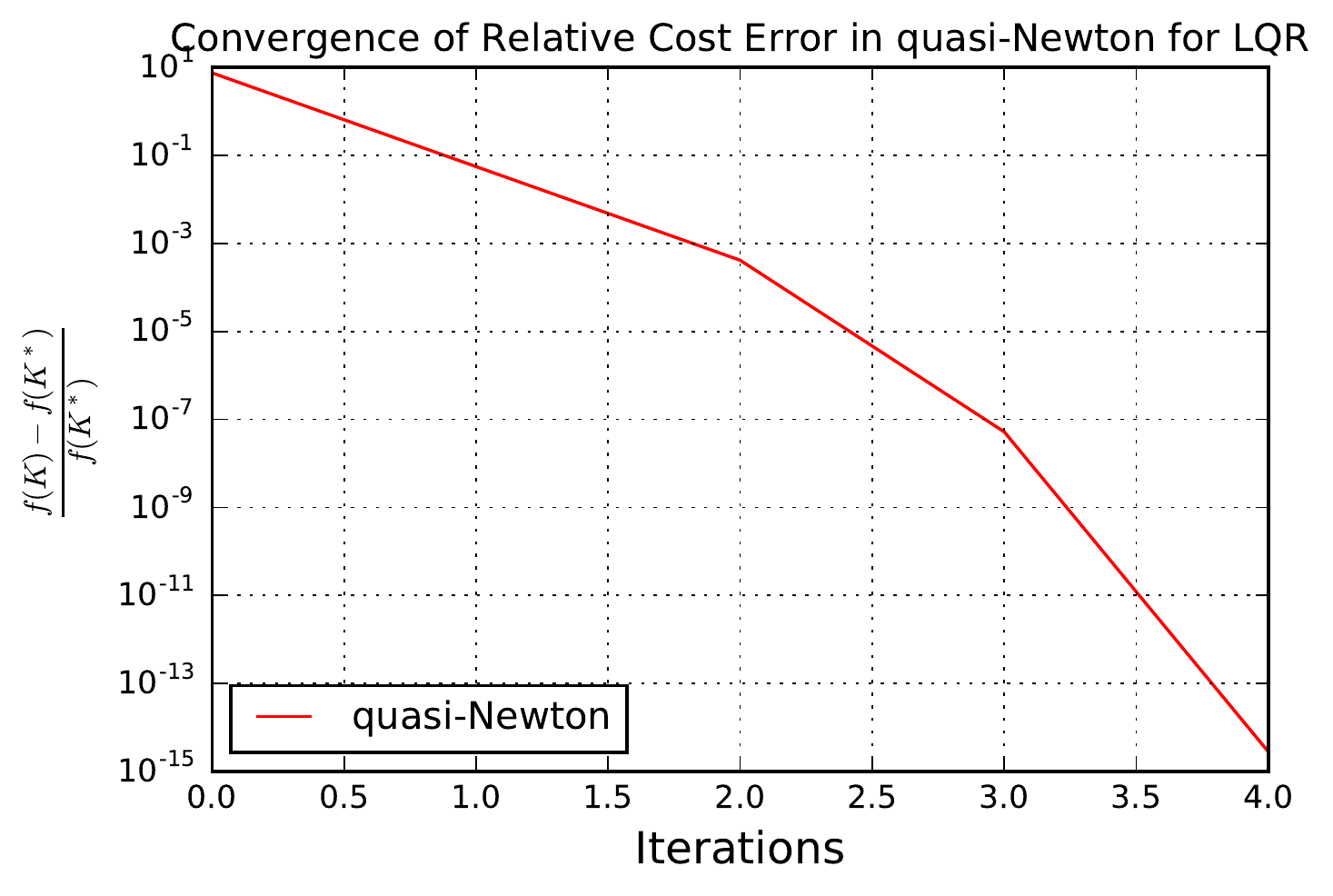}        
        \caption{Convergence of the relative error for the LQR cost under natural gradient descent}\label{fig:fig2_qn}
    \end{minipage}
\end{figure}
    We now examine the projected gradient descent for a system modeled over a $(10, 10)$-lollipop graph.\footnote{A lollipop graph consists of a complete graph on $10$ nodes and a path graph on $10$ nodes.} The system matrix $A$ is chosen as the Metropolis-Hastings weight matrix for the graph and $B=I$. The initial gain matrix is chosen as $K_0 = 0$. In each iteration the feedback gain is updated as,
    \begin{align*}
      K_{j} = P_{\ca U} (K_{j-1} - \eta \nabla f(K_{j-1})),
    \end{align*}
    where the projection is equivalent to zeroing out the entries that do not correspond to edges in the graph. Consistent with Lemma~\ref{lemma:sublinear_pgd}, Figure~\ref{fig:fig1_pgd} demonstrates that the sequence of feedback gains is stabilizing and converges to a first-order stationary point. Moreover, Figure~\ref{fig:fig2_pgd} depicts the convergence of the cost function $f(K)$.
\begin{figure}[ht]
    \centering
    \begin{minipage}{0.45\textwidth}
        \centering
     \includegraphics[width=0.95\textwidth]{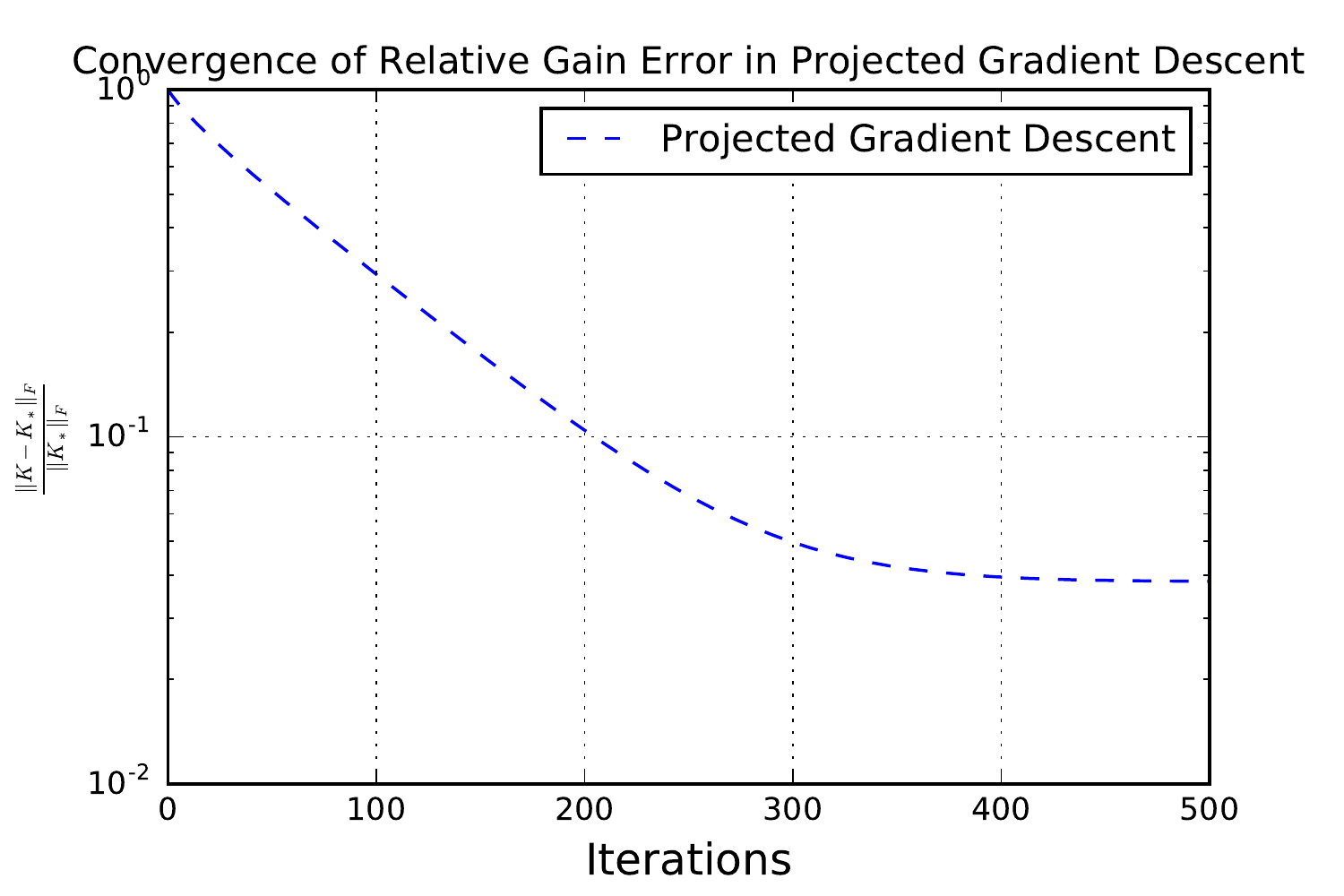}
      \caption{Convergence of the relative error for the feedback gain under projected gradient descent on a lollipop graph.}\label{fig:fig1_pgd}
    \end{minipage}\hfill
    \begin{minipage}{0.45\textwidth}
        \centering
        \includegraphics[width=0.95\textwidth]{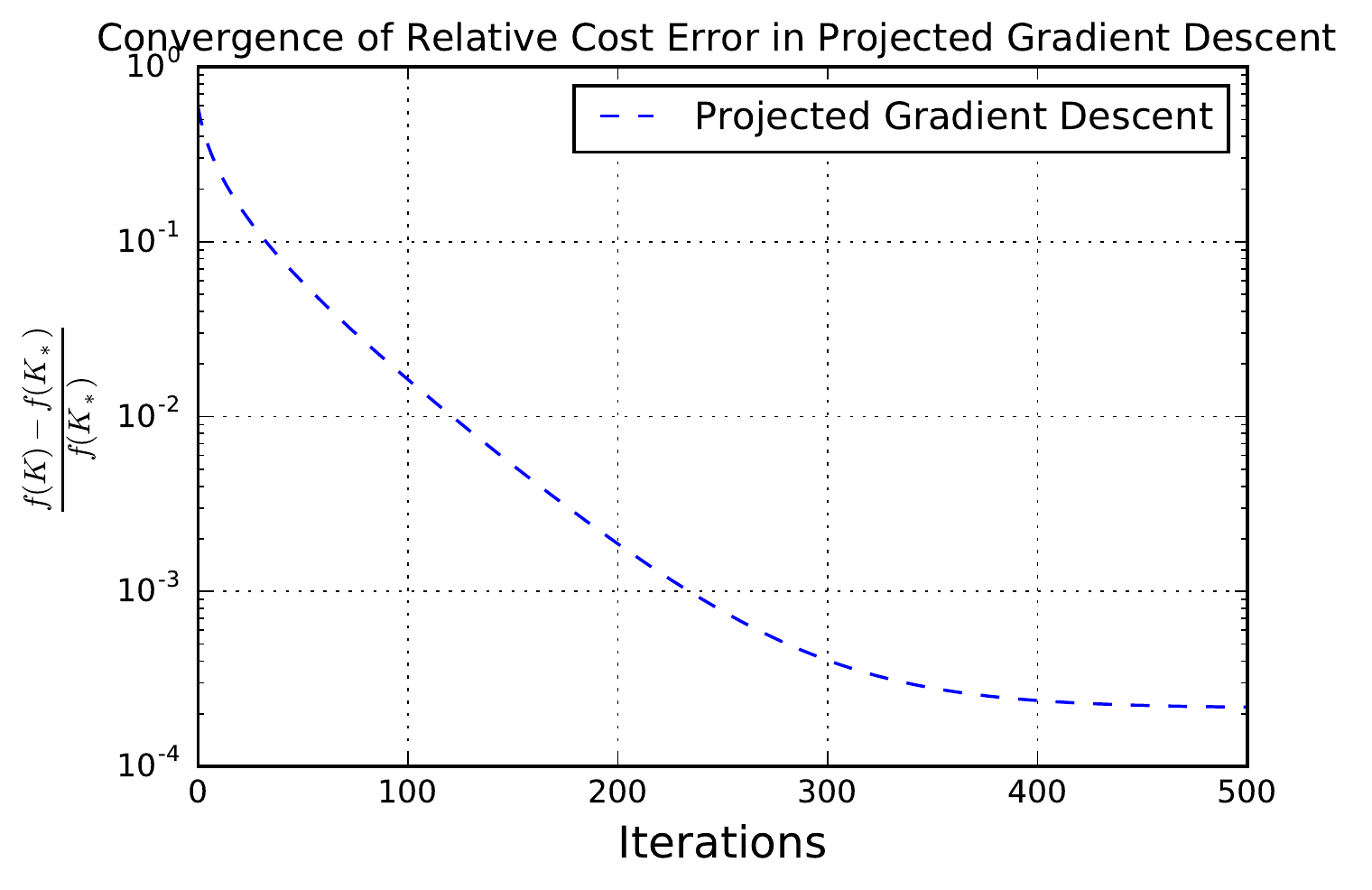}        
        \caption{Convergence of the relative error for the LQR cost under projected gradient descent on a lollipop graph.}\label{fig:fig2_pgd}
    \end{minipage}
\end{figure}

\section{Concluding Remarks} \label{sec:discussion}
The paper considers LQR through the lens of first order methods--an LQR calculus--where control synthesis
is viewed directly in terms of optimizing an objective function over the set of stabilizing feedback gains.
Using this narrative, we proceed to examine
gradient descent and its various extensions for solving the LQR problem.
The LQR objective is constructed over a set of linearly independent initial states to eliminate the dependency of the optimal policy on the initial state and encode closed loop stability.
It is shown that the corresponding cost function is smooth, coercive and gradient dominated (this latter fact was previously reported in the literature; we provide an alternate approach for its proof).
We next discussed three types of well-posed flows over the set of stabilizing controllers: gradient flow, natural gradient flow and the quasi-Newton flow.
We subsequently examine the discretization of these flows, 
and show that their realizations using the forward Euler method, i.e., gradient descent, natural gradient flow and quasi-Newton iterations,  lead to algorithms with linear convergence rate and quadratic convergence rate. 
Finally, we consider projected gradient descent for solving structured LQR. In this direction,
we provided a stepsize rule which leads to the sublinear convergence to the first-order stationary point.
\begin{center}
\section*{Acknowledgements}
\end{center} 
The authors acknowledge their many discussions with Sham Kakade and Rong Ge, particularly on the results reported in~\cite{fazel2018global} pertaining to the rigorous analysis of first order methods for LQR synthesis. This research was supported by DARPA Lagrange Grant FA8650-18-2-7836.\\
\begin{appendices}
\section{Bounding $\|Y'(\theta)\|_2$ in Lemma~\ref{lemma:gd_function_decrease}}
\label{appendix:bound_Y_prime}
\begin{proposition}
  Suppose that $K_{\eta} = K - 2 \eta MY$. Then for any $\theta$ such that $f(K_{\theta}) \le f(K)$, we have
\begin{align*}
  \|Y'(\theta)\|_2 \le \frac{4f(K)\lambda_n(Y)\|BMY\|_2}{\lambda_1(Q)}.
\end{align*}
\end{proposition}
\begin{proof}
  Putting $A_{\theta}= A-BK + 2 \theta BMY$, $Y_{\theta}$ solves
\begin{align*}
  Y_{\theta} = A_{\theta} Y_{\theta} A_{\theta}^{\top} + {\bf \Sigma}.
\end{align*}
Therefore, the derivative $Y_{\theta}'$ solves
\begin{align*}
  Y_{\theta}' = A_{\theta} Y_{\theta}' A_{\theta}^{\top} + 2BMY Y_{\theta} A_{\theta}^{\top} + A_{\theta} Y_{\theta} (2BMY)^{\top}.
\end{align*}
By Proposition~\ref{prop:linalg_facts}, we can bound the derivative by
\begin{align*}
  Y_{\theta}' - A_{\theta} Y_{\theta}' A_{\theta}^{\top} &= 2BMY Y_{\theta} A_{\theta}^{\top} + A_{\theta} Y_{\theta} (2BMY)^{\top} \\
  &\preceq 4\upsilon BMY Y_{\theta}(BMY)^{\top} + \frac{1}{\upsilon} A_{\theta} Y_{\theta} A_{\theta}^{\top} \\
                                                         &\preceq 4 \upsilon\|BMY\|_2^2 \lambda_n(Y_{\theta}) I+ \frac{1}{\upsilon} (Y_\theta - {\bf \Sigma}) \\
                                                         &\preceq 4 \upsilon\|BMY\|_2^2 \lambda_n(Y_{\theta}) I+ \frac{1}{\upsilon} \lambda_n(Y_\theta)I,
\end{align*} where $\upsilon \in \bb R_+$ is any positive real number.
This upper bound can be minimized by taking $\upsilon = 2\|BMY\|_2$. It thus follows that,
\begin{align*}
Y_{\theta}' - A_{\theta} Y_{\theta}' A_{\theta}^{\top} &\preceq 4\|BMY\|_2 \lambda_n(Y_\theta)I.
  \end{align*}
  It follows by Proposition~\ref{prop:linalg_facts} that
\begin{align*}
  Y_\theta' \preceq{4\|BMY\|_2 \lambda_n(Y_\theta) }Y.
\end{align*}
Similarly, we can get
\begin{align*}
  Y_\theta' \succeq -{4\|BMY\|_2 \lambda_n(Y_\theta) }Y.
\end{align*}
Hence, we have
\begin{align*}
  \|Y_\theta'\|_2 \le 4\|BMY\|_2 \lambda_n(Y_\theta)\lambda_n(Y).
\end{align*}
\end{proof}
\section{Lower Bounding Stepsize $\eta_j$ in Gradient Descent}
\label{appendix:upperbound_stepsize}
The purpose of this section is to show the stepsize rule we provided in Theorem~\ref{thrm:gd_linear} is not vanishing, i.e., bounded away from $0$. This is important to conclude the linear convergence. On the other hand side, we may certainly choose the lower bound as the constant stepsize. However, this will make the convergence rather slow. \par
We first bound the constant $b_j, c_j$ defined in Lemma~\ref{lemma:gd_function_decrease}.
\begin{proposition}
  \label{prop:appendix_bound_b_c}
  Over the sublevel set $S_{f(K_0)}$,
  \begin{align*}
    b_j &\le \frac{ \lambda_n(R) f(K_0) + \|B\|^2 \frac{(f(K_0))^2}{\lambda_1({\bf \Sigma})} + 4 \|B\|_2 \sqrt{(\lambda_n(R)+\|B\|_2^2 \frac{f(K_0)}{\lambda_1({\bf \Sigma})}) f(K_0)} \left(\frac{f(K_0)}{\lambda_1(Q)}\right)^2}{\lambda_1(Q)}, \\
    c_j &\le \frac{ 
          4\left(\lambda_n(R) + \|B\|^2 \frac{f(K_0 }{\lambda_1({\bf \Sigma})}\right)
          \|B\|_2 f(K_0)  \sqrt{(\lambda_n(R)+\|B\|_2^2 \frac{f(K_0)}{\lambda_1({\bf \Sigma})}}
          \left( \frac{ f(K_0)}{\lambda_1(Q)}\right)^2}
{\lambda_1(Q)}.
  \end{align*}
\end{proposition}
\begin{proof}
  To derive an upper bound, we need to upper bound $\lambda_n(R+B^{\top} X_j B), \|Y_j\|_2, \|M_{K_j}\|_2$. Note by virtually the same argument in Proposition~\ref{prop:gradient_dominant_bound}, we can upper bounded $\|Y_j\|_2$ by
\begin{align*}
  \|Y_j\|_2 \le \Tr(Y_j) \le \frac{f(K_0)}{\lambda_1(Q)}.
\end{align*}
  To upper bound $\lambda_n(R+B^{\top} X_j B)$, we observe
\begin{align*}
  \lambda_n(R+B^{\top}X_jB) \le \lambda_n(R) + \|B\|_2^2 \lambda_n(X_j) \le \lambda_n(R) + \|B\|_2^2 \frac{\Tr(X_j{\bf \Sigma})}{\lambda_1({\bf \Sigma})} \le \lambda_n(R)+\|B\|_2^2 \frac{f(K_0)}{\lambda_1({\bf \Sigma})}.
\end{align*}
By \eqref{eq:lower_bound_M_K} in the proof of Theorem~\ref{thrm:ngd_convergence}, we can upper bound $\|M_{K_j}\|_2$\footnote{Algebraically, the bound is easy to derive; however, we should be careful about the stabilization issue, namely, whether $K-\eta M_K$ is stabilizing.} by
\begin{align*}
  \|M_{K_j}\|_2^2 \le \Tr(M_{K_j}^{\top} M_{K_j}) \le \lambda_n(R+B^{\top} X_j B) f(K_0) \le (\lambda_n(R)+\|B\|_2^2 \frac{f(K_0)}{\lambda_1({\bf \Sigma})}) f(K_0).
\end{align*}
  The desired bounds can then be acquired by combining these bounds and triangular inequality.
\end{proof}
\begin{proposition}
  \label{prop:stepsize_bound}
  The stepsize in Theorem~\ref{thrm:gd_linear} is lower bounded away from $0$.
\end{proposition}
\begin{proof}
  Putting $d_j = \max\{b_j, c_j\}$, then $d_j \le \delta$ over the sublevel set $S_{f(K_0)}$ where $\delta > 0$ is the maximum of the two upper bounds in Proposition~\ref{prop:appendix_bound_b_c}. Now $\eta_j$ can be seen
\begin{align*}
  \eta_j = \sqrt{\frac{1}{d_j}+\frac{1}{9}} - \frac{1}{3} \ge \sqrt{\frac{1}{\delta} + \frac{1}{9}} -\frac{1}{3}> 0.
\end{align*}
\end{proof}
\end{appendices}
\bibliographystyle{ieeetran}
\bibliography{ref}

\begin{thebibliography}{10}
\providecommand{\url}[1]{#1}
\csname url@samestyle\endcsname
\providecommand{\newblock}{\relax}
\providecommand{\bibinfo}[2]{#2}
\providecommand{\BIBentrySTDinterwordspacing}{\spaceskip=0pt\relax}
\providecommand{\BIBentryALTinterwordstretchfactor}{4}
\providecommand{\BIBentryALTinterwordspacing}{\spaceskip=\fontdimen2\font plus
\BIBentryALTinterwordstretchfactor\fontdimen3\font minus
  \fontdimen4\font\relax}
\providecommand{\BIBforeignlanguage}[2]{{%
\expandafter\ifx\csname l@#1\endcsname\relax
\typeout{** WARNING: IEEEtran.bst: No hyphenation pattern has been}%
\typeout{** loaded for the language `#1'. Using the pattern for}%
\typeout{** the default language instead.}%
\else
\language=\csname l@#1\endcsname
\fi
#2}}
\providecommand{\BIBdecl}{\relax}
\BIBdecl

\bibitem{Kalman_BSMM_1960}
R.~E. Kalman, ``Contributions to the theory of optimal control,''
  \emph{Boletinde la Sociedad Matematica Mexicana}, vol.~5, no.~1, pp.
  102--119, 1960.

\bibitem{Anderson_book_1990}
B.~D.~O. Anderson and J.~B. Moore, \emph{Optimal {C}ontrol: {L}inear
  {Q}uadratic {M}ethods}.\hskip 1em plus 0.5em minus 0.4em\relax Upper Saddle
  River,{ NJ}: Prentice-Hall, Inc., 1990.

\bibitem{Hewer1971TAC}
G.~{Hewer}, ``An iterative technique for the computation of the steady state
  gains for the discrete optimal regulator,'' \emph{IEEE Transactions on
  Automatic Control}, vol.~16, no.~4, pp. 382--384, 1971.

\bibitem{Lancaster1995algebraic}
P.~Lancaster and L.~Rodman, \emph{Algebraic {R}iccati {E}quations}.\hskip 1em
  plus 0.5em minus 0.4em\relax New York, {NY}: Oxford University Press, 1995.

\bibitem{Balakrishnan2003TAC}
V.~{Balakrishnan} and L.~{Vandenberghe}, ``Semidefinite programming duality and
  linear time-invariant systems,'' \emph{IEEE Transactions on Automatic
  Control}, vol.~48, no.~1, pp. 30--41, 2003.

\bibitem{Jiang2012Auto}
Y.~Jiang and Z.-P. Jiang, ``Computational adaptive optimal control for
  continuous-time linear systems with completely unknown dynamics,''
  \emph{Automatica}, vol.~48, no.~10, pp. 2699--2704, 2012.

\bibitem{Young2012Auto}
J.~Y. Lee, J.~B. Park, and Y.~H. Choi, ``Integral {Q}-learning and explorized
  policy iteration for adaptive optimal control of continuous-time linear
  systems,'' \emph{Automatica}, vol.~48, no.~11, pp. 2850--2859, 2012.

\bibitem{Lee2019TAC}
D.~{Lee} and J.~{Hu}, ``Primal-dual {Q}-learning framework for {LQR} design,''
  \emph{IEEE Transactions on Automatic Control}, pp. 1--1, 2018.

\bibitem{Bradtke1994ACC}
S.~J. {Bradtke}, B.~E. {Ydstie}, and A.~G. {Barto}, ``Adaptive linear quadratic
  control using policy iteration,'' in \emph{Proceedings of 1994 American
  Control Conference}, vol.~3, 1994, pp. 3475--3479.

\bibitem{Lewis2009CSM}
F.~L. {Lewis} and D.~{Vrabie}, ``Reinforcement learning and adaptive dynamic
  programming for feedback control,'' \emph{IEEE Circuits and Systems
  Magazine}, vol.~9, no.~3, pp. 32--50, 2009.

\bibitem{lewis2012reinforcement}
F.~L. Lewis, D.~Vrabie, and K.~G. Vamvoudakis, ``Reinforcement learning and
  feedback control: using natural decision methods to design optimal adaptive
  controllers,'' \emph{IEEE Control Systems}, vol.~32, no.~6, pp. 76--105,
  2012.

\bibitem{Chun2016IJC}
T.~Y. Chun, J.~Y. Lee, J.~B. Park, and Y.~H. Choi, ``Stability and monotone
  convergence of generalised policy iteration for discrete-time linear
  quadratic regulations,'' \emph{International Journal of Control}, vol.~89,
  no.~3, pp. 437--450, 2016.

\bibitem{wenk1980parameter}
C.~Wenk and C.~Knapp, ``Parameter optimization in linear systems with
  arbitrarily constrained controller structure,'' \emph{IEEE Transactions on
  Automatic Control}, vol.~25, no.~3, pp. 496--500, 1980.

\bibitem{jilg2013optimized}
M.~Jilg and O.~Stursberg, ``Optimized distributed control and topology design
  for hierarchically interconnected systems,'' in \emph{European Control
  Conference}, 2013, pp. 4340--4346.

\bibitem{maartensson2009gradient}
K.~M{\aa}rtensson and A.~Rantzer, ``Gradient methods for iterative distributed
  control synthesis,'' in \emph{Joint IEEE Conference on Decision and Control
  and Chinese Control Conference}, 2009, pp. 549--554.

\bibitem{fazel2018global}
M.~Fazel, R.~Ge, S.~Kakade, and M.~Mesbahi, ``Global convergence of policy
  gradient methods for the linear quadratic regulator,'' in \emph{Proceedings
  of the 35th International Conference on Machine Learning}, 2018, pp.
  1467--1476.

\bibitem{mart2012phd}
K.~M{\aa}rtensson, ``Gradient methods for large-scale and distributed linear
  quadratic control,'' Ph.D. dissertation, Department of Automatic Control,
  Lund University, Sweden, 2012.

\bibitem{polyak1963gradient}
B.~T. Polyak, ``Gradient methods for the minimisation of functionals,''
  \emph{USSR Computational Mathematics and Mathematical Physics}, vol.~3,
  no.~4, pp. 864--878, 1963.

\bibitem{sontag2013mathematical}
E.~D. Sontag, \emph{Mathematical Control Theory: Deterministic Finite
  Dimensional Systems}, 2nd~ed.\hskip 1em plus 0.5em minus 0.4em\relax New
  York, {NY}: Springer Science \& Business Media, 1998.

\bibitem{horn2012matrix}
R.~A. Horn and C.~R. Johnson, \emph{Matrix {A}nalysis}, 2nd~ed.\hskip 1em plus
  0.5em minus 0.4em\relax New York, {NY}: Cambridge University Press, 2012.

\bibitem{bu2019topological_mimo}
J.~Bu, A.~Mesbahi, and M.~Mesbahi, ``On topological and metrical properties of
  stabilizing feedback gains: the {MIMO} case,'' \emph{arXiv preprint
  arXiv:1904.02737}, 2019.

\bibitem{bauschke2017convex}
H.~H. Bauschke and P.~L. Combettes, \emph{Convex Analysis and Monotone Operator
  Theory in Hilbert Spaces}, 2nd~ed.\hskip 1em plus 0.5em minus 0.4em\relax
  Springer Science \& Business Media, 2017.

\bibitem{lojasiewicz1963propriete}
S.~Lojasiewicz, ``Une propri{\'e}t{\'e} topologique des sous-ensembles
  analytiques r{\'e}els,'' \emph{Les {\'e}quations aux d{\'e}riv{\'e}es
  partielles}, vol. 117, pp. 87--89, 1963.

\bibitem{helmke2012optimization}
U.~Helmke and J.~B. Moore, \emph{Optimization and Dynamical Systems}.\hskip 1em
  plus 0.5em minus 0.4em\relax London: Springer Science \& Business Media,
  1994.

\bibitem{mori1988comments}
T.~Mori, ``Comments on "{A} matrix inequality associated with bounds on
  solutions of algebraic {R}iccati and {L}yapunov equation" by {J.M.} {S}aniuk
  and {I.B.} {R}hodes,'' \emph{IEEE Transactions on Automatic Control},
  vol.~33, no.~11, p. 1088, 1988.

\bibitem{tu2017differential}
L.~W. Tu, \emph{Differential geometry: connections, curvature, and
  characteristic classes}.\hskip 1em plus 0.5em minus 0.4em\relax Springer,
  2017, vol. 275.

\bibitem{feng-lavei:2019}
H.~Feng and J.~Lavaei, ``On the exponential number of connected components for
  the feasible set of optimal decentralized control problems,'' in
  \emph{American Control Conference}, 2019.

\bibitem{nesterov2013introductory}
Y.~Nesterov, \emph{Introductory {L}ectures on {C}onvex {O}ptimization: A
  {B}asic {C}ourse}.\hskip 1em plus 0.5em minus 0.4em\relax Springer Science \&
  Business Media, 2004.

\end{thebibliography}
\end{document}